\documentclass[a4paper,UKenglish,cleveref, autoref,thm-restate]{lipics-v2019}
\usepackage{mathtools}
\usepackage{multicol}

\title{PBS-Calculus: A Graphical Language for Coherent Control of Quantum Computations}

\author{Alexandre Cl\'ement}{Universit\'e de Lorraine, CNRS, Inria, LORIA, F-54000 Nancy, France}{alexandre.clement@loria.fr}{https://orcid.org/0000-0002-7958-5712}{}

\author{Simon Perdrix}{Universit\'e de Lorraine, CNRS, Inria, LORIA, F-54000 Nancy, France}{simon.perdrix@loria.fr}{https://orcid.org/0000-0002-1808-2409}{}

\authorrunning{A. Cl\'ement and S. Perdrix} 

\Copyright{Alexandre Cl\'ement and Simon Perdrix} 

\ccsdesc[500]{Theory of computation~Quantum computation theory}
\ccsdesc[500]{Theory of computation~Axiomatic semantics}
\ccsdesc[500]{Theory of computation~Categorical semantics}
\ccsdesc[100]{Hardware~Quantum computation}
\ccsdesc[100]{Hardware~Quantum communication and cryptography}

\keywords{Quantum Computing, Diagrammatic Language, Completeness, Quantum Control, Polarising Beam Splitter, Categorical Quantum Mechanics, Quantum Switch.} 

\category{} 

\relatedversion{This is the full version of the paper published at MFCS'20 \cite{alex2020pbscalculus}.} 

\supplement{}

\funding{This work is funded by ANR-17-CE25-0009 SoftQPro, ANR-17-CE24-0035 VanQuTe, PIA-GDN/Quantex, and LUE/UOQ.}

\acknowledgements{The authors want to thank Mehdi Mhalla, Emmanuel Jeandel and Titouan Carette for fruitful discussions.
All diagrams were written with the help of TikZit.}

\nolinenumbers 

\hideLIPIcs

\EventEditors{Javier Esparza and Daniel Kr{\'a}l'}
\EventNoEds{2}
\EventLongTitle{45th International Symposium on Mathematical Foundations of Computer Science (MFCS 2020)}
\EventShortTitle{MFCS 2020}
\EventAcronym{MFCS}
\EventYear{2020}
\EventDate{August 24--28, 2020}
\EventLocation{Prague, Czech Republic}
\EventLogo{}
\SeriesVolume{170}
\ArticleNo{24}

\usepackage[svgnames,dvipsnames]{xcolor}
\usepackage{tikzit}

\tikzstyle{diamant}=[diamond, fill=couleurdefond, draw=black]
\tikzstyle{newe}=[rectangle, fill=gray!15, draw=black, tikzit shape=rectangle]
\tikzstyle{cercle}=[circle, fill=couleurdefond, draw=black]
\tikzstyle{cartouche}=[rounded rectangle, fill=couleurdefond, draw=black]
\tikzstyle{neg}=[rounded rectangle, fill=couleurdefond, draw=black, execute at end node={$\neg$}]
\tikzstyle{diagrammevide}=[rectangle, fill=couleurdefond, draw=black, inner sep=1.25em, borddiagrammevide, tikzit shape=rectangle]
\tikzstyle{mdiagrammevide}=[rectangle, fill=couleurdefond, draw=black, inner sep=0.75em, sborddiagrammevide, tikzit shape=rectangle]
\tikzstyle{sdiagrammevide}=[rectangle, fill=couleurdefond, draw=black, inner sep=0.5em, sborddiagrammevide, tikzit shape=rectangle]
\tikzstyle{xsdiagrammevide}=[rectangle, fill=couleurdefond, draw=black, inner sep=0.4em, xsborddiagrammevide, tikzit shape=rectangle]

\tikzstyle{bs}=[shape=beam, fill=couleurdefond, draw, inner sep=0.25em, thick, tikzit fill=white]
\tikzstyle{sbs}=[shape=beam, fill=couleurdefond, draw, inner sep=0.2em, thick, tikzit fill=white]
\tikzstyle{boite22}=[fill=white, draw=black, shape=rectangle, minimum height=1cm, minimum width=0.5cm]
\tikzstyle{boite2}=[fill=white, draw=black, shape=rectangle, minimum height=0cm, minimum width=0cm]
\tikzstyle{snegpotentiel}=[fill=couleurdefond, draw=black, shape=rounded rectangle, inner sep=0.25em, tikzit fill={rgb,255: red,191; green,191; blue,191}, execute at end node={\footnotesize$\star$}]
\tikzstyle{negpotentiel}=[fill=couleurdefond, draw=black, shape=rounded rectangle, tikzit fill={rgb,255: red,191; green,191; blue,191}, execute at end node={$\star$}]
\tikzstyle{token}=[fill=black, draw=black, shape=circle, inner sep=0.1em]

\tikzstyle{new}=[-]
\tikzstyle{tirets}=[-, draw=black, dashed]
\tikzstyle{noire}=[-, draw=black]
\tikzstyle{longdashed}=[-, dash pattern=on 5pt off 5pt]
\tikzstyle{pointilles}=[-, draw=black, dotted]
\tikzstyle{grise}=[-, draw={rgb,255: red,191; green,191; blue,191}]
\tikzstyle{borddiagrammevide}=[-, dash pattern=on 0.5em off 0.5em on 0.5em off 0.5em on 0.5em off 0em]
\tikzstyle{sborddiagrammevide}=[-, dash pattern=on 0.2em off 0.2em on 0.2em off 0.2em on 0.2em off 0em]
\tikzstyle{xsborddiagrammevide}=[-, dash pattern=on 0.1em off 0.1em on 0.15em off 0.1em on 0.1em off 0em]

\input{figures/styles-pbs.tikzdefs}

\hypersetup{hypertexnames=false,raiselinks=true}
\usepackage{stmaryrd}
\usepackage{scalerel}
\usepackage{longtable}

\newcolumntype{C}{>{$}c<{$}}  
\newcolumntype{R}{>{$}r<{$}}  
\newcolumntype{L}{>{$}l<{$}}  
\newcommand{\interp}[1]{\left\llbracket #1 \right\rrbracket}
\newcommand{\ket}[1]{\left| #1 \right\rangle}

 {\everymath{\displaystyle\everymath{}}\array}%
 {\endarray}
 {\everymath{\scriptstyle\everymath{}}\array}%
 {\endarray}
\newcommand\T{\mathcal T}
\renewcommand\H{\mathcal H} 
\renewcommand\S{\mathcal{SLP}}
\newcommand\N{\mathbb N}
\newcommand\C{\mathbb{C}}
\newcommand\hv{\{\rightarrow,\uparrow\}}

\newlength{\xlutmvcyp}

\makeatletter
\newlength{\negph@wd}
\DeclareRobustCommand{\negphantom}[1]{%
  \ifmmode
    \mathpalette\negph@math{#1}%
  \else
    \negph@do{#1}%
  \fi
}
\newcommand{\negph@math}[2]{\negph@do{$\m@th#1#2$}}
\newcommand{\negph@do}[1]{%
  \settowidth{\negph@wd}{#1}%
  \hspace*{-\negph@wd}%
}
\makeatother

\makeatletter
\newlength{\halfnegph@wd}
\DeclareRobustCommand{\halfnegphantom}[1]{%
  \ifmmode
    \mathpalette\halfnegph@math{#1}%
  \else
    \halfnegph@do{#1}%
  \fi
}
\newcommand{\halfnegph@math}[2]{\halfnegph@do{$\m@th#1#2$}}
\newcommand{\halfnegph@do}[1]{%
  \settowidth{\halfnegph@wd}{#1}%
  \hspace*{-0.5\halfnegph@wd}%
}
\makeatother

\makeatletter
\newlength{\halfph@wd}
\DeclareRobustCommand{\halfphantom}[1]{%
  \ifmmode
    \mathpalette\halfph@math{#1}%
  \else
    \halfph@do{#1}%
  \fi
}
\newcommand{\halfph@math}[2]{\halfph@do{$\m@th#1#2$}}
\newcommand{\halfph@do}[1]{%
  \settowidth{\halfph@wd}{#1}%
  \hspace*{0.5\halfph@wd}%
}
\makeatother

\newcommand\changelargeurcentre[2]{\settowidth{\xlutmvcyp}{$#2$}\text{\makebox[\xlutmvcyp]{$#1$}}}

\newcommand\uparrowlarge{\changelargeurcentre{\uparrow}{\rightarrow}}

\newcommand{\lllbracket}{[\![\![}
\newcommand{\rrrbracket}{]\!]\!]}
\newcommand{\varinterp}[1]{\lllbracket#1\rrrbracket}

\newcommand\Id{\mathit{Id}}
\newcommand\id{\mathit{id}}

\newcommand\eqeqref[1]{\overset{\eqref{#1}}{=}}
\newcommand\eer[1]{\eqeqref{#1}}
\newcommand\eqdeuxeqref[2]{\overset{\eqref{#1}\eqref{#2}}{=}}
\newcommand\eqtroiseqref[3]{\overset{\eqref{#1}\eqref{#2}\eqref{#3}}{=}}
\newcommand\eqquatreeqref[4]{\overset{\eqref{#1}\eqref{#2}\eqref{#3}\eqref{#4}}{=}}

\newcommand\eqexpl[1]{\overset{#1}{=}}

\newlength\traitsdiagrammevide
\newcommand\echellefils{0.35}

\newcommand{\labeletpreuve}[2]{\hypersetup{hidelinks}\label{#1}\hyperref[preuve#1]{#2}}

\newcommand{\noeqbreak}{\binoppenalty10000 \relpenalty10000}

\newcommand{\urlalt}[2]{\href{#2}{\nolinkurl{#1}}}

\theoremstyle{definition}

\begin{document}

\maketitle

\begin{abstract}
We introduce the PBS-calculus to represent and reason on quantum computations involving coherent control of quantum operations. Coherent control, and in particular indefinite causal order, is known to enable multiple computational and communication advantages over classically ordered models like quantum circuits. The PBS-calculus is inspired by quantum optics, in particular the polarising beam splitter (PBS for short). We formalise the syntax and the semantics of the PBS-diagrams, and we equip the language with an equational theory, which is proved to be sound and complete: two diagrams are representing the same quantum evolution if and only if one can be transformed into the other using the rules of the PBS-calculus. Moreover, we show that the equational theory is minimal. Finally, we consider  applications like the implementation of controlled permutations and the unrolling of loops.
\end{abstract}

\section{Introduction}

Quantum computers can solve problems which are out of reach of  classical computers \cite{Shor1997,Harrow2009}. One of the resources offered by quantum mechanics to speed up algorithms is the superposition phenomenon which allows a quantum memory to be in several possible classical states at the same time, in superposition. Less explored in quantum computing models, one can also consider a superposition of processes. Called \emph{coherent control} or simply \emph{quantum control}, it can be illustrated with the following example called quantum switch: the order in which two unitary evolutions $U$ and $V$ are applied is controlled by the state of a control qubit. In particular, if the control qubit is in superposition, then both $UV$ and $VU$ are applied, in superposition.

Coherent control is loosely represented in usual formalisms of quantum computing. For instance, in the quantum circuit model, the only available quantum control is the controlled gate mechanism: a gate $U$ is applied or not depending on the state of a control qubit. The quantum switch cannot be implemented with a single copy of $U$ and a single copy of $V$ in the quantum circuit model, and more generally using any language with a fixed or classically controlled order of operations. Quantum switch has however been realised experimentally \cite{procopio2015experimental,rubino2017experimental}. Moreover, such a quantum control has  been proved to enable various computational and communication advantages over classically ordered models \cite{araujo2014computational,facchini2015quantum,feix2015quantum,guerin2016exponential,abbott2018communication},  for instance for deciding whether two unitary transformations are commuting or anti-commuting \cite{chiribella2012perfect} (see  \cref{ex:commut}).

Notice that other models of quantum computations (e.g. Quantum Turing Machines) or programming languages (e.g. Lineal \cite{dowek2017lineal} or QML \cite{altenkirch2005functional}), allow for arbitrary coherent control of quantum evolutions, the price to pay is, however, the presence of non-trivial well-formedness conditions to ensure that the represented evolution is valid. Indeed, the superposition (i.e. linear combination) of two unitary evolutions is not necessarily a unitary evolution. 

We introduce a graphical language, the PBS-calculus, for representing coherent control of quantum computations, where arbitrary gates can be coherently controlled. Our goal is to provide the foundations of a formal framework which will be further developed to explore the power and limits of the coherent control of quantum evolutions. Contrary to the quantum circuit model, the PBS-calculus allows a representation of the quantum switch with a single copy of each gate to be controlled. Moreover, any PBS-diagram is valid by construction (no side nor well-formedness condition). The syntax of the PBS-diagrams is inspired by quantum optics and is actually already used in several papers dealing with coherent control of quantum evolutions \cite{abbott2018communication,araujo2014computational}. Our contribution is to provide formal syntax and semantics (both operational and denotational) for these diagrams, and also to introduce an equational theory which allows one to transform diagrams. Our main technical contribution is the proof that the equational theory is complete (if two diagrams have the same semantics then one can be transformed into the other using the equational theory) and minimal (in the sense that each of the equations is necessary for the completeness of the language).

The syntax of the PBS-calculus is inspired by linear optics, and in particular by the peculiar behaviour of the polarising beam splitter. A polarising beam splitter transforms a superposition of polarisations into a superposition of positions: if the polarisation is vertical the photon is transmitted whereas it is reflected when the polarisation is horizontal (see  \cref{figBS}.a). 
As a consequence a photon can be routed in different parts of a scheme, this routing being quantumly controlled by the polarisation of the photon.  This is a unique behaviour which has no counterpart in the quantum circuit model for instance. Polarising beam splitters can be used to perform a quantum switch, as depicted as a PBS-diagram in \cref{figBS}.b.

\begin{figure}
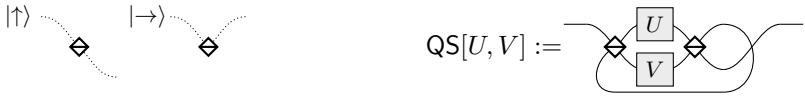

\input{pola.tikz} \qquad\qquad \qquad $\mathsf{QS}[U,V]:=$\!\!\input{qs.tikz}\caption{\label{figBS}(a) Intuitive behaviour of a polarising beam splitter: vertical polarisation goes through, horizontal polarisation is reflected; (b) Quantum switch of two matrices $U$ and $V$.}
\end{figure}

\subparagraph{Related works.} In the context of categorical quantum mechanics several graphical languages have already been introduced: ZX-calculus \cite{coecke2011interacting,jeandel2018complete}, ZW-calculus \cite{hadzihasanovic2017algebra}, ZH-calculus \cite{backens2018zh} and their variants. Notice in particular a proposal for representing fermionic (non polarising)  beam splitters in the ZW-calculus \cite{lmcs5736}.  
An apparent difference between the PBS-calculus and these languages, is that the category of PBS-diagrams is \emph{traced} but not \emph{compact closed}. This difference is probably not fundamental, as for any traced monoidal category there is a completion of it  to a compact closed category \cite{joyal1996traced}. 
The fundamental difference is the parallel composition: in the PBS-calculus two parallel wires correspond to two possible positions of a single particle (i.e. a direct sum in terms of semantics), whereas, in the other languages it corresponds to two particles (i.e. a tensor product).  

The parallel composition makes the PBS-calculus closer to the \emph{graphical linear algebra} approach \cite{bonchi2017interacting,bonchi2019graphical,bonchi2019diagrammatic}, however the generators and the fundamental structures (e.g. Frobenius algebra, Hopf algebra)  are \emph{a priori} unrelated to those of the PBS-calculus. 

In the context of quantum programming languages, there are a few proposals for representing quantum control \cite{dowek2017lineal,altenkirch2005functional,ying2014alternation,sabry2018symmetric}. Colnaghi et al. \cite{colnaghi2012quantum} have introduced a graphical language with \emph{programmable connections}. The language uses the quantum switch as a generator, but does not aim to describe schemes with polarising beam splitters. Notice also that the inputs/outputs of the language are quantum channels.

\subparagraph{Structure of the paper.} In Section \ref{syntax}, the syntax of the PBS-diagrams is introduced. The PBS-diagrams are considered up to a structural congruence which allows one to deform the diagrams at will. Section \ref{semantics} is dedicated to the semantics of the language: two semantics, a path semantics and a denotational semantics, are introduced. The denotational semantics is proved to be adequate with respect to the path semantics. In Section \ref{equationaltheory}, the axiomatisation of the PBS-calculus is introduced, and our main result, the soundness and completeness of the language, is proved. The axiomatisation is also proved to be minimal in the sense that none of the axioms can be derived from the others. Finally, in Section \ref{trace-free},  we consider the application of the PBS-calculus to the problem of loop unrolling. We show in particular that any PBS-diagram involving unitary matrices can be transformed into a trace-free diagram. The paper is written such that the reader does not need any particular knowledge in category theory. Basic definitions, in particular of Traced PROP, are however given in Appendix \ref{categoricalnotions} for completeness.

\section{Syntax}\label{syntax}

A $\textup{PBS}$-diagram is made of polarising beam splitters \input{beamsplitter-xs.tikz}, polarisation flips \begin{tikzpicture}[scale=0.8]
	\begin{pgfonlayer}{nodelayer}
		\node [style=cartouche] (0) at (0, 0) {\tiny$\neg$};
		\node [style=none] (1) at (-0.8, 0) {};
		\node [style=none] (2) at (0.8, 0) {};
	\end{pgfonlayer}
	\begin{pgfonlayer}{edgelayer}
		\draw [style=noire] (1.center) to (0);
		\draw [style=noire] (0) to (2.center);
	\end{pgfonlayer}
\end{tikzpicture}
, and gates \begin{tikzpicture}[scale=0.6]
	\begin{pgfonlayer}{nodelayer}
		\node [style=newe] (0) at (0, 0) {\tiny\!$U$\!};
		\node [style=none] (1) at (-1, 0) {};
		\node [style=none] (2) at (1, 0) {};
	\end{pgfonlayer}
	\begin{pgfonlayer}{edgelayer}
		\draw [style=new] (1.center) to (2.center);
	\end{pgfonlayer}
\end{tikzpicture} for any matrix $U\in \mathbb C^{q\times q}$, where $q$ is a fixed positive integer. One can also use wires like  the identity \begin{tikzpicture}
	\begin{pgfonlayer}{nodelayer}
		\node [style=none] (0) at (-0.3, 0) {};
		\node [style=none] (1) at (0.3, 0) {};
	\end{pgfonlayer}
	\begin{pgfonlayer}{edgelayer}
		\draw [style=noire] (0.center) to (1.center);
	\end{pgfonlayer}
\end{tikzpicture} or the swap \input{swap-xs.tikz}. Diagrams can  be combined by means of sequential composition $ \circ $, parallel composition $ \otimes $, and trace $Tr(\cdot)$. The trace consists in connecting the last output of a diagram to its last input, like a feedback loop. The symbol\begin{tikzpicture}[scale=0.3]
	\begin{pgfonlayer}{nodelayer}
		\node [style=xsdiagrammevide] (0) at (0, 0) {};
	\end{pgfonlayer}
\end{tikzpicture}
 represents the empty diagram. Any diagram has a type $n\to n$ which corresponds to the numbers of input/output wires.  The syntax of the language is the following:
 \begin{definition} 
Given $q\in \mathbb N\setminus\{0\}$, a $\textup{PBS}_q$-diagram $D:n\to n$ is inductively defined as:

\vspace{-0.5cm}

\[\begin{tikzpicture}[scale=0.3]
	\begin{pgfonlayer}{nodelayer}
		\node [style=sdiagrammevide] (0) at (0, 0) {};
	\end{pgfonlayer}
\end{tikzpicture}
:0\to 0\quad\qquad\begin{tikzpicture}
	\begin{pgfonlayer}{nodelayer}
		\node [style=none] (0) at (-0.3, 0) {};
		\node [style=none] (1) at (0.3, 0) {};
	\end{pgfonlayer}
	\begin{pgfonlayer}{edgelayer}
		\draw [style=noire] (0.center) to (1.center);
	\end{pgfonlayer}
\end{tikzpicture}:1\to 1\quad\qquad \begin{tikzpicture}[scale=0.8]
	\begin{pgfonlayer}{nodelayer}
		\node [style=cartouche] (0) at (0, 0) {\footnotesize$\neg$};
		\node [style=none] (1) at (-0.8, 0) {};
		\node [style=none] (2) at (0.8, 0) {};
	\end{pgfonlayer}
	\begin{pgfonlayer}{edgelayer}
		\draw [style=noire] (1.center) to (0);
		\draw [style=noire] (0) to (2.center);
	\end{pgfonlayer}
\end{tikzpicture}
:1\to 1\quad\qquad \input{swap-s.tikz}:2\to 2\qquad\quad \input{beamsplitter-s.tikz}:2\to 2\]

\vspace{-0.5cm}

\[\dfrac{U\in \mathbb C^{q\times q}}{\begin{tikzpicture}[scale=0.8]
	\begin{pgfonlayer}{nodelayer}
		\node [style=newe] (0) at (0, 0) {\footnotesize$U$};
		\node [style=none] (1) at (-1, 0) {};
		\node [style=none] (2) at (1, 0) {};
	\end{pgfonlayer}
	\begin{pgfonlayer}{edgelayer}
		\draw [style=new] (1.center) to (2.center);
	\end{pgfonlayer}
\end{tikzpicture}:1{\to}1}\qquad \dfrac{D_1:n{\to} n \quad D_2 : n{\to} n}{D_2\circ D_1:n{\to} n}\qquad\dfrac{D_1:n{\to} n \quad D_2 : m{\to} m}{D_1\otimes D_2:n{+}m\to n{+}m}\qquad \dfrac{D:n{+}1\to n{+}1}{Tr(D):n\to n}\]
\end{definition}
Sequential composition $D_2\circ D_1$, parallel composition $ D_1\otimes D_2$, and trace $Tr(D)$ are respectively  depicted as follows:\\ \centerline{$\input{composition-s.tikz}\qquad\qquad\input{produittensoriel-s.tikz}\qquad\qquad\input{trace-s.tikz}$}

\vspace{0.2cm}

In the following, the positive integer $q$  will be omitted when it is useless or clear from the context.

Notice that two distinct terms, like $\begin{tikzpicture}[scale=0.8]
	\begin{pgfonlayer}{nodelayer}
		\node [style=cartouche] (0) at (0, 0) {\footnotesize$\neg$};
		\node [style=none] (1) at (-0.8, 0) {};
		\node [style=none] (2) at (0.8, 0) {};
	\end{pgfonlayer}
	\begin{pgfonlayer}{edgelayer}
		\draw [style=noire] (1.center) to (0);
		\draw [style=noire] (0) to (2.center);
	\end{pgfonlayer}
\end{tikzpicture}
 \circ (\begin{tikzpicture}[scale=0.8]
	\begin{pgfonlayer}{nodelayer}
		\node [style=newe] (0) at (0, 0) {\footnotesize$U$};
		\node [style=none] (1) at (-1, 0) {};
		\node [style=none] (2) at (1, 0) {};
	\end{pgfonlayer}
	\begin{pgfonlayer}{edgelayer}
		\draw [style=new] (1.center) to (2.center);
	\end{pgfonlayer}
\end{tikzpicture}\circ \begin{tikzpicture}[scale=0.8]
	\begin{pgfonlayer}{nodelayer}
		\node [style=cartouche] (0) at (0, 0) {\footnotesize$\neg$};
		\node [style=none] (1) at (-0.8, 0) {};
		\node [style=none] (2) at (0.8, 0) {};
	\end{pgfonlayer}
	\begin{pgfonlayer}{edgelayer}
		\draw [style=noire] (1.center) to (0);
		\draw [style=noire] (0) to (2.center);
	\end{pgfonlayer}
\end{tikzpicture}
)$   and $(\begin{tikzpicture}[scale=0.8]
	\begin{pgfonlayer}{nodelayer}
		\node [style=cartouche] (0) at (0, 0) {\footnotesize$\neg$};
		\node [style=none] (1) at (-0.8, 0) {};
		\node [style=none] (2) at (0.8, 0) {};
	\end{pgfonlayer}
	\begin{pgfonlayer}{edgelayer}
		\draw [style=noire] (1.center) to (0);
		\draw [style=noire] (0) to (2.center);
	\end{pgfonlayer}
\end{tikzpicture}
\circ \begin{tikzpicture}[scale=0.8]
	\begin{pgfonlayer}{nodelayer}
		\node [style=newe] (0) at (0, 0) {\footnotesize$U$};
		\node [style=none] (1) at (-1, 0) {};
		\node [style=none] (2) at (1, 0) {};
	\end{pgfonlayer}
	\begin{pgfonlayer}{edgelayer}
		\draw [style=new] (1.center) to (2.center);
	\end{pgfonlayer}
\end{tikzpicture})\circ \begin{tikzpicture}[scale=0.8]
	\begin{pgfonlayer}{nodelayer}
		\node [style=cartouche] (0) at (0, 0) {\footnotesize$\neg$};
		\node [style=none] (1) at (-0.8, 0) {};
		\node [style=none] (2) at (0.8, 0) {};
	\end{pgfonlayer}
	\begin{pgfonlayer}{edgelayer}
		\draw [style=noire] (1.center) to (0);
		\draw [style=noire] (0) to (2.center);
	\end{pgfonlayer}
\end{tikzpicture}
$, can lead to the same graphical representation: $\begin{tikzpicture}[scale=0.8]
	\begin{pgfonlayer}{nodelayer}
		\node [style=cartouche] (0) at (0, 0) {\footnotesize$\neg$};
		\node [style=none] (1) at (-0.8, 0) {};
		\node [style=none] (2) at (0.8, 0) {};
	\end{pgfonlayer}
	\begin{pgfonlayer}{edgelayer}
		\draw [style=noire] (1.center) to (0);
		\draw [style=noire] (0) to (2.center);
	\end{pgfonlayer}
\end{tikzpicture}
\!\begin{tikzpicture}[scale=0.8]
	\begin{pgfonlayer}{nodelayer}
		\node [style=newe] (0) at (0, 0) {\footnotesize$U$};
		\node [style=none] (1) at (-1, 0) {};
		\node [style=none] (2) at (1, 0) {};
	\end{pgfonlayer}
	\begin{pgfonlayer}{edgelayer}
		\draw [style=new] (1.center) to (2.center);
	\end{pgfonlayer}
\end{tikzpicture}\!\begin{tikzpicture}[scale=0.8]
	\begin{pgfonlayer}{nodelayer}
		\node [style=cartouche] (0) at (0, 0) {\footnotesize$\neg$};
		\node [style=none] (1) at (-0.8, 0) {};
		\node [style=none] (2) at (0.8, 0) {};
	\end{pgfonlayer}
	\begin{pgfonlayer}{edgelayer}
		\draw [style=noire] (1.center) to (0);
		\draw [style=noire] (0) to (2.center);
	\end{pgfonlayer}
\end{tikzpicture}
$.  To avoid ambiguity, we define diagrams modulo the structural congruence given in Figure \ref{fig:TracedProp} in \cref{categoricalnotions}. Roughly speaking the structural congruence guarantees that (\emph{i}) two terms leading to the same graphical representation are equivalent, and (\emph{ii}) a diagram can be deformed at will, e.g.:
{\tikzset{tikzfig/.style={baseline=-0.25em,scale=\echellefils,every node/.style={scale=0.7}}}
\[
\!\!\!\!\input{swapswap-xs.tikz}=\input{filsparalleleslongs-xs.tikz}\qquad \input{natswap1.tikz}=\input{natswap2.tikz}\qquad \input{yankingvariantecentresurfil-xs.tikz}=\begin{tikzpicture}[scale=0.6,xscale=0.6]
	\begin{pgfonlayer}{nodelayer}
		\node [style=none] (0) at (-2, 0) {};
		\node [style=none] (1) at (2, 0) {};
	\end{pgfonlayer}
	\begin{pgfonlayer}{edgelayer}
		\draw [style=new] (0.center) to (1.center);
	\end{pgfonlayer}
\end{tikzpicture}\qquad \input{tracegrandfgb-xs.tikz}=\input{tracegbgrandf-xs.tikz}\]}
\\\indent In the categorical framework of PROP \cite{MacLane1965,zanasi2015tel}, PBS-diagrams modulo the structural congruence form a Traced PROP, i.e. they are morphisms of a traced strict symmetric monoidal category whose objects are natural numbers. 
It is known (Theorem 20\footnote{Notice that in \cite{Sel2009-graphical}, the author points out that this result relies on a result by Kelly and Laplaza (Theorem. 8.2, \cite{kelly1980coherence}) which is only proven for simple signatures -- which is not the case for the PBS-diagrams. The general case does not appear in the literature.} of \cite{Sel2009-graphical}) that two diagrams are equivalent according to the axioms of a traced PROP if and only if they are isomorphic in a graph-theoretical sense, that is, if one can be obtained from the other by moving, stretching and reorganising the wires in any way, while keeping their two ends fixed.

\section{Semantics}\label{semantics}

In this section, we introduce the semantics of the PBS-diagrams. First, we introduce an operational semantics for PBS-diagrams with a classical control. The operational semantics, called \emph{path semantics} is based on the graphical intuition of a routed particle. Then we introduce a denotational semantics for the general case, with a quantum control. We show the adequacy between the two semantics, providing a graphical way to compute the denotational semantics of a PBS-diagram. 

In this paper, we only consider the case where a \emph{single} particle, say a photon, is present in the diagram. The particle is made of a polarisation and an additional data register.
 The particle has: an initial polarisation, which is an arbitrary superposition of the horizontal ($\rightarrow$) and vertical ($\uparrow$) polarisations (that we call \emph{classical} polarisations in the following); an arbitrary position, which is a  superposition of the possible input wires of the diagram;  and an input data state, which is a vector $\varphi\in \mathbb C^q$.

\subsection{Classical control -- Path semantics}

\subparagraph{Classical control.} We first consider input particles with a classical polarisation and a classical position. Roughly speaking, the particle is initially located on one of the input wires with a given polarisation in $\hv$, and moves through the diagram depending on its polarisation. 
The action of a PBS-diagram can be \emph{informally} described as follows using a token made of the current polarisation $c$ of the particle and a matrix $U$ representing the matrix applied so far to the data register: 
{\tikzset{tikzfig/.style={baseline=-0.25em,scale=1.1*\echellefils,every node/.style={scale=0.9}}}
\begin{itemize}
\item The particle is either reflected or transmitted by a beam splitter, depending on its polarisation: 

\vspace{-0.2cm}

\begin{longtable}{RCLCRCL}
\input{tokenhorUhbs.tikz}&\rightarrow&\input{bstokenhorUh.tikz}&\qquad\qquad&\input{tokenhorUbbs.tikz}&\rightarrow&\input{bstokenhorUb.tikz}\\\\[-0.2cm]
\input{tokenverUhbs.tikz}&\rightarrow&\input{bstokenverUb.tikz}&&\input{tokenverUbbs.tikz}&\rightarrow&\input{bstokenverUh.tikz}
\end{longtable}

\vspace{-0.2cm}

\item The polarisation may vary but remains classical (that is, in $\hv$) as the polarisation flip --  the only one which acts on the polarisation --  interchanges horizontal and vertical polarisations: 

\vspace{-0.3cm}

 \begin{longtable}{RCLCRCL}\input{tokencUneg-H.tikz}&\rightarrow&\input{negtokencbarreU-V.tikz}&\qquad\qquad&\ \input{tokencUneg-V.tikz}&\rightarrow&\input{negtokencbarreU-H.tikz}
 \end{longtable}
 
 \vspace{-0.3cm}
 
\item \begin{tikzpicture}
	\begin{pgfonlayer}{nodelayer}
		\node [style=newe] (0) at (0, 0) {$V$};
		\node [style=none] (1) at (-1.5, 0) {};
		\node [style=none] (2) at (1.5, 0) {};
	\end{pgfonlayer}
	\begin{pgfonlayer}{edgelayer}
		\draw [style=new] (1.center) to (2.center);
	\end{pgfonlayer}
\end{tikzpicture}
 acts on the data register, transforming the state $\varphi$ into $V\varphi$:

\vspace{-0.3cm}

\begin{longtable}{RCL}\input{tokencUgateV.tikz}&\rightarrow&\input{gateVtokencVU.tikz} \end{longtable}

\vspace{-0.3cm}

\item The particle can freely move through wires, e.g.:

\vspace{-0.3cm}

 \begin{longtable}{RCLCRCL}
\input{tokencUhx.tikz}&\rightarrow&\input{xtokencUb.tikz}&\qquad\qquad&
\input{tracetokencUdroite.tikz}&\rightarrow&\input{tracetokencUgauche.tikz}
 \end{longtable}

\vspace{-0.3cm}

\end{itemize}}

Thus the token follows a path from the input to the output and accumulates a matrix along the path. We formalise this intuitive behaviour as a big-step operational semantics that we call \emph{path semantics} in this context.  
A \emph{configuration} is a triplet $(D,c,p)$, where $D:n\to n$ is a PBS-diagram, $c\in \hv$ is the input polarisation of the particle, and $p\in [n] := \{0,\ldots, n-1\}$ its input position: $0$ means that the particle is located on the  first upper input wire, $1$ on the second one and so on. The result is made of the final polarisation $c'$ and position $p'$, and of the matrix $U$ representing the overall action of $D$ on the data register.

\begin{definition}[Path semantics] Given a PBS-diagram $D:n\to n$, a polarisation $c\in \hv$ and a position $p\in [n]$, let $(D,c,p)\xRightarrow{U}(c',p')$ (or simply $(D,c,p)\Rightarrow(c',p')$ when $U$ is the identity) be inductively defined as follows: 

  \begin{longtable}{CCCCCCC}
 \left(\begin{tikzpicture}
	\begin{pgfonlayer}{nodelayer}
		\node [style=none] (0) at (-0.3, 0) {};
		\node [style=none] (1) at (0.3, 0) {};
	\end{pgfonlayer}
	\begin{pgfonlayer}{edgelayer}
		\draw [style=noire] (0.center) to (1.center);
	\end{pgfonlayer}
\end{tikzpicture},c,0\right)\Rightarrow(c,0)&\,&\left(\begin{tikzpicture}[scale=0.8]
	\begin{pgfonlayer}{nodelayer}
		\node [style=cartouche] (0) at (0, 0) {\footnotesize$\neg$};
		\node [style=none] (1) at (-0.8, 0) {};
		\node [style=none] (2) at (0.8, 0) {};
	\end{pgfonlayer}
	\begin{pgfonlayer}{edgelayer}
		\draw [style=noire] (1.center) to (0);
		\draw [style=noire] (0) to (2.center);
	\end{pgfonlayer}
\end{tikzpicture}
,\uparrow,0\right)\Rightarrow(\rightarrow,0)&\,&\left(\begin{tikzpicture}[scale=0.8]
	\begin{pgfonlayer}{nodelayer}
		\node [style=cartouche] (0) at (0, 0) {\footnotesize$\neg$};
		\node [style=none] (1) at (-0.8, 0) {};
		\node [style=none] (2) at (0.8, 0) {};
	\end{pgfonlayer}
	\begin{pgfonlayer}{edgelayer}
		\draw [style=noire] (1.center) to (0);
		\draw [style=noire] (0) to (2.center);
	\end{pgfonlayer}
\end{tikzpicture}
,\rightarrow,0\right)\Rightarrow(\uparrow,0)&\,&\left(\begin{tikzpicture}[scale=0.8]
	\begin{pgfonlayer}{nodelayer}
		\node [style=newe] (0) at (0, 0) {\footnotesize$U$};
		\node [style=none] (1) at (-1, 0) {};
		\node [style=none] (2) at (1, 0) {};
	\end{pgfonlayer}
	\begin{pgfonlayer}{edgelayer}
		\draw [style=new] (1.center) to (2.center);
	\end{pgfonlayer}
\end{tikzpicture},c,0\right)\xRightarrow{U}(c,0)
  \end{longtable}
  \vspace{-0.3cm}
 \begin{longtable}{R@{}C@{}LCCCC}
  \left(\input{swap-s.tikz},\,c,\,p\right)&\Rightarrow&(c,1-p)&& \dfrac{(D_1,c,p)\xRightarrow{U}(c',p')\qquad(D_2,c',p')\xRightarrow{V}(c'',p'')}{(D_2\circ D_1,c,p)\xRightarrow{VU}(c'',p'')}(\circ)\\\\
 
 \left(\input{beamsplitter-s.tikz},\rightarrow,p\right)&\Rightarrow&(\rightarrow,p)  &&\dfrac{D_1:n\to n\qquad p<n\qquad (D_1,c,p)\xRightarrow{U}(c',p')}{(D_1\otimes D_2,c,p)\xRightarrow{U}(c',p')}(\otimes1)\\\\
 
\left(\input{beamsplitter-s.tikz},\,\uparrow\,,p\right)&\Rightarrow&(\uparrow,1-p)&&\dfrac{D_1:n\to n\qquad p\geq n\qquad (D_2,c,p-n)\xRightarrow{U}(c',p')}{(D_1\otimes D_2,c,p)\xRightarrow{U}(c',p'+n)}(\otimes2)\\\\
 
&&&&\hspace{-3cm}\dfrac{D:n\to n\qquad \forall i\in\{0,\ldots,k\}, (D,c_i,p_i)\xRightarrow{U_i\,}(c_{i+1},p_{i+1})\qquad(p_{i+1}=n){\Leftrightarrow}(i<k)}{(Tr(D),c_0,p_0)\xRightarrow{U_k\cdots U_0\,}(c_{k+1},p_{k+1})}(\mathsf T_k)~~~~~~~~\\
  \end{longtable}
  
  \vspace{-0.3cm}
  
  \noindent with $k\in \{0,1,2\}$. 
  \end{definition}

  \begin{example} \label{exQSPath} As expected, the path semantics of the quantum switch $\mathsf{QS}[U,V]:=\linebreak Tr\left(\input{swap-xs.tikz}\circ\input{beamsplitter-xs.tikz}\circ(\begin{tikzpicture}[scale=0.6]
	\begin{pgfonlayer}{nodelayer}
		\node [style=newe] (0) at (0, 0) {\tiny\!$U$\!};
		\node [style=none] (1) at (-1, 0) {};
		\node [style=none] (2) at (1, 0) {};
	\end{pgfonlayer}
	\begin{pgfonlayer}{edgelayer}
		\draw [style=new] (1.center) to (2.center);
	\end{pgfonlayer}
\end{tikzpicture}\otimes\begin{tikzpicture}[scale=0.6]
	\begin{pgfonlayer}{nodelayer}
		\node [style=newe] (0) at (0, 0) {\tiny\!$V$\!};
		\node [style=none] (1) at (-1, 0) {};
		\node [style=none] (2) at (1, 0) {};
	\end{pgfonlayer}
	\begin{pgfonlayer}{edgelayer}
		\draw [style=new] (1.center) to (2.center);
	\end{pgfonlayer}
\end{tikzpicture}
)\circ\input{beamsplitter-xs.tikz}\right)$ (see \cref{figBS}.b) is $(\mathsf{QS}[U,V],\uparrow,0) \xRightarrow{UV}(\uparrow,0)$ and $(\mathsf{QS}[U,V],\rightarrow,0) \xRightarrow{VU}(\rightarrow,0)$.
\end{example}

\begin{example}
PBS-diagrams implementing a controlled permutation are given in Figure \ref{fig:perm3}. 
\end{example}

\begin{figure}[ht]
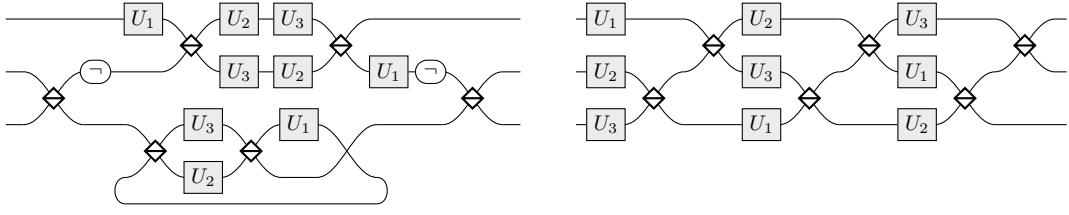

\[\input{perm3var.tikz}\qquad\input{perm3.tikz}\]
\caption{Two diagrams having the same semantics, that implement a controlled permutation of 3 unitary maps. Given a permutation $(xyz)$ of $(123)$, we have $(D,c,x)\xRightarrow{U_{z}U_{y}U_{x}}(c,x)$, where $D$ is any of the two diagrams and $c=\rightarrow$ if the signature of $(xyz)$ is $1$, $c=\uparrow$ otherwise. A generalisation to  the controlled permutation of $n$ unitary maps is given in Appendix \ref{diagrammecontrolepermutations}.\label{fig:perm3}}
\end{figure}

Notice that the path semantics does not need to be defined for the empty diagram $\begin{tikzpicture}[scale=0.3]
	\begin{pgfonlayer}{nodelayer}
		\node [style=sdiagrammevide] (0) at (0, 0) {};
	\end{pgfonlayer}
\end{tikzpicture}
$. Indeed, for any diagram $D:0\to0$ there is no valid configuration $(D,c,p)$ as $p$ should be one of the input wires of $D$.

    The $(\mathsf T_k)$-rule is parametrised by an integer $k$. Intuitively, this parameter is the number of times the photon goes through the corresponding trace. We show in the following that roughly speaking, a particle can never go through a given trace more than twice. In other words, the path semantics which assumes $k\le 2$, is well-defined for any valid configuration:

\begin{proposition}\label{determinismterm}
For any diagram $D:n\to n$ and any $(c,p)\in\hv\times[n]$, there exist unique $(c',p')\in\hv\times[n]$ and $U\in\C^{q\times q}$ such that $(D,c,p)\xRightarrow{U}(c',p')$.
\end{proposition}

\begin{proof}\vspace{-0.16cm}The proof is given in appendix, section \ref{preuvedeterminismtermreversibility}. \end{proof}

In the previous proposition, uniqueness means that the path semantics is deterministic: since diagrams are considered modulo structural congruence (i.e. up to deformation), it implies that these deformations preserve the path semantics.  

Moreover, all PBS-diagrams are invertible in the following sense:

\begin{proposition}\label{reversibility}
For any diagram $D:n\to n$ and any $(c,p)\in\hv\times[n]$, there exist unique $(c',p')\in\hv\times[n]$ and $U\in\C^{q\times q}$ such that $(D,c',p')\xRightarrow{U}(c,p)$.
\end{proposition}

\begin{proof}\vspace{-0.16cm}The proof is given in appendix, Section \ref{preuvedeterminismtermreversibility}. \end{proof}

As a consequence, any diagram $D:n\to n$ essentially acts as a permutation on $\hv\times [n]$,  if one ignores its action on the data register. We introduce  dedicated notations for representing the corresponding permutation, as well as the actions on the data register: 

\begin{definition}
For any diagram $D:n\to n$, we call $\tau_D$ the permutation of $\hv\times[n]$ and for any $ c,p \in \hv\times[n]$, we call $[D]_{c,p}\in \mathbb C^{q\times q}$ the matrix such that $(D,c,p)\xRightarrow{[D]_{c,p}} \tau_D(c,p)$. 
\end{definition}

In a PBS-diagram, the particle can go through each wire at most twice, otherwise, roughly speaking,  it would go back to the same position with the same polarisation and thus will come back again and again to this same configuration and thus enter an infinite loop -- which is prevented by Proposition \ref{determinismterm}. In particular, each gate of the diagram is visited at most twice: 

\begin{proposition}
Any gate $U$ of a diagram $D$ contributes to at most two paths $[D]_{c_0,p_0}$ and $[D]_{c_1,p_1}$, i.e.~given $D'$ the diagram $D$ where one occurrence of $U$ has been replaced by an arbitrary matrix $V$, $\forall (c,p)\notin \{(p_0,c_0),(p_1,c_1)\}$, $[D]_{c,p} = [D']_{c,p}$. 
\end{proposition}
\begin{proof}\vspace{-0.16cm}
The proof is straightforward by induction on $D$.
\end{proof}

As a consequence the diagrams of Figure \ref{fig:perm3} are optimal in the number of uses of each $U_i$: since each of the 6 paths must depend on each $U_i$, at least three copies of each $U_i$ are required in a diagram which solves the permutation problem of $3$ unitaries.

\subsection{Quantum control -- Denotational semantics}\label{definterp}

A crucial property of PBS-diagram is to offer the ability to have a quantum control, i.e. a particle whose input state is a superposition of polarisations, positions, or both. To encounter the quantum control, we introduce in this section a denotational semantics which associates with any diagram a map acting on the state space $\mathcal H_n:= \mathbb C^{\{\rightarrow, \uparrow\}}\otimes \mathbb C^n\otimes \mathbb C^q$. Using Dirac notations,  $\{\ket{\rightarrow}, \ket{\uparrow}\}$ (resp. $\{\ket {x} ~|~ x\in \{0\ldots k-1\}\}$) is an orthonormal basis of $\mathbb C^{\hv}$ (resp. $\mathbb C^k$).  Thus  $\{\ket {c,p,x} ~|~ c\in \hv, p\in [n], x\in [q]\}$ is an orthonormal basis of $\mathcal H_n$.

\begin{definition}
The denotational semantics of a \textup{PBS}-diagram $D:n\to n$ is the linear map $\interp D : \mathcal H_n \to \mathcal H_n$  inductively defined as follows:
 \begin{longtable}{R@{}C@{}LCR@{}C@{}L}
\interp{~\begin{tikzpicture}[scale=0.3]
	\begin{pgfonlayer}{nodelayer}
		\node [style=sdiagrammevide] (0) at (0, 0) {};
	\end{pgfonlayer}
\end{tikzpicture}
~}&{~=~}&0 && \interp{\,\begin{tikzpicture}
	\begin{pgfonlayer}{nodelayer}
		\node [style=none] (0) at (-0.3, 0) {};
		\node [style=none] (1) at (0.3, 0) {};
	\end{pgfonlayer}
	\begin{pgfonlayer}{edgelayer}
		\draw [style=noire] (0.center) to (1.center);
	\end{pgfonlayer}
\end{tikzpicture}\,} &{~=~} & \ket{c,0,x}\mapsto \ket{c,0,x}\\\\[-0.2cm]
\interp {~\input{swap-s.tikz}~} &{=}& \ket {c,p,x}\mapsto \ket{c,1-p,x}&& \interp{\,\begin{tikzpicture}[scale=0.8]
	\begin{pgfonlayer}{nodelayer}
		\node [style=newe] (0) at (0, 0) {\footnotesize$U$};
		\node [style=none] (1) at (-1, 0) {};
		\node [style=none] (2) at (1, 0) {};
	\end{pgfonlayer}
	\begin{pgfonlayer}{edgelayer}
		\draw [style=new] (1.center) to (2.center);
	\end{pgfonlayer}
\end{tikzpicture}\,} &=&\ket {c,0,x} \mapsto  \ket {c,0}\otimes U\ket x\\\\[-0.2cm]
\interp{\,\begin{tikzpicture}[scale=0.8]
	\begin{pgfonlayer}{nodelayer}
		\node [style=cartouche] (0) at (0, 0) {\footnotesize$\neg$};
		\node [style=none] (1) at (-0.8, 0) {};
		\node [style=none] (2) at (0.8, 0) {};
	\end{pgfonlayer}
	\begin{pgfonlayer}{edgelayer}
		\draw [style=noire] (1.center) to (0);
		\draw [style=noire] (0) to (2.center);
	\end{pgfonlayer}
\end{tikzpicture}
\,} &{=}&\begin{cases}\ket {\rightarrow,0,x} \mapsto  \ket {\,\uparrow\,,0,x} &\\
    \ket {\,\,\uparrow\,,0,x} \mapsto  \ket {\rightarrow,0,x}&\end{cases} && \interp{~\input{beamsplitter-s.tikz}~} &=&\begin{cases}\ket {\rightarrow,p,x} \mapsto  \ket {\rightarrow,p,x} &\\
    \ket {\,\,\uparrow\,,p,x} \mapsto  \ket {\,\,\uparrow\,,1-p,x}&\end{cases}\\\\[-0.2cm]
    \multicolumn{7}{C}{\interp{D_2\circ D_1} ~=~ \interp{D_2}\circ \interp{D_1} \quad\qquad \interp{D_1\otimes D_2} ~=~ \interp{D_1}\boxplus\interp{D_2} \quad\qquad \interp{Tr(D)} ~=~ \T(\interp D)}
  \end{longtable}
\vspace{-0.2cm}\noindent where:
\begin{itemize}
\item $f\boxplus g := \varphi \circ (f\oplus g)\circ \varphi^{-1}$ with $\varphi\colon\mathcal H_n\oplus \mathcal H_{m} \to \mathcal H_{n+m}$  the isomorphism defined as $(\ket{c,p,x},\ket{c',p',x'})\mapsto \ket {c,p,x}+\ket {c',p'+n,x'}$.

\item  $\T(f):=\displaystyle\sum_{k\in \mathbb N}\pi_1 \circ (f \circ \pi_0)^k \circ  f \circ \iota$ with $\iota:\mathcal H_n {\to} \mathcal H_{n+1}::\ket{c,x,y} \mapsto \ket{c,x,y}$, $\pi_0:\mathcal H_{n+1} {\to} \mathcal H_{n+1}::\ket{c,x,y} \mapsto \begin{cases} 0 &\text{if $x<n$}\\\ket{c,n,y}&\text{if $x = n$}\end{cases}$, and $\pi_1:\mathcal H_{n+1} {\to} \mathcal H_{n}::\ket{c,x,y} \mapsto \begin{cases} \ket{c,x,y}&\text{if $x<n$}\\0&\text{if $x = n$.}\end{cases}$. 
\end{itemize}
\end{definition}

Notice that while the semantics of the trace is defined by means of an infinite sum, this sum is actually made of a finite number of nonzero elements, which guarantees that the denotational semantics is well-defined:

\begin{proposition}\label{welldefinednessofdenotationalsemantics} For any diagram $D:n\to n$, $\interp D \in \S_n$, where $\S_n$ is the monoid of the linear maps $f\colon \mathcal H_n \to \mathcal H_n$ such that $f\ket{c,p,x} = \ket{\tau (c,p)}\otimes U_{c,p}\ket x$ for some permutation $\tau$ on $\hv \times [n]$ and matrices $U_{c,p}\in \mathbb C^{q\times q}$. \end{proposition}
\begin{proof}\vspace{-0.16cm}The proof is given in appendix, Section \ref{proofad}. \end{proof}

The denotational semantics is adequate with respect to the path  semantics:

\begin{theorem}[Adequacy]\label{adequacy}
For any $D:n\to n$, $\interp{D}=\ket{c,p,x}\mapsto\ket{\tau_D(c,p)}\otimes [D]_{c,p}\ket{x}$, \\where $\tau_D$  and $[D]_{c,p}$ are such that $(D,c,p)\xRightarrow{[D]_{c,p}} \tau_D(c,p)$
\end{theorem}
\begin{proof}\vspace{-0.16cm}The proof is given in appendix, Section \ref{proofad}. \end{proof} 

The adequacy theorem implies that two diagrams have the same denotational semantics if and only if they have the same path semantics. As a consequence, it provides a graphical characterisation of the denotational semantics. Indeed, for any diagram $D:n\to n$, $\interp D$ is, by linearity,  entirely defined by $\tau_D$ and $\{[D]_{c,p}\}_{c\in \hv, p\in [n]}$. Since $\tau_D$ and $[D]_{c,p}$ have a nice graphical interpretation as paths from the inputs to the outputs, the adequacy theorem provides a graphical way to compute the denotational semantics of any PBS-diagram.

\begin{example}\label{ex:commut}
The quantum switch (\cref{figBS}.b and \cref{exQSPath}) can be used to decide whether $U$ and $V$ are commuting or anti-commuting \cite{chiribella2012perfect}. The semantics of the quantum switch is $\interp{\mathsf{QS}[U,V]} = \begin{cases}\ket{\rightarrow,0,x} \mapsto  \ket{\rightarrow,0} \otimes VU \ket x\\ \ket{\uparrow,0,x} \mapsto  \ket{\uparrow,0} \otimes UV \ket x\end{cases}$. We assume that $UV=(-1)^k VU$ and call the quantum switch with a control qubit in a uniform superposition: $\interp{\mathsf{QS}[U,V]}\frac {\ket{\rightarrow} + \ket{\uparrow}}{\sqrt 2}\otimes\ket {0,x}=\frac {\ket{\rightarrow,0}\otimes VU\ket x + \ket{\uparrow,0}\otimes UV\ket x}{\sqrt 2}=  \frac {\ket{\rightarrow,0}\otimes VU\ket x + (-1)^k\ket{\uparrow,0}\otimes VU\ket x}{\sqrt 2} =$ \\$ 
\frac {\ket{\rightarrow} + (-1)^k \ket{\uparrow}}{\sqrt 2}\otimes VU\ket {0,x}$. Thus, by measuring the control qubit in the  $\{\frac {\ket{\rightarrow} + \ket{\uparrow}}{\sqrt 2},\frac {\ket{\rightarrow} - \ket{\uparrow}}{\sqrt 2}\}$-basis, one can decide whether $U$ and $V$ are commuting or anti-commuting. 
\end{example}

\section{Equational theory -- PBS-calculus }\label{equationaltheory}

The representation of a quantum computation using PBS-diagrams is not unique, in the sense that two distinct PBS-diagrams may have the same semantics (e.g.~diagrams of Figure \ref{fig:perm3}). In this section, we introduce 10 equations on PBS-diagrams (see Figure \ref{axioms}) as the axioms of a language that we call the PBS-calculus. We prove that the PBS-calculus is sound (that is, consistent with the semantics), complete (that is, it captures entirely the semantic equivalence) and minimal (that is, all axioms are necessary to have completeness). Completeness is proved by means of a normal form.

\subsection{Axiomatisation}\label{axiomatisation}

\begin{definition}[PBS-calculus]Two \textup{PBS}-diagrams $D_1, D_2$ are equivalent according to the rules of the $\textup{PBS}$-calculus, denoted $\textup{PBS}\vdash D_1=D_2$,  if one can transform $D_1$ into $D_2$ using the equations given in Figure \ref{axioms}. 
More precisely, $\textup{PBS}\vdash \cdot = \cdot$ is defined as the smallest congruence\footnote{see \cref{defcongruence} in appendix for a formal definition of congruence in this context.} which satisfies equations of figures \ref{fig:TracedProp} and \ref{axioms}. 
\end{definition}

Equations  \eqref{idbox} and \eqref{fusionuvsimple} in Figure \ref{axioms} reflect the monoidal structure of the matrices, with the identity element (Equation \eqref{idbox}) and the associative binary operation (Equation \eqref{fusionuvsimple}).
Equations \eqref{negu} and \eqref{bsuu} mean that both the polarising beam splitter and the polarisation flip commute with a gate.  Moreover, the polarising beam splitter is self inverse (Equation \eqref{bsbs}).
Notice that the negation is also self-inverse and that this is a consequence of the axioms (see Example \ref{ex:doubleneg}). Equation \eqref{bsnnnn} translates the fact that flipping the control state before and after performing a control of the position results in flipping the final position.
To give a meaning to Equation \eqref{bsnbsh}, it is useful to flip it upside down, and to remark that in a two-wire diagram, polarising beam splitters and negations on the bottom wire each perform a CNOT on the qubits representing the polarisation and the position, in opposite ways, so that each side of the equation combines 3 CNOTs and thus performs a swap between these two qubits.
In Equation \eqref{duplicateloop}, there are essentially two steps: first, the wire with the gate $V$ is a dead code, as no photon can go to the wire, so it can be discarded; the second step consists in merging the two polarising beam splitters. 
Equation \eqref{bsbsbs} is the only equation acting on three wires: in this particular configuration given by the left hand side of the equation, two polarising beam splitters can be replaced by swaps.
Equation \eqref{loopemptysimple} reflects the fact that isolated parts of a diagram have no effect on the rest.

\begin{example}\label{ex:doubleneg}
The fact that the negation is self inverse can be derived in the PBS-calculus:  $\textup{PBS} \vdash \begin{tikzpicture}[scale=0.8]
	\begin{pgfonlayer}{nodelayer}
		\node [style=cartouche] (0) at (0, 0) {\footnotesize$\neg$};
		\node [style=none] (1) at (-0.8, 0) {};
		\node [style=none] (2) at (0.8, 0) {};
	\end{pgfonlayer}
	\begin{pgfonlayer}{edgelayer}
		\draw [style=noire] (1.center) to (0);
		\draw [style=noire] (0) to (2.center);
	\end{pgfonlayer}
\end{tikzpicture}
\!\!\begin{tikzpicture}[scale=0.8]
	\begin{pgfonlayer}{nodelayer}
		\node [style=cartouche] (0) at (0, 0) {\footnotesize$\neg$};
		\node [style=none] (1) at (-0.8, 0) {};
		\node [style=none] (2) at (0.8, 0) {};
	\end{pgfonlayer}
	\begin{pgfonlayer}{edgelayer}
		\draw [style=noire] (1.center) to (0);
		\draw [style=noire] (0) to (2.center);
	\end{pgfonlayer}
\end{tikzpicture}
 = \begin{tikzpicture}
	\begin{pgfonlayer}{nodelayer}
		\node [style=none] (0) at (-0.3, 0) {};
		\node [style=none] (1) at (0.3, 0) {};
	\end{pgfonlayer}
	\begin{pgfonlayer}{edgelayer}
		\draw [style=noire] (0.center) to (1.center);
	\end{pgfonlayer}
\end{tikzpicture}$ (the derivation is given in \cref{derivationnegneg}). A more sophisticated example is the proof that the two diagrams of Figure \ref{fig:perm3} are equivalent, given in \cref{equivperm3}.
\end{example}

\setlength{\columnsep}{0cm}
\newcommand{\sca}{0.7}
\begin{figure}
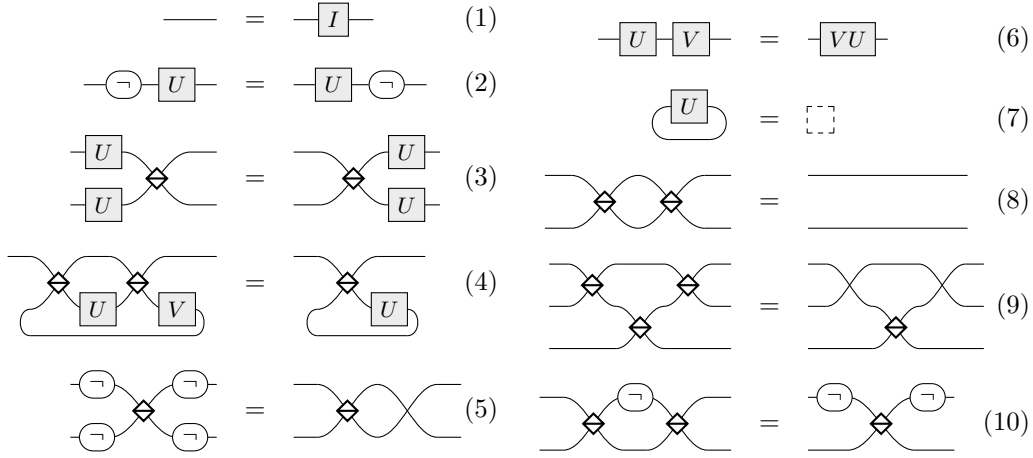

\begin{multicols}{2}
\setcounter{equation}{0}
\begin{eqnarray}
\label{idbox}\input{filcourt-s2.tikz}&=&\input{gateI-s.tikz}\\[0.3cm]
\label{negu}\input{negU-s.tikz}&=&\input{Uneg-s.tikz}\\[0.3cm]
\label{bsuu}\input{UUbeamsplitter-s.tikz}&=&\input{beamsplitterUU-s.tikz}\\[0.3cm]
\label{duplicateloop}\input{filtrerienUetV-s.tikz}&=&\input{boucletraverseU-s.tikz}\\[0.3cm]
\label{bsnnnn}\input{beamsplitternnnn-s.tikz}&=&\input{beamsplitterswap-s.tikz}
\end{eqnarray}

\begin{eqnarray}
\label{fusionuvsimple}\input{gateUgateV-s.tikz}&=&\input{gateVU-s.tikz}\\[0.3cm]
\label{loopemptysimple}\input{bouclevideU-s.tikz}&=&\input{diagrammevide-s.tikz}\\[0.3cm]
\label{bsbs}\input{beamsplitterbeamsplitter-s.tikz}&=&\input{filsparalleleslongs-s.tikz}\\[0.3cm]
\label{bsbsbs}\input{bsbsbspointeenbas-s.tikz}&=&\input{xbsxpointeenbas-s.tikz}\\[0.3cm]
\label{bsnbsh}\input{beamsplitternhautbeamsplitter-s.tikz}&=&\input{beamsplitternnhaut-s.tikz}
\end{eqnarray}
\end{multicols}
\caption{Axioms of the PBS-calculus. Given $q$ a  positive integer, $U,V \in \mathbb C^{q\times q}$ are arbitrary matrices, $I \in \mathbb C^{q\times q}$  is the identity. \label{axioms}}
\end{figure}

All these equations preserve the semantics of the PBS-diagrams: 

\begin{proposition}[Soundness]\label{soundnessoftheequations}
For any two diagrams $D_1$ and $D_2$, if $\textup{PBS}\vdash D_1 = D_2$ then $\interp{D_1}=\interp{D_2}$.
\end{proposition}
\begin{proof}\vspace{-0.16cm}The proof is given in appendix, Section \ref{proofsoundness}. \end{proof}

\subsection{Normal forms}

In this section, we introduce a notion of diagrams in normal form which is used in the next sections to prove both the universality and the  completeness of the PBS-calculus. 
They are made of two parts: the first one corresponds to a superposition of linear maps, and the second one corresponds to a permutation of the polarisations and positions, written in a way that is convenient here.

\begin{definition}[Normal Form] 
Diagrams in normal form are inductively defined as:\begin{tikzpicture}[scale=0.3]
	\begin{pgfonlayer}{nodelayer}
		\node [style=xsdiagrammevide] (0) at (0, 0) {};
	\end{pgfonlayer}
\end{tikzpicture}
 is in normal form, and for any $N:n\to n$ in normal form, 
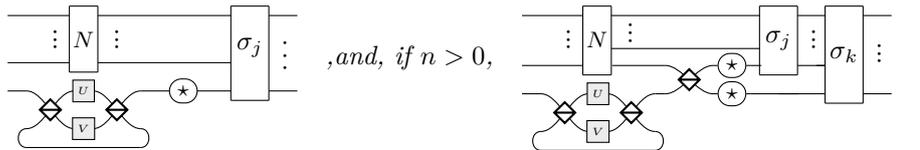
\[\input{typeAfiltreUVsansprime.tikz}\quad\text{,and, if $n>0$,}\quad \input{typeCfiltreUVsansprime.tikz}\quad,\]

\vspace{-0.5cm}
\noindent are in normal form, where\begin{tikzpicture}[scale=0.6]
	\begin{pgfonlayer}{nodelayer}
		\node [style=none] (1) at (-1.2, 0) {};
		\node [style=none] (2) at (1.2, 0) {};
		\node [style=snegpotentiel] (4) at (0, 0) {};
	\end{pgfonlayer} 
	\begin{pgfonlayer}{edgelayer}
		\draw [style=noire] (4) to (2.center);
		\draw [style=noire] (1.center) to (4);
	\end{pgfonlayer}
\end{tikzpicture}
 denotes either
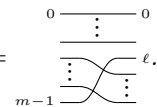
\begin{tikzpicture}
	\begin{pgfonlayer}{nodelayer}
		\node [style=none] (0) at (-0.3, 0) {};
		\node [style=none] (1) at (0.3, 0) {};
	\end{pgfonlayer}
	\begin{pgfonlayer}{edgelayer}
		\draw [style=noire] (0.center) to (1.center);
	\end{pgfonlayer}
\end{tikzpicture} or\begin{tikzpicture}[scale=0.8]
	\begin{pgfonlayer}{nodelayer}
		\node [style=cartouche] (0) at (0, 0) {\footnotesize$\neg$};
		\node [style=none] (1) at (-0.8, 0) {};
		\node [style=none] (2) at (0.8, 0) {};
	\end{pgfonlayer}
	\begin{pgfonlayer}{edgelayer}
		\draw [style=noire] (1.center) to (0);
		\draw [style=noire] (0) to (2.center);
	\end{pgfonlayer}
\end{tikzpicture}
, and $\sigma_\ell :m\to m= \input{sigmaell.tikz}$. 
\end{definition}

\vspace{-0.3cm}
\begin{remark}\label{decompNFsupperm}
For any $U,V\in \mathbb C^{q\times q}$ let $\textup{E}(U,V) := \input{filtreUV-s.tikz}$. A diagram in normal form can be written in the form $P\circ E$, where $E$ is of the form $E(U_0,V_0)\otimes\cdots\otimes E(U_{n-1},V_{n-1})$, and $P$ is built using only\begin{tikzpicture}
	\begin{pgfonlayer}{nodelayer}
		\node [style=none] (0) at (-0.3, 0) {};
		\node [style=none] (1) at (0.3, 0) {};
	\end{pgfonlayer}
	\begin{pgfonlayer}{edgelayer}
		\draw [style=noire] (0.center) to (1.center);
	\end{pgfonlayer}
\end{tikzpicture},\begin{tikzpicture}[scale=0.8]
	\begin{pgfonlayer}{nodelayer}
		\node [style=cartouche] (0) at (0, 0) {\tiny$\neg$};
		\node [style=none] (1) at (-0.8, 0) {};
		\node [style=none] (2) at (0.8, 0) {};
	\end{pgfonlayer}
	\begin{pgfonlayer}{edgelayer}
		\draw [style=noire] (1.center) to (0);
		\draw [style=noire] (0) to (2.center);
	\end{pgfonlayer}
\end{tikzpicture}
,\input{beamsplitter-xs.tikz},\input{swap-xs.tikz}, $\circ$ and $\otimes$.
\end{remark}

In the following we show that any diagram is equivalent to a diagram in normal form. 

\begin{lemma}
If $N_1$ and $N_2$ are in normal form then $N_1\otimes  N_2$ is in normal form. 
\end{lemma}
\begin{proof}\vspace{-0.16cm} By definition of the normal forms. \end{proof}

\begin{lemma}\label{NFcomposition}
If $N_1:n\to n$ and $N_2: n\to n$ are in normal form then there exists $N':n\to n$ in normal form such that $\textup{PBS}\vdash N_2\circ N_1 =N'$. 
\end{lemma}
\begin{proof}\vspace{-0.16cm} Notice that using the axioms of PROP, $N_2 = g_\ell\circ \ldots \circ g_0$ where each $g_k$ consists of either $E(U,V)$, $\begin{tikzpicture}
	\begin{pgfonlayer}{nodelayer}
		\node [style=none] (0) at (-0.3, 0) {};
		\node [style=none] (1) at (0.3, 0) {};
	\end{pgfonlayer}
	\begin{pgfonlayer}{edgelayer}
		\draw [style=noire] (0.center) to (1.center);
	\end{pgfonlayer}
\end{tikzpicture}$, $\begin{tikzpicture}[scale=0.8]
	\begin{pgfonlayer}{nodelayer}
		\node [style=cartouche] (0) at (0, 0) {\tiny$\neg$};
		\node [style=none] (1) at (-0.8, 0) {};
		\node [style=none] (2) at (0.8, 0) {};
	\end{pgfonlayer}
	\begin{pgfonlayer}{edgelayer}
		\draw [style=noire] (1.center) to (0);
		\draw [style=noire] (0) to (2.center);
	\end{pgfonlayer}
\end{tikzpicture}
$, $\input{beamsplitter-xs.tikz}$ or $\input{swap-xs.tikz}$ acting on any one or two consecutive positions, in parallel with the identity on the other positions. We show that every $g_k$ can be successively integrated to the normal form (see  appendix, Section \ref{proofcompleteness_comp}).
\end{proof}

\begin{lemma}\label{NFtrace}
If $N:n+1\to n+1$ is in normal form then there exists $N':n\to n$ in normal form such that $\textup{PBS}\vdash Tr(N)=N'$.
\end{lemma}
\begin{proof}\vspace{-0.16cm}The proof is given in appendix, Section \ref{proofcomptrace}. \end{proof}

We are now ready to prove that any PBS-diagram can be put in normal form:
\begin{proposition}\label{existenceofthenormalform}
For any $D:n\to n$, there exists a \textup{PBS}-diagram $N: n\to n$ in normal form such that $\textup{PBS}\vdash D = N$.
\end{proposition}
\begin{proof}\vspace{-0.16cm}
Combining the previous three lemmas, it remains to prove that any generator of the language can be put in normal form. We do so in appendix, Section \ref{proofNFGenerators}.
\end{proof}

\begin{remark}
By unfolding the proof of \cref{existenceofthenormalform}, one can obtain a deterministic procedure to transform any diagram into its normal form. Its complexity, defined as the number of transformations by one of Equations \eqref{idbox} to \eqref{bsnbsh}, is $\mathcal O\bigl(tm^2\bigr)$, where $m$ is the number of generators (\input{beamsplitter-xs.tikz},   \begin{tikzpicture}[scale=0.8]
	\begin{pgfonlayer}{nodelayer}
		\node [style=cartouche] (0) at (0, 0) {\tiny$\neg$};
		\node [style=none] (1) at (-0.8, 0) {};
		\node [style=none] (2) at (0.8, 0) {};
	\end{pgfonlayer}
	\begin{pgfonlayer}{edgelayer}
		\draw [style=noire] (1.center) to (0);
		\draw [style=noire] (0) to (2.center);
	\end{pgfonlayer}
\end{tikzpicture}
, and  \begin{tikzpicture}[scale=0.6]
	\begin{pgfonlayer}{nodelayer}
		\node [style=newe] (0) at (0, 0) {\tiny\!$U$\!};
		\node [style=none] (1) at (-1, 0) {};
		\node [style=none] (2) at (1, 0) {};
	\end{pgfonlayer}
	\begin{pgfonlayer}{edgelayer}
		\draw [style=new] (1.center) to (2.center);
	\end{pgfonlayer}
\end{tikzpicture}), and $t$ the number of traces in the diagram. Notice that this procedure has probably not the best possible complexity.
\end{remark}

\subsection{Completeness}

The main application of the normal forms is the proof of completeness: 

\begin{theorem}[Completeness]
For any $D,D':n\to n$, if $\interp D = \interp{D'}$ then $\textup{PBS}\vdash D = D'$. 
\end{theorem}
\begin{proof}\vspace{-0.16cm}
There exist $N,N'$ in normal form such that $\textup{PBS}\vdash D= N$ and $\textup{PBS}\vdash D'= N'$. Moreover, by soundness, $\interp N = \interp{D} = \interp {D'} = \interp{N'}$. Finally,   one can show  that $\interp N= \interp {N'}$ implies that $N=N'$. In particular, one can show inductively that the normal form is entirely determined by its semantics by considering the path semantics for a particle located on the last input wire.
\end{proof}

\subsection{Minimality of the set of axioms}

In the following we show that each of the ten equations of Figure \ref{axioms} is necessary for the completeness of the PBS-calculus: 

\begin{theorem}[Minimality]\label{minimality}
None of Equations \eqref{idbox} to \eqref{bsnbsh} is a consequence of the others.
\end{theorem}

\begin{proof}\vspace{-0.16cm}The proof is given in appendix, Section \ref{proofmin}. \end{proof} 

Notice that all equations involving matrices, except Equation \eqref{idbox}, are schemes of equations i.e. one equation for each possible matrix (or matrices). In   Theorem \ref{minimality}, we show that each of these equations, for most of the matrices, cannot be derived from the other axioms. More precisely, Equation \eqref{duplicateloop}  (resp. \eqref{loopemptysimple}) is not a consequence of the nine others for any $U$ (resp. any $U,V$); Equation  \eqref{negu}  (resp. \eqref{fusionuvsimple} ) is not a consequence of the others for any $U\neq I$ (resp. any $U,V\neq I$). Finally, if $\det(U)\neq 1$, then Equation \eqref{bsuu} is not a consequence of the others. We conjecture that the condition $\det(U)\neq 1$ can be relaxed to $U\neq I$.

\subsection{Universality}

\vspace{-0.1cm}

A PBS-diagram represents a superposition of linear maps together with a permutation of polarisations and positions. Indeed, \cref{welldefinednessofdenotationalsemantics} shows that for any diagram $D:n\to n$, $\interp D \in \S_n$, where $\S_n$ is the monoid of the linear maps $f\colon \mathcal H_n \to \mathcal H_n$ such that $f\ket{c,p,x} = \ket{\tau (c,p)}\otimes U_{c,p}\ket x$ for some permutation $\tau$ on $\hv \times [n]$ and matrices $U_{c,p}\in \mathbb C^{q\times q}$.  
We show in the following that the PBS-calculus is universal, in the sense that  any linear map in $\S_n$ can be represented by a PBS-diagram: 
\begin{theorem}
The \textup{PBS}-calculus is universal: for any $f\in \S_n$, $\exists D : n\to n$, $\interp D = f$.  
\end{theorem}

\begin{proof}\vspace{-0.16cm}
The proof relies on the normal forms: given a linear map $f\in \S_n$ one can inductively construct a diagram in normal form, by considering the image of $f$ when the particle is located on the last position ($p=n-1$).
\end{proof}

Notice that $\S_n$ is strictly included in the set of linear maps from $\mathcal H_n$ to $\mathcal H_n$. Thus while being universal for $\S_n$ the PBS-diagrams are not expressive enough to represent a (non-polarising) beam splitter for instance.

\vspace{-0.2cm}

\section{Removing the trace -- Loop unrolling}\label{trace-free}

\vspace{-0.2cm}

We consider in this section an application of the PBS-calculus. The semantics of the language points out that each trace, or feedback loop, is {used} at most twice. As a consequence, a natural question is to decide whether all loops can be unrolled, in order to transform any PBS-diagram into a trace-free PBS-diagram. Such a transformation is possible when all matrices are invertible:

\begin{proposition}\label{traceinutilesiunitaire}
Let $D:n\to n$ with $n\geq 2$ be a \textup{PBS}-diagram such that all matrices appearing in some gate \begin{tikzpicture}[scale=0.6]
	\begin{pgfonlayer}{nodelayer}
		\node [style=newe] (0) at (0, 0) {\tiny\!$U$\!};
		\node [style=none] (1) at (-1, 0) {};
		\node [style=none] (2) at (1, 0) {};
	\end{pgfonlayer}
	\begin{pgfonlayer}{edgelayer}
		\draw [style=new] (1.center) to (2.center);
	\end{pgfonlayer}
\end{tikzpicture} in $D$ are invertible. Then there exists a trace-free \textup{PBS}-diagram $D'$ such that $\textup{PBS}\vdash D=D'$.
\end{proposition}

\begin{proof}\vspace{-0.16cm}The proof is given in appendix, Section \ref{preuvetraceinutilesiunitaire}. \end{proof}

Notice that Proposition \ref{traceinutilesiunitaire} is not true for PBS-diagrams with a single input/output. Indeed  a trace-free diagram of type $1\to 1$ is made of  generators acting on $1$ wire only, so in particular it has no polarising beam splitter and as a consequence cannot have a behaviour which depends on the polarisation.  For instance, the diagram $E(U,V)$ used in the normal forms (see Remark \ref{decompNFsupperm}) cannot be transformed into a trace-free diagram unless $U=V$.

On the other hand, PBS-diagrams involving at least one non-invertible matrix are not necessarily equivalent to a trace-free one. Indeed, we have the following property:

\begin{lemma}\label{zerooudeuxnoninversibles}
For any trace-free PBS-diagram $D$, either all $[D]_{c,p}$ are invertible or at least two of them are not.
\end{lemma}

\begin{proof}\vspace{-0.16cm}The proof is given in appendix, Section \ref{preuvezerooudeuxnoninversibles}. \end{proof}

This prevents the following diagram from being equivalent to a trace-free one:

\begin{example}
If $U$ is not invertible, then the  diagram $D_U:2\to 2 = \input{filtreUriensurfil.tikz}$ is not equivalent, according to the rules of the PBS-calculus, to any trace-free diagram.
Indeed, for any $(c,p)\neq(\rightarrow,1)$ we have $[D_U]_{c,p}=I_q$, which is invertible, whereas $[D_U]_{\rightarrow,1}=U$.
\end{example}

Another interesting property is that loop unrolling, when it is possible, requires the use of matrices that were not present in the original diagram.  
This is a consequence of the following lemma:

\vspace{-0.1cm}

\begin{lemma}\label{proddetscarresU}
Given any diagram $D:n\to n$, let us define $|D|:=\displaystyle\prod_{c\in\hv,p\in[n]}\det\left([D]_{c,p}\right)$. Then for any trace-free diagram $D$, we have $|D|=\displaystyle\prod_{G\text{ gate in }D}\det\left(U(G)\right)^2$ where $U(G)$ denotes the matrix with which $G$ is labelled.
\end{lemma}

\begin{proof}\vspace{-0.16cm}
Intuitively, due to the invertibility of the PBS-diagrams (\cref{reversibility}), for each wire of a trace-free diagram $D$, there are exactly two initial configurations which are going through this particular wire. As a consequence each gate of $D$ contributes twice to $|D|$ (see  appendix, Section \ref{preuveproddetscarresU}).
\end{proof}

\begin{example}
Unless $\det(U)$ is a $k$th root of unity for some odd integer $k$, the following diagram $D_U$ does not have the same semantics as any trace-free diagram in which all gates are labelled by $U$: \input{filtreUriensurfil.tikz}.
Indeed, we have $|D_U|=\det(U)$, and by \cref{proddetscarresU}, if $D_U$ is equivalent through $PBS$ to a trace-free diagram $D_U'$ in which all gates are labelled by $U$, then we have $|D_U|=\det(U)=\det(U)^{2N}$, where $N$ is the number of gates in $D_U'$. By \cref{zerooudeuxnoninversibles}, we have $\det(U)\neq0$, so that $\det(U)^{2N-1}=1$, that is, $\det(U)$ is a $k$th root of unity with $k=2N-1$ odd (if $N=0$ then $\det(U)=1$ so the result is still true).
\end{example}

\section{Conclusion and Perspectives }

In this paper, we have introduced a rigorous framework to reason on quantum computations involving coherent control, which are sometimes informally represented by schemes involving polarising beam splitters and black boxes. The main result is the introduction of an equational theory which makes the PBS-calculus sound and complete. We have also proved that the axiomatisation is minimal in the sense that each axiom is necessary for the completeness. Moreover, we have demonstrated for instance that the PBS-calculus can be used for loop unrolling.  

So we have introduced the foundations of a formal framework, that we believe will be a useful tool to study the power and the limits of computations and protocols involving coherent control. We mention here three perspectives in the development of the PBS-calculus.

First, the expressivity of the language can be increased by adding,  for instance, a (not polarising) beam splitter  as a generator of the language, or by allowing more than one particle in the diagrams. Both are necessary for the representation of Boson sampling for instance. 

Another perspective is to allow the gates to be arbitrary quantum channels. Indeed recent results \cite{ebler2018enhanced,abbott2018communication} point out interesting and unexpected behaviours of  coherently controlled quantum channels. Our objective is to make the PBS-calculus a formal framework to explore and study such phenomena.

Finally, the calculus can be made more resource-sensitive, by allowing only the equations for which the number of occurrences of each gate (or black box) is preserved. For instance, we have seen examples in which  loop unrolling requires to introduce new gates that were not present in the initial diagram. Transforming a diagram into its normal form is another example that does not, in general, preserve the number of occurrences of each gate.

\bibliography{pbs-diagramsMFCSarxiv}

\appendix

\section{Categorical notions}\label{categoricalnotions}
\begin{figure}[hb]
$\changelargeurcentre{(D_3\circ D_2)\circ D_1 =  D_3\circ (D_2\circ D_1)}{(D_1\otimes D_1)\otimes D_1 =  D_1\otimes (D_2\otimes D)}
\quad\ |\quad\ I_k\circ D = D= D\circ I_k\quad|\quad  \input{diagrammevide-s.tikz}\otimes D = D = D\otimes \input{diagrammevide-s.tikz} \\
(D_1\otimes D_2)\otimes D_3 =  D_1\otimes (D_2\otimes D_3) \quad|\quad (D_3\circ D_1)\otimes(D_4\circ D_2) = (D_3\otimes D_4)\circ (D_1\otimes D_2) \\
\changelargeurcentre{\sigma_{1,k}\circ (I_1\otimes D) = (D\otimes I_1)\circ \sigma_{1,k}}{(D_1\otimes D))\otimes D_1 =  D_1\otimes (D_1\otimes D_3)}
\quad\quad\phantom{a}|\negphantom{a}
\changelargeurcentre{Tr(D_1\otimes D_2) = D_1\otimes Tr(D_2)}{(D_3\circ D_1)\otimes(D_4\circ D) = (D_3\otimes D_4)\circ (D_1\otimes D_2)} \\
\input{swap-xs.tikz}\circ\input{swap-xs.tikz}=I_2\quad|\quad Tr(D_2\circ(D_1\otimes\input{filcourt-s.tikz}))=Tr(D_2)\circ D_1\quad|\quad Tr((D_2\otimes\input{filcourt-s.tikz})\circ D_1)=D_2\circ Tr(D_1) \\
Tr_k((I_n\otimes D_2)\circ D_1)=Tr_k(D_1\circ(I_n\otimes D_2)) \text{ where $D_2:k\to k$}\qquad\quad|\qquad\quad Tr(\input{swap-xs.tikz})=\input{filcourt-s.tikz}$.
\caption{Structural congruence / Coherence conditions of Traced PROP,  see Definition \ref{deftracedprop} for details. $I_0 := \input{diagrammevide-s.tikz}$; $I_{k+1} := I_k\otimes  \input{filcourt-s.tikz}$. $\sigma_{1,0}:=\input{filcourt-s.tikz}$; $\sigma_{1,k+1}:=(I_k\otimes\input{swap-xs.tikz})\circ(\sigma_{1,k}\otimes\input{filcourt-s.tikz})$. $Tr_0(D):=D$; $Tr_{k+1}(D):=Tr(Tr_k(D))$.}
\label{fig:TracedProp}
\end{figure}

\begin{definition}\label{deftracedprop}
A traced PROP is a category $\mathcal C$ whose objects are the natural integers, equipped with
\begin{itemize}
\item a functor $\otimes\colon \mathcal C\times \mathcal C\to \mathcal C$ such that for any $n,m\geq 0$, $n\otimes m=n+m$
\item an arrow $\sigma\colon 2\to 2$
\item a family of functions $\noeqbreak Tr\colon \mathcal C(n+1,m+1)\to \mathcal C(n,m)$
\end{itemize}
satisfying the properties listed below.\\

Let $\sigma_{1,0}:=\id_1$; and for any $n\geq 0$, let $\sigma_{1,n+1}:=(\id_{n}\otimes\sigma)\circ(\sigma_{1,n}\otimes\id_{1})$.
Let $Tr_0(f):=f$; and for any $j\geq 0$, let $Tr_{j+1}:=Tr\circ Tr_j$.\\

That $\mathcal C$ is a category means that it satisfies the following properties:

\begin{itemize}
\item \emph{Neutrality of the identity:} for any $n,m$ and $f\colon n\to m$, \[\begin{array}{rcccl}f\circ \id_n&=&f&=&\id_m\circ f.\\\\
\input{filsf.tikz}&=&\input{f.tikz}&=&\input{ffils.tikz}\\&\end{array}\]
\item \emph{Associativity of composition:} for any $n,m,k,l\geq 0$, $f\colon n\to m$, $g\colon m\to k$ and $h\colon k\to l$, \[\begin{array}{rcl}(h\circ g)\circ f&=&h\circ(g\circ f).\\\\
\input{fboitegh.tikz}&=&\input{fgboiteh.tikz}\\&\end{array}\]
\end{itemize}

That $\otimes$ is a functor means that it satisfies the following properties:

\begin{itemize}
\item \emph{Preservation of source and target:} for any $n,m,k,l\geq 0$, $f\colon n\to m$ and $g\colon k\to l$, \[\begin{array}{c}f\otimes g\colon n+k\to m+l.\\\\
\input{fsurg.tikz}\\\end{array}\]
\item \emph{Preservation of the identity:} for any $n,m\geq 0$, \[\begin{array}{rcl}\id_n\otimes\id_m&=&\id_{n+m}.\\\\
\input{idmultiplesuridmultiple.tikz}&=&\input{idmultiples.tikz}\\&\end{array}\]
\item \emph{Preservation of composition:} for any $n,m,m',k,l,l'\geq 0$, $f_1\colon n\to m$, $f_2\colon m\to m'$, $g_1\colon k\to l$ and $g_2\colon l\to l'$, \[\begin{array}{rcl}(f_2\circ f_1)\otimes (g_2\circ g_1)&=&(f_2\otimes g_2)\circ(f_1\otimes g_1).\\\\
\input{f1f2boitesurg1g2boite.tikz}&=&\input{f1surg1boitef2surg2boite.tikz}\\&\end{array}\]
\end{itemize}

The additional required properties are:\\\\

\emph{Strict monoidal structure:}
\begin{itemize}
\item \emph{Naturality of the associator:} for any $n,n',m,m',k,k'$, $f\colon n\to n'$, $g\colon m\to m'$ and $h\colon k\to k'$, \[\begin{array}{rcl}(f\otimes g)\otimes h&=&f\otimes(g\otimes h).\\\\
\input{fsurgboitesurh.tikz}&=&\input{fsurboitegsurh.tikz}\\&\end{array}\]
\item \emph{Naturality of the left unitor:} for any $n,m\geq 0$ and $f\colon n\to m$, \[\begin{array}{rcl}\id_0\otimes f&=&f.\\\\
\input{diagrammevidesurf56.tikz}&=&\input{f56.tikz}\\&\end{array}\]
\item \emph{Naturality of the right unitor:} for any $n,m\geq 0$ and $f\colon n\to m$, \[\begin{array}{rcl}f\otimes\id_0&=&f.\\\\
\input{f56surdiagrammevide.tikz}&=&\input{f56.tikz}\\&\end{array}\]
\end{itemize}
\emph{Strict symmetric monoidal structure:}
\begin{itemize}
\item \emph{Naturality of the swap:} for any $n,m\geq 0$ and $f\colon n\to m$, \[\tikzset{tikzfig/.style={baseline=-0.25em,scale=0.4,every node/.style={scale=0.8}}}\begin{array}{rcl}\sigma_{1,m}\circ(\id_1\otimes f)&=&(f\otimes \id_1)\circ\sigma_{1,n}.\\\\
\input{id1fmultiswap.tikz}&=&\input{multiswapfid1.tikz}\\&\end{array}\]
\item \emph{Inverse law:}
\[\tikzset{tikzfig/.style={baseline=-0.25em,scale=0.35,every node/.style={scale=0.8},borddiagrammevide/.style={-, dash pattern=on \traitsdiagrammevide off \traitsdiagrammevide on \traitsdiagrammevide off \traitsdiagrammevide on \traitsdiagrammevide off 0em}}}\begin{array}{rcl}\sigma\circ\sigma&=&\id_{2}.\\\\
\input{swapswap.tikz}&=&\begin{tikzpicture}
	\begin{pgfonlayer}{nodelayer}
		\node [style=none] (0) at (-3, 1) {};
		\node [style=none] (1) at (3, 1) {};
		\node [style=none] (2) at (3, -1) {};
		\node [style=none] (3) at (-3, -1) {};
	\end{pgfonlayer}
	\begin{pgfonlayer}{edgelayer}
		\draw [style=new] (0.center) to (1.center);
		\draw [style=new] (3.center) to (2.center);
	\end{pgfonlayer}
\end{tikzpicture}\\&\end{array}\]
\end{itemize}
\emph{Axioms about the trace:}
\begin{itemize}\tikzset{tikzfig/.style={baseline=-0.25em,scale=0.35,every node/.style={scale=0.7,font=\Large}}}
\item \emph{Naturality in the input:} for any $n,m,l\geq 0$, $f\colon l\to n$ and $g\colon n+1\to m+1$, \[\begin{array}{rcl}Tr(g\circ(f\otimes\id_1))&=&Tr(g)\circ f.\\\\
\input{tracefhgrandg.tikz}&=&\input{fhpuistracegrandg.tikz}\\&\end{array}\]
\item \emph{Naturality in the output:} for any $n,m,l\geq 0$, $f\colon n+1\to m+1$ and $g\colon m\to l$, \[\begin{array}{rcl}Tr((g\otimes\id_1)\circ f)&=&g\circ Tr(f).\\\\
\input{tracegrandfgh.tikz}&=&\input{tracegrandfpuisgh.tikz}\\&\end{array}\]
\item \emph{Dinaturality:} for any $n,m,i,j\geq 0$, $f\colon n+i\to m+j$ and $g\colon j\to i$, \[\begin{array}{rcl}Tr_i((\id_m\otimes g)\circ f)&=&Tr_j(f\circ(\id_n\otimes g)).\\\\
\input{tracegrandfgb.tikz}&=&\input{tracegbgrandf.tikz}\\&\end{array}\]
\item \emph{Superposing:} for any $n,m,l,p\geq 0$, $f\colon n+1\to m+1$ and $g\colon l\to p$, \[\tikzset{tikzfig/.style={baseline=-0.25em,scale=0.35,every node/.style={scale=0.68, font=\Large},borddiagrammevide/.style={-, dash pattern=on \traitsdiagrammevide off \traitsdiagrammevide on \traitsdiagrammevide off \traitsdiagrammevide on \traitsdiagrammevide off 0em}}}\begin{array}{rcl}Tr(g\otimes f)&=&g\otimes Tr(f).\\\\
\input{tracegsurgrandfboite.tikz}&=&\input{gsurtracegrandfboite.tikz}\\&\end{array}\]
\item \emph{Yanking:}
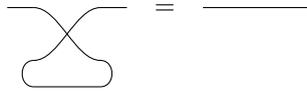
\[\begin{array}{rcl}Tr(\sigma)&=&\id_1.\\\\
\input{yankingvariantecentresurfil.tikz}&=&\begin{tikzpicture}
	\begin{pgfonlayer}{nodelayer}
		\node [style=none] (0) at (-2, 0) {};
		\node [style=none] (1) at (2, 0) {};
	\end{pgfonlayer}
	\begin{pgfonlayer}{edgelayer}
		\draw [style=new] (0.center) to (1.center);
	\end{pgfonlayer}
\end{tikzpicture}\\&\end{array}\]
\end{itemize}

\end{definition}

\begin{remark}
We can define $\sigma_{0,m}:=\id_m$ and for any $n\geq0$, $\sigma_{n+1,m}:=(\sigma_{n,m}\otimes\id_1)\circ(\id_m\otimes\sigma_{1,m})$. Then the unit coherence $\sigma_{n,0}=\id_n$ and the associativity coherence $(\id_m\otimes\sigma_{n,k})\circ(\sigma_{n,m}\otimes\id_k)=\sigma_{n,m+k}$ of the multiwire swaps are consequences of the definition of $\sigma_{n,m}$. The general inverse law axiom $\sigma_{n,m}\circ\sigma_{m,n}=\id_{n+m}$ can be deduced easily from its restricted version by induction on $n$ and $m$. The general naturality of the swap, $\sigma_{n',m'}\circ(f\otimes g)=(g\otimes f)\circ\sigma_{n,m}$ for any $f\colon n\to n'$ and $g\colon m\to m'$, can be deduced from its restricted version by first iterating it to get a swap with multiple wires on both sides, then using the general inverse law axiom to flip it upside down. The vanishing axioms $Tr_0(f)=f$ and $Tr_i(Tr_j(f))=Tr_{i+j}(f)$ are consequences of the definition of $Tr_j$. The general yanking axiom $Tr_n(\sigma_{n,n})=\id_n$ can be deduced by induction on $n$ from its restricted version using a multiwire version of the Yang-Baxter equation (which is a consequence of the naturality of the swap) and the inverse law. The general (that is, with multiwire traces) versions of naturality of the trace in the input, in the output, and of superposing can be deduced by iteration of their respective restricted versions.
\end{remark}

\begin{definition}\label{defprotweb}
A \emph{traced weak braided category} is a strict monoidal category that is additionally a weak braided monoidal category in the sense of \cite{FSV2011homc} or \cite{solberg2015weak} and a right traced category in the sense of \cite{Sel2009-graphical}. A \emph{PROTWEB} is a traced weak braided category whose objects are generated from the monoidal unit and a single object by monoidal product, and identified with the natural integers.\\

Namely, a PROTWEB is a category $\mathcal C$ whose objects are the natural integers, equipped with
\begin{itemize}
\item a functor $\otimes\colon \mathcal C\times \mathcal C\to \mathcal C$ such that for any $n,m\geq 0$, $n\otimes m=n+m$
\item an arrow $\sigma\colon 2\to 2$
\item a family of functions $\noeqbreak Tr\colon \mathcal C(n+1,m+1)\to \mathcal C(n,m)$
\end{itemize}
satisfying the properties listed below.\\

Let $\sigma_{1,0}:=\id_1$; for any $n\geq 0$, let $\sigma_{1,n+1}:=(\id_{n}\otimes\sigma)\circ(\sigma_{1,n}\otimes\id_{1})$; for any $m\geq 0$, let $\sigma_{0,m}=\id_m$; and for any $n\geq0$, let $\sigma_{n+1,m}:=(\sigma_{n,m}\otimes\id_1)\circ(\id_m\otimes\sigma_{1,m})$.
Let $Tr_0(f):=f$; and for any $j\geq 0$, let $Tr_{j+1}:=Tr\circ Tr_j$.\\

On top of the fact that $\mathcal C$ is a category and that $\otimes$ is a functor and acts as the addition on objects, the required properties are, using the same names as in \cref{deftracedprop}:\\\\
\emph{Strict monoidal structure:}
\begin{itemize}
\item Naturality of the associator
\item Naturality of the left unitor
\item Naturality of the right unitor
\end{itemize}
\emph{Strict weak braided monoidal structure:}
\begin{itemize}
\item \emph{General naturality of the swap:} for any $n,n',m,m'\geq 0$, $f\colon n\to n'$ and $g\colon m\to m'$, \[\tikzset{tikzfig/.style={baseline=-0.25em,scale=0.35,every node/.style={scale=0.7,font=\Large},borddiagrammevide/.style={-, dash pattern=on \traitsdiagrammevide off \traitsdiagrammevide on \traitsdiagrammevide off \traitsdiagrammevide on \traitsdiagrammevide off 0em}}}\begin{array}{rcl}\sigma_{n'm'}\circ(f\otimes g)&=&(g\otimes f)\circ\sigma_{nm}.\\\\
\input{fgmultiswap.tikz}&=&\input{multiswapgf.tikz}\\&\end{array}\]
\end{itemize}
\emph{Axioms about the trace:}
\begin{itemize}
\item Naturality in the input
\item Naturality in the output
\item Dinaturality
\item Superposing
\end{itemize}
\end{definition}

\begin{remark}
Similarly as in the case of the axioms of a traced PROP, the unit coherence and the two associativity coherences $(\id_m\otimes\sigma_{nk})\circ(\sigma_{nm}\otimes\id_k)=\sigma_{n(m+k)}$ and $(\sigma_{nk}\otimes\id_m)\circ(\id_n\otimes\sigma_{mk})=\sigma_{n(m+k)}$ (which are both needed when $\sigma$ is not self-inverse) of the multiwire swaps are consequences of the definition of $\sigma_{nm}$, the vanishing axioms are consequences of the definition of $Tr_j$, and the general versions of naturality of the trace in the input, in the output, and of superposing can be deduced by iteration of their respective restricted versions.
\end{remark}

\begin{lemma}\label{addswapswapyanking}
A category is a traced PROP if and only if it is a PROTWEB and satisfies inverse law and yanking:
\[\tikzset{tikzfig/.style={baseline=-0.25em,scale=0.35}}\begin{array}{rcl@{\qquad}rcl}\sigma\circ\sigma&=&\id_{2}&Tr(\sigma)&=&\id_1.\\\\
\input{swapswap.tikz}&=&
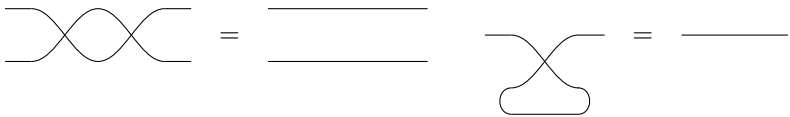
\begin{tikzpicture}
	\begin{pgfonlayer}{nodelayer}
		\node [style=none] (0) at (-3, 1) {};
		\node [style=none] (1) at (3, 1) {};
		\node [style=none] (2) at (3, -1) {};
		\node [style=none] (3) at (-3, -1) {};
	\end{pgfonlayer}
	\begin{pgfonlayer}{edgelayer}
		\draw [style=new] (0.center) to (1.center);
		\draw [style=new] (3.center) to (2.center);
	\end{pgfonlayer}
\end{tikzpicture}&\input{yankingvariantecentresurfil.tikz}&=&\begin{tikzpicture}
	\begin{pgfonlayer}{nodelayer}
		\node [style=none] (0) at (-2, 0) {};
		\node [style=none] (1) at (2, 0) {};
	\end{pgfonlayer}
	\begin{pgfonlayer}{edgelayer}
		\draw [style=new] (0.center) to (1.center);
	\end{pgfonlayer}
\end{tikzpicture}\\&\end{array}\]
\end{lemma}
\begin{proof}\vspace{-0.16cm}
This follows directly from comparing the lists of axioms given respectively in \cref{deftracedprop,defprotweb}, and from the fact that the general naturality of the swap is a consequence of the axioms of traced PROP.
\end{proof}

\section{Proofs}

\subsection{Semantics}

\subsubsection{Proof of Propositions \ref{determinismterm} and \ref{reversibility}}\label{preuvedeterminismtermreversibility}
We first prove the two propositions without assuming the axioms of traced PROP. At the end we will have to prove that any two diagrams equivalent modulo the axioms of traced PROP have the same path semantics.

Not assuming the axioms of traced PROP implies that for any diagram $D$, we are in exactly one of the following cases:
\begin{itemize}
\item $D=\begin{tikzpicture}[scale=0.3]
	\begin{pgfonlayer}{nodelayer}
		\node [style=xsdiagrammevide] (0) at (0, 0) {};
	\end{pgfonlayer}
\end{tikzpicture}
,\begin{tikzpicture}
	\begin{pgfonlayer}{nodelayer}
		\node [style=none] (0) at (-0.3, 0) {};
		\node [style=none] (1) at (0.3, 0) {};
	\end{pgfonlayer}
	\begin{pgfonlayer}{edgelayer}
		\draw [style=noire] (0.center) to (1.center);
	\end{pgfonlayer}
\end{tikzpicture},\input{swap-xs.tikz},\begin{tikzpicture}[scale=0.6]
	\begin{pgfonlayer}{nodelayer}
		\node [style=newe] (0) at (0, 0) {\tiny\!$U$\!};
		\node [style=none] (1) at (-1, 0) {};
		\node [style=none] (2) at (1, 0) {};
	\end{pgfonlayer}
	\begin{pgfonlayer}{edgelayer}
		\draw [style=new] (1.center) to (2.center);
	\end{pgfonlayer}
\end{tikzpicture}\text{ or }\input{beamsplitter-xs.tikz}$
\item there exist unique $D_1$ and $D_2$ such that $D=D_2\circ D_1$
\item there exist unique $D_1$ and $D_2$ such that $D=D_1\otimes D_2$
\item there exists a unique $D'$ such that $D=Tr(D')$.
\end{itemize}

We prove both propositions together by structural induction on $D$. \\

If $D=\begin{tikzpicture}[scale=0.3]
	\begin{pgfonlayer}{nodelayer}
		\node [style=sdiagrammevide] (0) at (0, 0) {};
	\end{pgfonlayer}
\end{tikzpicture}
$ then $\hv\times[n]$ is empty so both propositions hold.\\

If $D$ is a generator then we have $n=1$ if $D=\begin{tikzpicture}
	\begin{pgfonlayer}{nodelayer}
		\node [style=none] (0) at (-0.3, 0) {};
		\node [style=none] (1) at (0.3, 0) {};
	\end{pgfonlayer}
	\begin{pgfonlayer}{edgelayer}
		\draw [style=noire] (0.center) to (1.center);
	\end{pgfonlayer}
\end{tikzpicture},\begin{tikzpicture}[scale=0.8]
	\begin{pgfonlayer}{nodelayer}
		\node [style=cartouche] (0) at (0, 0) {\tiny$\neg$};
		\node [style=none] (1) at (-0.8, 0) {};
		\node [style=none] (2) at (0.8, 0) {};
	\end{pgfonlayer}
	\begin{pgfonlayer}{edgelayer}
		\draw [style=noire] (1.center) to (0);
		\draw [style=noire] (0) to (2.center);
	\end{pgfonlayer}
\end{tikzpicture}
\text{ or }\begin{tikzpicture}[scale=0.6]
	\begin{pgfonlayer}{nodelayer}
		\node [style=newe] (0) at (0, 0) {\tiny\!$U$\!};
		\node [style=none] (1) at (-1, 0) {};
		\node [style=none] (2) at (1, 0) {};
	\end{pgfonlayer}
	\begin{pgfonlayer}{edgelayer}
		\draw [style=new] (1.center) to (2.center);
	\end{pgfonlayer}
\end{tikzpicture}$ and $n=2$ if $D=\input{swap-xs.tikz}\text{ or }\input{beamsplitter-xs.tikz}$, and in any case it is easy to see that both propositions hold.\\

If $D=D_2\circ D_1$, then for any $(c,p)\in\hv\times[n]$, by induction hypothesis there exist unique $(c',p')\in\hv\times[n]$ and $U\in\C^{q\times q}$ such that $(D_1,c,p)\xRightarrow{U}(c',p')$, and again by induction hypothesis there exist unique $(c'',p'')\in\hv\times[n]$ and $V\in\C^{q\times q}$ such that $(D_2,c',p')\xRightarrow{V}(c'',p'')$. Therefore, there is exactly one way of meeting the premises of the only rule that can reduce $(D,c,p)$ and these premises completely determine the conclusion of the rule, so \cref{determinismterm} holds for $D$.

Similarly, for any $(c,p)\in\hv\times[n]$, by induction hypothesis there exist unique $(c',p')\in\hv\times[n]$ and $U\in\C^{q\times q}$ such that $(D_2,c',p')\xRightarrow{U}(c,p)$, and again by induction hypothesis there exist unique $(c'',p'')\in\hv\times[n]$ and $V\in\C^{q\times q}$ such that $(D_1,c'',p'')\xRightarrow{V}(c,p)$. Therefore, there is exactly one way to meet the premises of the only rule that can reduce $D$ to get a reduction with right-hand side $(c,p)$. These premises completely determine the conclusion of the rule, so \cref{reversibility} holds for $D$.\\

If $D=D_1\otimes D_2$ with $D_1:n_1\to n_1$ and $D_2:n-n_1\to n-n_2$, let $(c,p)\in\hv\times[n]$.

If $p<n_1$, then by induction hypothesis there exist unique $(c',p')\in\hv\times[n_1]$ and $U\in\C^{q\times q}$ such that $(D_1,c,p)\xRightarrow{U}(c',p')$, so that there is exactly one rule that allows us to reduce $(D,c,p)$ (Rule $\otimes1$), and exactly one way to meet its premises, so \cref{determinismterm} holds for $D$. If $p\geq n_1$, then by induction hypothesis there exist unique $(c',p')\in\hv\times[n-n_1]$ and $U\in\C^{q\times q}$ such that $(D_2,c,p-n_1)\xRightarrow{U}(c',p')$, so that there is exactly one rule that allows us to reduce $(D,c,p)$ (Rule $\otimes2$), and exactly one way to meet its premises, so \cref{determinismterm} holds for $D$.

Similarly, if $p<n_1$, then by induction hypothesis there exist unique $(c',p')\in\hv\times[n_1]$ and $U\in\C^{q\times q}$ such that $(D_1,c',p')\xRightarrow{U}(c,p)$, so that there is exactly one rule that allows us to reduce $D$ and get $(c,p)$ (Rule $\otimes1$), and exactly one way to meet its premises, so \cref{reversibility} holds for $D$. If $p\geq n_1$, then by induction hypothesis there exist unique $(c',p')\in\hv\times[n-n_1]$ and $U\in\C^{q\times q}$ such that $(D_2,c,p-n_1)\xRightarrow{U}(c',p')$, so that there is exactly one rule that allows us to reduce $D$ and get $(c,p)$ (Rule $\otimes2$), and exactly one way to meet its premises, so \cref{reversibility} holds for $D$.\\

If $D=Tr(D')$ with $D':n+1\to n+1$, then for any $(c_0,p_0)\in\hv\times[n]$, by induction hypothesis of \cref{determinismterm} there exist unique $(c_1,p_1)\in\hv\times[n+1]$ and $U_0\in\C^{q\times q}$ such that $(D',c_0,p_0)\xRightarrow{U_0}(c_1,p_1)$. If $p_1<n$, then there is exactly one reduction from $(D,c_0,p_0)$ which comes from applying Rule $\mathsf T_0$, so \cref{determinismterm} holds for $D$. If $p_1=n$, then again by induction hypothesis of \cref{determinismterm} there exist unique $(c_2,p_2)\in\hv\times[n+1]$ and $U_1\in\C^{q\times q}$ such that $(D',c_1,n)\xRightarrow{U_1}(c_2,p_2)$. If $p_2<n$, then there is exactly one reduction from $(D,c_0,p_0)$, which comes from applying Rule $\mathsf T_1$, so \cref{determinismterm} holds for $D$.

By uniqueness in the induction hypothesis of \cref{reversibility}, since $(D',c_0,p_0)\xRightarrow{U_0}(c_1,n)$, $(D',c_1,n)\xRightarrow{U_1}(c_2,p_2,U_1)$ and $(c_0,p_0)\neq(c_1,n)$, we have $(c_1,n)\neq(c_2,p_2)$, so that if $p_2=n$ then $c_2=\bar c_1$. In this case, again by induction hypothesis of \cref{determinismterm}, there exist unique $(c_3,p_3)\in\hv\times[n+1]$ and $U_2\in\C^{q\times q}$ such that $(D',\bar c_1,n)\xRightarrow{U_2}(c_3,p_3)$. Again by uniqueness in the induction hypothesis of \cref{reversibility}, since $(D',c_0,p_0)\xRightarrow{U_0}(c_1,n)$ and $(c_0,p_0)\neq(\bar c_1,n)$ we have $(c_3,p_3)\neq(c_1,n)$, and since $(D',c_1,n)\xRightarrow{U_1}(\bar c_1,n)$ and $(c_1,n)\neq(\bar c_1,n)$ we have $(c_3,p_3)\neq(\bar c_1,n)$. Therefore, we cannot have $p_3=n$, so $p_3<n$ and then there is exactly one reduction from $(D,c_0,p_0)$, which comes from applying Rule $\mathsf T_2$. So \cref{determinismterm} holds for $D$.

Similarly, by induction hypothesis of \cref{reversibility} there exist unique $(c_1,p_1)\in\hv\times[n+1]$ and $U_0\in\C^{q\times q}$ such that $(D',c_1,p_1)\xRightarrow{U_0}(c_0,p_0)$. If $p_1<n$, then there is exactly one reduction from $D$ with right-hand side $(c_0,p_0)$, which comes from applying Rule $\mathsf T_0$. So \cref{reversibility} holds for $D$. If $p_1=n$, then again by induction hypothesis of \cref{reversibility} there exist unique $(c_2,p_2)\in\hv\times[n+1]$ and $U_1\in\C^{q\times q}$ such that $(D',c_2,p_2)\xRightarrow{U_1}(c_1,n)$. If $p_2<n$, then there is exactly one reduction from $D$ with right-hand side $(c_0,p_0)$, which comes from applying Rule $\mathsf T_1$. So \cref{reversibility} holds for $D$.

By uniqueness in the induction hypothesis of \cref{determinismterm}, since $(D',c_1,n)\xRightarrow{U_0}(c_0,p_0)$, $(D',c_2,p_2)\xRightarrow{U_1}(c_1,n)$ and $(c_1,n)\neq(c_0,p_0)$, we have $(c_1,n)\neq(c_2,p_2)$, so that if $p_2=n$ then $c_2=\bar c_1$. In this case, again by induction hypothesis of \cref{reversibility}, there exist unique $(c_3,p_3)\in\hv\times[n+1]$ and $U_2\in\C^{q\times q}$ such that $(D',c_3,p_3)\xRightarrow{U_2}(\bar c_1,n)$. Again by uniqueness in the induction hypothesis of \cref{determinismterm}, since $(D',c_1,n)\xRightarrow{U_0}(c_0,p_0)$ and $(c_0,p_0)\neq(\bar c_1,n)$ we have $(c_3,p_3)\neq(c_1,n)$, and since $(D',\bar c_1,n)\xRightarrow{U_1}(c_1,n)$ and $(c_1,n)\neq(\bar c_1,n)$ we have $(c_3,p_3)\neq(\bar c_1,n)$. Therefore, we cannot have $p_3=n$, so $p_3<n$ and then there is exactly one reduction from $D$ with right-hand side $(c_0,p_0)$, which comes from applying Rule $\mathsf T_2$. So \cref{reversibility} holds for $D$.\\

To finish proving the result, we have to check that two diagrams equivalent modulo the axioms of traced PROP have the same path semantics. To do this, it suffices to check for each of the axioms given in Figure \ref{fig:TracedProp} that both sides have the same operational semantics, which is straightforward.

\subsubsection{Proof of Proposition \ref{welldefinednessofdenotationalsemantics} and Theorem \ref{adequacy}}\label{proofad}

We first prove the following three lemmas:

\begin{lemma}\label{unseultermenonnuldindiceconnu}
Let $n\geq 0$ and $f\in\S_{n+1}$, and let $\tau$ be the permutation and $U_{c,p}$ the family of matrices, such that $f=\ket{c,p,y}\mapsto\ket{\tau(c,p)}\otimes U_{c,p}\ket{y}$. For any $(c,p,y)\in\hv\times[n]\times[q]$, the series $\displaystyle\sum_{k\in \mathbb N}\pi_{1} \circ (f \circ \pi_{0})^k \circ  f \circ \iota(\ket{c,p,y})$ has at most one nonzero term (exactly one if $f$ is injective), of index $k_1-1$, where $k_1$ is the smallest $k\geq 1$ such that $\tau^k(c,p)\in\hv\times[n]$, or equivalently, the smallest $k\geq 1$ such that $f^k(\ket{c,p,y})\in\H_n$. Moreover, we have $k_1\geq 3$.
\end{lemma}

\begin{lemma}\label{stabilitedeSparT}
For any $n\geq 0$ and $f\in\S_{n+1}$, $\T(f)$ is well-defined and $\T(f)\in\S_n$.
\end{lemma}

\begin{lemma}\label{calculdutermenonnulavecpermutation}
Let $n\geq 0$ and $f\in\S_{n+1}$. Let $\tau$ be the permutation and $U_{c,p}$ the family of matrices, such that $f=\ket{c,p,y}\mapsto\ket{\tau(c,p)}\otimes U_{c,p}\ket{y}$. For any $(c,p,y)\in\hv\times[n]\times[q]$, we have $\T(f)(\ket{c,p,y})=\ket{\tau^{k_1}(c,p)}\otimes U_{\tau^{k_1-1}(c,p)}\cdots U_{c,p}\ket{y}$, where $k_1$ is the smallest $k\geq 1$ such that $\tau^k(c,p)\in\hv\times[n]$.
\end{lemma}

\begin{proof}[Proof of \cref{unseultermenonnuldindiceconnu,calculdutermenonnulavecpermutation}]\vspace{-0.16cm}\binoppenalty10000
Let $(c,p,y)\in\hv\times[n]\times[q]$ and let $k_1$ be the smallest $k\geq 1$ such that $\tau^k(c,p)\in\hv\times[n]$. Since the sequence $(\tau^{k}(c,p))_{k\in\mathbb{N}}$ is periodic and $\tau^0(c,p)=(c,p)\in\hv\times[n]$, $k_1$ exists. Since $\tau$ is injective, if there were $\noeqbreak{} 1\leq k'<k''\leq k_1$ such that $\tau^{k'}(c,p)=\tau^{k''}(c,p)$, this would mean that $\noeqbreak \tau^{k''-k'}(c,p)=(c,p)\in\hv\times[n]$, with $1\leq k''-k'<k_1$, which contradicts the definition of $k_1$. Therefore, the couples $\tau(c,p),\tau^{2}(c,p),...,\tau^{k_1-1}(c,p)$ are all different. By definition of $k_1$, these couples are all in the set $\hv\times\{n\}$, which has only two elements, so that $k_1\leq 3$.

Let us prove by finite induction that for every $k\in\{0,...,k_1-1\}$, we have\\ $(f \circ \pi_{0})^k \circ  f \circ \iota(\ket{c,p,y})=f^{k+1}(\ket{c,p,y})$. This is obviously true for $k=0$, and assuming that this is true for some $0\leq k<k_1-1$, we have $(f \circ \pi_{0})^{k+1} \circ  f \circ \iota(\ket{c,p,y})=f(\pi_{0}((f \circ \pi_{0})^{k} \circ  f \circ \iota(\ket{c,p,y})))=f(\pi_{0}(f^{k+1}(\ket{c,p,y})))$, and by definition of $k_1$, we have $f^{k+1}(\ket{c,p,y})\in\hv\times\{n\}$ so that $\pi_{0}(f^{k+1}(\ket{c,p,y}))=f^{k+1}(\ket{c,p,y})$, and consequently $(f \circ \pi_{0})^{k+1} \circ  f \circ \iota_j(\ket{c,p,y})=f^{k+2}(\ket{c,p,y})$. This finishes the induction.

Additionally, for any $k\in\N$, we have $f^k(\ket{c,p,y})=\ket{\tau^k(c,p)}\otimes U_{\tau^{k-1}(c,p)}\cdots U_{c,p}\ket{y}$.

For any $k<k_1-1$, by definition of $k_1$, we have $\tau^{k+1}(c,p)\in\hv\times\{n\}$ so that $\pi_1(f^{k+1}(\ket{c,p,y}))=0$, that is, the term of index $k$ of the series is zero.

We have $\tau^{k_1}(c,p)\in\hv\times[n]$, so that the term of index $k_1-1$ of the series is not zero unless $U_{\tau^{k_1-1}(c,p)}\cdots U_{c,p}\ket{y}=0$, and this term is equal to $\pi_1(f^{k_1}(\ket{c,p,y}))=\ket{\tau^{k_1}(c,p)}\otimes U_{\tau^{k_1-1}(c,p)}\cdots U_{c,p}\ket{y}$.

For any $k\geq k_1$, we have $(f \circ \pi_{0})^k \circ  f \circ \iota(\ket{c,p,y})=(f \circ \pi_{0})^{k-k_1}\circ f(\pi_0(f^{k_1}(\ket{c,p,y})))$, and since $\tau^{k_1}(c,p)\in\hv\times[n]$, we have $\pi_0(f^{k_1}(\ket{c,p,y}))=0$, so that the term of index $k$ of the series is zero.
\end{proof}

\begin{proof}[Proof of \cref{stabilitedeSparT}]\vspace{-0.16cm}
Well-definedness is a direct consequence of \cref{calculdutermenonnulavecpermutation}. Given $f\in\S_{n+1}$, by \cref{calculdutermenonnulavecpermutation} there exist a family of matrices $V_{c,p}$ such that $\T(f)=\ket{c,p,y}\mapsto\ket{\tau^*(c,p)}\otimes V_{c,p}\ket{y}$, where $\tau^*:\hv\times[n]\to\hv\times[n]::(c,p)\mapsto\tau^{k_1}(c,p)$ with $k_1$ the smallest $k\geq 1$ such that $\tau^k(c,p)\in\hv\times[n]$. What we have to prove is that $\tau^*$ is a permutation, that is, that it is a bijection. 

We claim that this is the case and that its inverse is $(\tau^{-1})^*:\hv\times[n]\to\hv\times[n]::(c,p)\mapsto(\tau^{-1})^{k_2}(c,p)$ with $k_2$ the smallest $k\geq 1$ such that $(\tau^{-1})^k(c,p)\in\hv\times[n]$.

Indeed, let $(c,p)\in\hv\times[n]$ and $k_1$ be the smallest $k\geq 1$ such that $\tau^k(c,p)\in\hv\times[n]$. Then for any $k\in\{1,...,k_1-1\}$, we have $(\tau^{-1})^k(\tau^*(c,p))=(\tau^{-1})^k(\tau^{k_1}(c,p))=\tau^{k_1-k}(c,p)$, which, by definition of $k_1$, is not in $\hv\times[n]$ because $1\leq k_1-k< k_1$. We also have $(\tau^{-1})^{k_1}(\tau^*(c,p))=(\tau^{-1})^{k_1}(\tau^{k_1}(c,p))=(c,p)$, which is in $\hv\times[n]$. Therefore, the smallest $k\geq 1$ such that $(\tau^{-1})^{k}(\tau^*(c,p))\in\hv\times[n]$ is $k_1$, so that $(\tau^{-1})^{*}(\tau^*(c,p))=(\tau^{-1})^{k_1}(\tau^{k_1}(c,p))=(c,p)$. This proves that $(\tau^{-1})^*\circ\tau^*=\id$. We can prove in the same way that $\tau^*\circ(\tau^{-1})^*=\id$, which proves our claim.
\end{proof}

\paragraph{Proof of \cref{welldefinednessofdenotationalsemantics}}

First, we do not assume the axioms of traced PROP and we prove by structural induction that for any diagram $D:n\to n$, $\interp D$ is well-defined and in $\S_n$.

If $D=\begin{tikzpicture}[scale=0.3]
	\begin{pgfonlayer}{nodelayer}
		\node [style=xsdiagrammevide] (0) at (0, 0) {};
	\end{pgfonlayer}
\end{tikzpicture}
,\begin{tikzpicture}
	\begin{pgfonlayer}{nodelayer}
		\node [style=none] (0) at (-0.3, 0) {};
		\node [style=none] (1) at (0.3, 0) {};
	\end{pgfonlayer}
	\begin{pgfonlayer}{edgelayer}
		\draw [style=noire] (0.center) to (1.center);
	\end{pgfonlayer}
\end{tikzpicture},\begin{tikzpicture}[scale=0.8]
	\begin{pgfonlayer}{nodelayer}
		\node [style=cartouche] (0) at (0, 0) {\tiny$\neg$};
		\node [style=none] (1) at (-0.8, 0) {};
		\node [style=none] (2) at (0.8, 0) {};
	\end{pgfonlayer}
	\begin{pgfonlayer}{edgelayer}
		\draw [style=noire] (1.center) to (0);
		\draw [style=noire] (0) to (2.center);
	\end{pgfonlayer}
\end{tikzpicture}
,\input{swap-xs.tikz}\text{ or }\input{beamsplitter-xs.tikz}$, then this is a direct consequence of the definition of $\interp.$.

If $D=D_2\circ D_1$, then by induction hypothesis, $\interp{D_1}$ and $\interp{D_2}$ are well-defined and in $\S_n$. By definition we have $\interp{D}=\interp{D_2}\circ\interp{D_1}$, and it is easy to see that $\S_n$ is closed under composition.

If $D=D_1\otimes D_2$, with $D_1:n_1\to n_1$, then $D_2:n-n_1\to n-n_1$ and by induction hypothesis, $\interp{D_1}$ and $\interp{D_2}$ are well-defined and we have $\interp{D_1}\in\S_{n_1}$ and $\interp{D_2}\in\S_{n-n_1}$. It is easy to see that for any $f\in\S_m$ and $g\in\S_k$ we have $f\boxplus g\in\S_{m+k}$, so that $\interp{D}:=\interp{D_1}\boxplus\interp{D_2}\in\S_n$.

If $D=Tr(D')$, then by induction hypothesis, $\interp{D'}$ is well-defined and in $\S_{n+1}$. By \cref{stabilitedeSparT} this implies that, $\interp{D}:=\T(\interp{D'})$ is well-defined and in $\S_n$.

The last thing to prove is (still not assuming the axioms of traced PROP) that two diagrams that are equivalent modulo the axioms of traced PROP have the same denotational semantics. For this it suffices to remark that the proof of \cref{adequacy} does not need to assume the axioms of traced PROP, so \cref{adequacy} still holds if we do not assume them. Then, since two diagrams equivalent modulo these axioms have the same path semantics (see \cref{preuvedeterminismtermreversibility}), that is, the same permutation $\tau_D$ and matrices $[D]_{c,p}$, by \eqref{adequacy} they have the same denotational semantics.

\paragraph{Proof of \cref{adequacy}}

We proceed by structural induction on $D$.
\begin{itemize}
\item If $D=\begin{tikzpicture}
	\begin{pgfonlayer}{nodelayer}
		\node [style=none] (0) at (-0.3, 0) {};
		\node [style=none] (1) at (0.3, 0) {};
	\end{pgfonlayer}
	\begin{pgfonlayer}{edgelayer}
		\draw [style=noire] (0.center) to (1.center);
	\end{pgfonlayer}
\end{tikzpicture}$, then we have $\tau_D=\id$, $[D]_{c,p}=I_q$ for every $c,p$, and $\interp{D}=\ket{c,0,x}\mapsto\ket{c,0,x}$, so the result holds.
\item If $D=\begin{tikzpicture}[scale=0.8]
	\begin{pgfonlayer}{nodelayer}
		\node [style=cartouche] (0) at (0, 0) {\footnotesize$\neg$};
		\node [style=none] (1) at (-0.8, 0) {};
		\node [style=none] (2) at (0.8, 0) {};
	\end{pgfonlayer}
	\begin{pgfonlayer}{edgelayer}
		\draw [style=noire] (1.center) to (0);
		\draw [style=noire] (0) to (2.center);
	\end{pgfonlayer}
\end{tikzpicture}
$, then we have $\tau_D=\begin{array}{l}(\rightarrow,p)\mapsto(\uparrowlarge,p)\\(\uparrowlarge,p)\mapsto(\rightarrow,p)\end{array}$, $[D]_{c,p}=I_q$ for every $c,p$, and $\interp{D}=\begin{array}{l}\ket{\rightarrow,p,y}\mapsto\ket{\uparrowlarge,p,y}\\\ket{\uparrowlarge,p,y}\mapsto\ket{\rightarrow,p,y}\end{array}$, so the result holds.
\item If $D=\input{swap-s.tikz}$, then we have $\tau_D=(c,p)\mapsto(c,1-p)$, $[D]_{c,p}=I_q$ for every $c,p$, and $\interp{D}=\ket{c,p,y}\mapsto\ket{c,1-p,y}$, so the result holds.
\item If $D=\input{beamsplitter-s.tikz}$, then we have $\tau_D=\begin{array}{l}(\rightarrow,p)\mapsto(\rightarrow,p)\\(\uparrowlarge,p)\mapsto(\uparrowlarge,1-p)\end{array}$, $[D]_{c,p}=I_q$ for every $c,p$, and $\interp{D}=\begin{array}{l}\ket{\rightarrow,p,y}\mapsto\ket{\rightarrow,p,y}\\\ket{\uparrowlarge,p,y}\mapsto\ket{\uparrowlarge,1-p,y}\end{array}$, so the result holds.
\item If $D=\begin{tikzpicture}[scale=0.8]
	\begin{pgfonlayer}{nodelayer}
		\node [style=newe] (0) at (0, 0) {\footnotesize$U$};
		\node [style=none] (1) at (-1, 0) {};
		\node [style=none] (2) at (1, 0) {};
	\end{pgfonlayer}
	\begin{pgfonlayer}{edgelayer}
		\draw [style=new] (1.center) to (2.center);
	\end{pgfonlayer}
\end{tikzpicture}$, then we have $\tau_D=\id$, $[D]_{c,p}=U$ for every $c,p$, and $\interp{D}=\ket{c,p,y}\mapsto \ket{c,p}\otimes U\ket{y}$, so the result holds.
\item If $D=D_2\circ D_1$, then on the one hand, for any $(c,p)\in\hv\times[n]$, we have\linebreak $(D_1,c,p)\xRightarrow{[D_1]_{c,p}}\tau_{D_1}(c,p)$ and  $(D_2,\tau_{D_1}(c,p))\xRightarrow{[D_2]_{\tau_{D_1}(c,p)}}\tau_{D_2}(\tau_{D_1}(c,p))$, so by Rule $(\circ)$ we have $(D,c,p)\xRightarrow{[D_2]_{\tau_{D_1}(c,p)}[D_1]_{c,p}}\tau_{D_2}(\tau_{D_1}(c,p))$, so that $\tau_D=\tau_{D_2}\circ\tau_{D_1}$ and $[D]_{c,p}=[D_2]_{\tau_{D_1}(c,p)}[D_1]_{c,p}$.
On the other hand, by induction hypothesis, we have $\interp{D_1}=\ket{c,p,y}\mapsto\ket{\tau_{D_1}(c,p)}\otimes [D_1]_{c,p}\ket{y}$ and $\interp{D_2}=\ket{c,p,y}\mapsto\ket{\tau_{D_2}(c,p)}\otimes [D_2]_{c,p}\ket{y}$.
Therefore, for any $(c,p,y)\in\hv\times[n]\times[q]$ we have $\interp D(\ket{c,p,y})=\interp{D_2}(\interp{D_1}(\ket{c,p,y}))=\interp{D_2}(\ket{\tau_{D_1}(c,p)}\otimes[D_1]_{c,p}\ket{y}=\ket{\tau_{D_2}(\tau_{D_1}(c,p))}\otimes[D_2]_{\tau_{D_1}(c,p)}[D_1]_{c,p}\ket{y}$.
So the result holds for $D$.
\item If $D=D_1\otimes D_2 $ with $D_1:n_1\to n_1$, then on the one hand, we have\\ $\tau_D=(c,p)\mapsto\begin{cases}\tau_{D_1}(c,p)&\text{if $p<n_1$}\\(c',p'+n_1)&\text{if $p\geq n_1$, where $(c',p')=\tau_{D_2}(c,p-n_1)$}\end{cases}$ and for any\linebreak $(c,p)\in\hv\times[n]$ we have $[D]_{c,p}=\begin{cases}[D_1]_{c,p}&\text{if $p<n_1$}\\ [D_2]_{c,p-n_1}&\text{if $p\geq n_1$}\end{cases}$. On the other hand, by induction hypothesis, we have $\interp{D_1}=\ket{c,p,y}\mapsto\ket{\tau_{D_1}(c,p)}\otimes U_{c,p}\ket{y}$ and $\interp{D_2}=\ket{c,p,y}\mapsto\ket{\tau_{D_2}(c,p)}\otimes [D_2]_{c,p}\ket{y}$,\\
so that $\interp{D}=\interp{D_1}\boxplus\interp{D_2}=\ket{c,p,y}\mapsto\begin{cases}\ket{\tau_{D_1}(c,p)}\otimes [D_1]_{c,p}\ket{y}&\text{if $p<n_1$}\\\ket{c',p'+n_1}\otimes [D_2]_{c,p-n_1}\ket{y}&\text{if $p\geq n_1$}\end{cases}$ where $(c',p')=\tau_{D_2}(c,p-n_1)$. So the result holds for $D$.
\item If $D=Tr(D')$, let $(c,p,y)\in\hv\times[n]\times[q]$, and let $k_1$ be the smallest $k\geq 1$ such that $\tau_{D'}^k(c,p)\in\hv\times[n]$. On the one hand, if we write $\tau_{D'}^k(c,p)$ as $(c_{k},p_k)$, then for all $i\in\{0,...,k_1-1\}$ we have $(D',c_i,p_i)\xRightarrow{[D']_{c_i,p_i}}(c_{i+1},p_{i+1})$, and by definition of $k_1$, we have $\tau_{D'}^{i+1}(c,p)\notin\hv\times[n]$, that is, $p_{i+1}=n$, if and only if $i<k_1$. Therefore, by Rule $(\mathsf T_{k_1})$, we have $(Tr(D'),c,p)\xRightarrow{[D']_{\tau_{D'}^{k_1-1}(c,p)}\cdots[D']_{c,p}}(\tau_{D'}(c,p))$. On the other hand, by induction hypothesis we have $\interp D'=\ket{c,p,y}\mapsto\ket{\tau_{D'}(c,p)}\otimes[D']_{c,p}\ket{y}$. By \cref{calculdutermenonnulavecpermutation}, this implies that $\interp D(\ket{c,p,y})=\T(\interp{D'})(\ket{c,p,y})=\ket{\tau_{D'}^{k_1}(c,p)}\otimes [D']_{\tau_{D'}^{k_1-1}(c,p)}\cdots[D']_{c,p}\ket{y}$. So the result holds for $D$.
\end{itemize}

\subsection{Equational theory -- PBS-calculus }

\subsubsection{Proof of Proposition \ref{soundnessoftheequations}}\label{proofsoundness}

\begin{definition}\label{defcongruence}
A \emph{congruence} is an equivalence relation $\mathrel{\mathcal R}$ on the set of diagrams such that if $D_1\mathrel{\mathcal R}D_1'$ and $D_2\mathrel{\mathcal R}D_2'$ then $(D_2\circ D_1)\mathrel{\mathcal R}(D_2'\circ D_1')$ and $(D_1\otimes D_2)\mathrel{\mathcal R} (D_1'\otimes D_2')$, and if $D\mathrel{\mathcal R}D'$ then $Tr(D)\mathrel{\mathcal R}Tr(D')$.
\end{definition}

Let $\sim$ be the relation such that $D_1\sim D_2$ if and only if $\interp{D_1}=\interp{D_2}$ and $\approx$ be the relation such that $D_1\approx D_2$ if and only if $\textup{PBS}\vdash D_1=D_2$. By definition, $\approx$ is the smallest congruence preserving Equations \eqref{idbox} to \eqref{bsnbsh}. It is clear that $\sim$ is a congruence, so it suffices to prove that it preserves Equations \eqref{idbox} to \eqref{bsnbsh} too. This can be done easily by using the graphical way to compute the denotational semantics provided by \cref{adequacy}.

\subsubsection{Normal forms}

\paragraph{Proof of \cref{NFcomposition}}\label{proofcompleteness_comp}

We have to show the following result:

\begin{lemma}\label{NFajout}
For any diagram $N:n\to n$ in normal form and any diagram $g$ of the form $(\begin{tikzpicture}
	\begin{pgfonlayer}{nodelayer}
		\node [style=none] (0) at (-0.3, 0) {};
		\node [style=none] (1) at (0.3, 0) {};
	\end{pgfonlayer}
	\begin{pgfonlayer}{edgelayer}
		\draw [style=noire] (0.center) to (1.center);
	\end{pgfonlayer}
\end{tikzpicture}^{\otimes i})\otimes h\otimes (\begin{tikzpicture}
	\begin{pgfonlayer}{nodelayer}
		\node [style=none] (0) at (-0.3, 0) {};
		\node [style=none] (1) at (0.3, 0) {};
	\end{pgfonlayer}
	\begin{pgfonlayer}{edgelayer}
		\draw [style=noire] (0.center) to (1.center);
	\end{pgfonlayer}
\end{tikzpicture}^{\otimes n-i-1})$ with $h=\begin{tikzpicture}
	\begin{pgfonlayer}{nodelayer}
		\node [style=none] (0) at (-0.3, 0) {};
		\node [style=none] (1) at (0.3, 0) {};
	\end{pgfonlayer}
	\begin{pgfonlayer}{edgelayer}
		\draw [style=noire] (0.center) to (1.center);
	\end{pgfonlayer}
\end{tikzpicture},\begin{tikzpicture}[scale=0.8]
	\begin{pgfonlayer}{nodelayer}
		\node [style=cartouche] (0) at (0, 0) {\tiny$\neg$};
		\node [style=none] (1) at (-0.8, 0) {};
		\node [style=none] (2) at (0.8, 0) {};
	\end{pgfonlayer}
	\begin{pgfonlayer}{edgelayer}
		\draw [style=noire] (1.center) to (0);
		\draw [style=noire] (0) to (2.center);
	\end{pgfonlayer}
\end{tikzpicture}
\text{ or }E(U,V)$, or $(\begin{tikzpicture}
	\begin{pgfonlayer}{nodelayer}
		\node [style=none] (0) at (-0.3, 0) {};
		\node [style=none] (1) at (0.3, 0) {};
	\end{pgfonlayer}
	\begin{pgfonlayer}{edgelayer}
		\draw [style=noire] (0.center) to (1.center);
	\end{pgfonlayer}
\end{tikzpicture}^{\otimes i})\otimes h\otimes (\begin{tikzpicture}
	\begin{pgfonlayer}{nodelayer}
		\node [style=none] (0) at (-0.3, 0) {};
		\node [style=none] (1) at (0.3, 0) {};
	\end{pgfonlayer}
	\begin{pgfonlayer}{edgelayer}
		\draw [style=noire] (0.center) to (1.center);
	\end{pgfonlayer}
\end{tikzpicture}^{\otimes n-i-2})$ with $h=\input{swap-xs.tikz}\text{ or }\input{beamsplitter-xs.tikz}$, there exists $N'$ in normal form such that $\textup{PBS}\vdash g\circ N=N'$.
\end{lemma}

We proceed by induction on $n$.

If $n=0$, then there is no such $g$ so the result trivially holds.

If $n\geq 1$, we write $N$ in the form
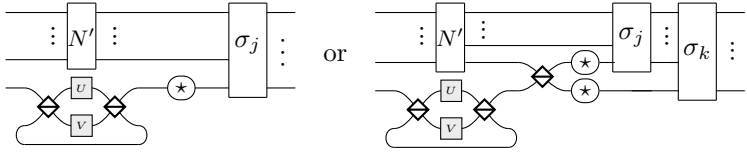
\[\input{typeAfiltreUV.tikz}\quad\text{or}\quad\input{typeCfiltreUV.tikz}.\]
We call these two forms type A and B respectively.

By induction hypothesis we only have to prove that $g\circ N$ can be put in the form
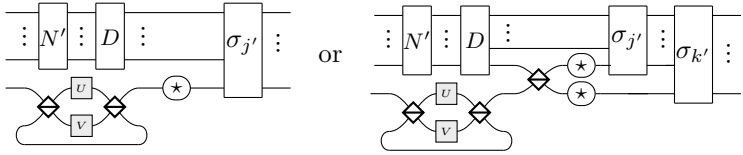
\[\input{typeAfiltreUVinductionD.tikz}\quad\text{or}\quad\input{typeCfiltreUVinductionD.tikz}\]
for some diagram $D:n-1\to n-1$ built using\begin{tikzpicture}[scale=0.3]
	\begin{pgfonlayer}{nodelayer}
		\node [style=xsdiagrammevide] (0) at (0, 0) {};
	\end{pgfonlayer}
\end{tikzpicture}
,\begin{tikzpicture}
	\begin{pgfonlayer}{nodelayer}
		\node [style=none] (0) at (-0.3, 0) {};
		\node [style=none] (1) at (0.3, 0) {};
	\end{pgfonlayer}
	\begin{pgfonlayer}{edgelayer}
		\draw [style=noire] (0.center) to (1.center);
	\end{pgfonlayer}
\end{tikzpicture},\begin{tikzpicture}[scale=0.8]
	\begin{pgfonlayer}{nodelayer}
		\node [style=cartouche] (0) at (0, 0) {\tiny$\neg$};
		\node [style=none] (1) at (-0.8, 0) {};
		\node [style=none] (2) at (0.8, 0) {};
	\end{pgfonlayer}
	\begin{pgfonlayer}{edgelayer}
		\draw [style=noire] (1.center) to (0);
		\draw [style=noire] (0) to (2.center);
	\end{pgfonlayer}
\end{tikzpicture}
,\input{swap-xs.tikz},\input{beamsplitter-xs.tikz}, $E(U',V')$, $\circ$ and $\otimes$.

To prove this, we proceed by case distinction:

\begin{itemize}

\item If $h=\begin{tikzpicture}
	\begin{pgfonlayer}{nodelayer}
		\node [style=none] (0) at (-0.3, 0) {};
		\node [style=none] (1) at (0.3, 0) {};
	\end{pgfonlayer}
	\begin{pgfonlayer}{edgelayer}
		\draw [style=noire] (0.center) to (1.center);
	\end{pgfonlayer}
\end{tikzpicture}$, then $g\circ N=N$, so there is nothing to do.

\item If $h=\begin{tikzpicture}[scale=0.8]
	\begin{pgfonlayer}{nodelayer}
		\node [style=cartouche] (0) at (0, 0) {\tiny$\neg$};
		\node [style=none] (1) at (-0.8, 0) {};
		\node [style=none] (2) at (0.8, 0) {};
	\end{pgfonlayer}
	\begin{pgfonlayer}{edgelayer}
		\draw [style=noire] (1.center) to (0);
		\draw [style=noire] (0) to (2.center);
	\end{pgfonlayer}
\end{tikzpicture}
$, then we slide it through $\sigma_j$ ($\sigma_k$ and $\sigma_j$ if $N$ is of type B),
\begin{itemize}
\item if it does not arrive on the last wire if $N$ is of type A, or one of the last two wires if $N$ is of type B, then we get the desired form with $D=(\begin{tikzpicture}
	\begin{pgfonlayer}{nodelayer}
		\node [style=none] (0) at (-0.3, 0) {};
		\node [style=none] (1) at (0.3, 0) {};
	\end{pgfonlayer}
	\begin{pgfonlayer}{edgelayer}
		\draw [style=noire] (0.center) to (1.center);
	\end{pgfonlayer}
\end{tikzpicture}^{\otimes i'})\otimes \begin{tikzpicture}[scale=0.8]
	\begin{pgfonlayer}{nodelayer}
		\node [style=cartouche] (0) at (0, 0) {\tiny$\neg$};
		\node [style=none] (1) at (-0.8, 0) {};
		\node [style=none] (2) at (0.8, 0) {};
	\end{pgfonlayer}
	\begin{pgfonlayer}{edgelayer}
		\draw [style=noire] (1.center) to (0);
		\draw [style=noire] (0) to (2.center);
	\end{pgfonlayer}
\end{tikzpicture}
\otimes (\begin{tikzpicture}
	\begin{pgfonlayer}{nodelayer}
		\node [style=none] (0) at (-0.3, 0) {};
		\node [style=none] (1) at (0.3, 0) {};
	\end{pgfonlayer}
	\begin{pgfonlayer}{edgelayer}
		\draw [style=noire] (0.center) to (1.center);
	\end{pgfonlayer}
\end{tikzpicture}^{\otimes n-i'-2})$
\item if it arrives on the last wire (resp. on one of the last two wires), then it merges with the\begin{tikzpicture}[scale=0.6]
	\begin{pgfonlayer}{nodelayer}
		\node [style=none] (1) at (-1.2, 0) {};
		\node [style=none] (2) at (1.2, 0) {};
		\node [style=snegpotentiel] (4) at (0, 0) {};
	\end{pgfonlayer} 
	\begin{pgfonlayer}{edgelayer}
		\draw [style=noire] (4) to (2.center);
		\draw [style=noire] (1.center) to (4);
	\end{pgfonlayer}
\end{tikzpicture}
 on its wire and changes its value: if\begin{tikzpicture}[scale=0.6]
	\begin{pgfonlayer}{nodelayer}
		\node [style=none] (1) at (-1.2, 0) {};
		\node [style=none] (2) at (1.2, 0) {};
		\node [style=snegpotentiel] (4) at (0, 0) {};
	\end{pgfonlayer} 
	\begin{pgfonlayer}{edgelayer}
		\draw [style=noire] (4) to (2.center);
		\draw [style=noire] (1.center) to (4);
	\end{pgfonlayer}
\end{tikzpicture}
 is\begin{tikzpicture}
	\begin{pgfonlayer}{nodelayer}
		\node [style=none] (0) at (-0.3, 0) {};
		\node [style=none] (1) at (0.3, 0) {};
	\end{pgfonlayer}
	\begin{pgfonlayer}{edgelayer}
		\draw [style=noire] (0.center) to (1.center);
	\end{pgfonlayer}
\end{tikzpicture} then $h$ simply takes its place, and if\begin{tikzpicture}[scale=0.6]
	\begin{pgfonlayer}{nodelayer}
		\node [style=none] (1) at (-1.2, 0) {};
		\node [style=none] (2) at (1.2, 0) {};
		\node [style=snegpotentiel] (4) at (0, 0) {};
	\end{pgfonlayer} 
	\begin{pgfonlayer}{edgelayer}
		\draw [style=noire] (4) to (2.center);
		\draw [style=noire] (1.center) to (4);
	\end{pgfonlayer}
\end{tikzpicture}
 is\begin{tikzpicture}[scale=0.8]
	\begin{pgfonlayer}{nodelayer}
		\node [style=cartouche] (0) at (0, 0) {\tiny$\neg$};
		\node [style=none] (1) at (-0.8, 0) {};
		\node [style=none] (2) at (0.8, 0) {};
	\end{pgfonlayer}
	\begin{pgfonlayer}{edgelayer}
		\draw [style=noire] (1.center) to (0);
		\draw [style=noire] (0) to (2.center);
	\end{pgfonlayer}
\end{tikzpicture}
 then the two negations cancel by the following equation, which is proved in \cref{derivationnegneg} to be a consequence of the axioms of the $\textup{PBS}$-calculus:
\stepcounter{equation}
\begin{equation}\tag{\ref{negneg}}{\tikzset{tikzfig/.style={baseline=-0.25em,scale=\echellefils,every node/.style={scale=0.8}}}\begin{array}{rcl}\input{negneg.tikz}&=&\begin{tikzpicture}
	\begin{pgfonlayer}{nodelayer}
		\node [style=none] (0) at (-2, 0) {};
		\node [style=none] (1) at (2, 0) {};
	\end{pgfonlayer}
	\begin{pgfonlayer}{edgelayer}
		\draw [style=new] (0.center) to (1.center);
	\end{pgfonlayer}
\end{tikzpicture}\end{array}}\end{equation}
\end{itemize}

\item If $h=E(U',V')$, then we slide it through $\sigma_j$ ($\sigma_k$ and $\sigma_j$ if $N$ is of type B),
\begin{itemize}
\item if it does not arrive on the last wire if $N$ is of type A, or one of the last two wires if $N$ is of type B, then we get the desired form with $D=(\begin{tikzpicture}
	\begin{pgfonlayer}{nodelayer}
		\node [style=none] (0) at (-0.3, 0) {};
		\node [style=none] (1) at (0.3, 0) {};
	\end{pgfonlayer}
	\begin{pgfonlayer}{edgelayer}
		\draw [style=noire] (0.center) to (1.center);
	\end{pgfonlayer}
\end{tikzpicture}^{\otimes i'})\otimes h\otimes (\begin{tikzpicture}
	\begin{pgfonlayer}{nodelayer}
		\node [style=none] (0) at (-0.3, 0) {};
		\node [style=none] (1) at (0.3, 0) {};
	\end{pgfonlayer}
	\begin{pgfonlayer}{edgelayer}
		\draw [style=noire] (0.center) to (1.center);
	\end{pgfonlayer}
\end{tikzpicture}^{\otimes n-i'-2})$
\item if it arrives on the last wire (resp. on one of the last two wires), then it commutes with the\begin{tikzpicture}[scale=0.6]
	\begin{pgfonlayer}{nodelayer}
		\node [style=none] (1) at (-1.2, 0) {};
		\node [style=none] (2) at (1.2, 0) {};
		\node [style=snegpotentiel] (4) at (0, 0) {};
	\end{pgfonlayer} 
	\begin{pgfonlayer}{edgelayer}
		\draw [style=noire] (4) to (2.center);
		\draw [style=noire] (1.center) to (4);
	\end{pgfonlayer}
\end{tikzpicture}
 on its wire, trivially if\begin{tikzpicture}[scale=0.6]
	\begin{pgfonlayer}{nodelayer}
		\node [style=none] (1) at (-1.2, 0) {};
		\node [style=none] (2) at (1.2, 0) {};
		\node [style=snegpotentiel] (4) at (0, 0) {};
	\end{pgfonlayer} 
	\begin{pgfonlayer}{edgelayer}
		\draw [style=noire] (4) to (2.center);
		\draw [style=noire] (1.center) to (4);
	\end{pgfonlayer}
\end{tikzpicture}
 is\begin{tikzpicture}
	\begin{pgfonlayer}{nodelayer}
		\node [style=none] (0) at (-0.3, 0) {};
		\node [style=none] (1) at (0.3, 0) {};
	\end{pgfonlayer}
	\begin{pgfonlayer}{edgelayer}
		\draw [style=noire] (0.center) to (1.center);
	\end{pgfonlayer}
\end{tikzpicture}, and by the following equation, that we prove below to be a consequence of the axioms of the $\textup{PBS}$-calculus, if\begin{tikzpicture}[scale=0.6]
	\begin{pgfonlayer}{nodelayer}
		\node [style=none] (1) at (-1.2, 0) {};
		\node [style=none] (2) at (1.2, 0) {};
		\node [style=snegpotentiel] (4) at (0, 0) {};
	\end{pgfonlayer} 
	\begin{pgfonlayer}{edgelayer}
		\draw [style=noire] (4) to (2.center);
		\draw [style=noire] (1.center) to (4);
	\end{pgfonlayer}
\end{tikzpicture}
 is\begin{tikzpicture}[scale=0.8]
	\begin{pgfonlayer}{nodelayer}
		\node [style=cartouche] (0) at (0, 0) {\tiny$\neg$};
		\node [style=none] (1) at (-0.8, 0) {};
		\node [style=none] (2) at (0.8, 0) {};
	\end{pgfonlayer}
	\begin{pgfonlayer}{edgelayer}
		\draw [style=noire] (1.center) to (0);
		\draw [style=noire] (0) to (2.center);
	\end{pgfonlayer}
\end{tikzpicture}
:\footnote{In the equations, $U,V,U'\text{ and }V'$ stand for generic matrices, not necessarily related to the context.}
\begingroup
\tikzset{tikzfig/.style={baseline=-0.25em,scale=\echellefils,every node/.style={scale=0.8}}}
\begin{equation}\labeletpreuve{neguv}{\begin{array}{rcl}
\input{negfiltreUV.tikz}&=&\input{filtreVUneg.tikz}
\end{array}}\end{equation}
then, if $N$ is of type B, it passes through the beam splitter by one of the following two equations:
\begin{equation}\labeletpreuve{bsuvh}{\begin{array}{rcl}
\input{grandbsfiltreUVhaut.tikz}&=&\input{filtreUIhautfiltreIVbasgrandbs.tikz}
\end{array}}\end{equation}
\begin{equation}\labeletpreuve{bsuvb}{\begin{array}{rcl}
\input{grandbsfiltreUVbas.tikz}&=&\input{filtreIVhautfiltreUIbasgrandbs.tikz}
\end{array}}\end{equation}
finally, the top part becomes part of $D$, and the bottom part merges with the $E(U,V)$ from $N$ by the following equation:
\begin{equation}\labeletpreuve{fusionuv}{\begin{array}{rcl}
\input{filtreUVfiltreUprimeVprime.tikz}&=&\input{filtreUprimeUVprimeV.tikz}
\end{array}}\end{equation}
\end{itemize}
\endgroup

\item If $h=\input{swap-xs.tikz}$, then by manipulating the wires according to the axioms of traced PROP, we can write $g\circ N$ in one of the desired forms, with $D$ being a permutation of the wires (that is, a composition of tensor products of\input{swap-xs.tikz} and\begin{tikzpicture}
	\begin{pgfonlayer}{nodelayer}
		\node [style=none] (0) at (-0.3, 0) {};
		\node [style=none] (1) at (0.3, 0) {};
	\end{pgfonlayer}
	\begin{pgfonlayer}{edgelayer}
		\draw [style=noire] (0.center) to (1.center);
	\end{pgfonlayer}
\end{tikzpicture}).

\item If $h=\input{beamsplitter-xs.tikz}$ then we look at the indices $i_1$ and $i_2$ of the wires to which $h$ is connected on the other side of $\sigma_j$ (on the other side of $\sigma_k\circ(\sigma_j\otimes\begin{tikzpicture}
	\begin{pgfonlayer}{nodelayer}
		\node [style=none] (0) at (-0.3, 0) {};
		\node [style=none] (1) at (0.3, 0) {};
	\end{pgfonlayer}
	\begin{pgfonlayer}{edgelayer}
		\draw [style=noire] (0.center) to (1.center);
	\end{pgfonlayer}
\end{tikzpicture})$ if $N$ is of type B). The wire $i_1$ is connected to the top wire of $h$ and the wire $i_2$ to the bottom wire of $h$.
\begin{itemize}
\item If $i_1,i_2<n-1$ in the case of type A ($i_1,i_2<n-2$ in the case of type B), then $i_2=i_1+1$ and we can slide the beam splitter across $\sigma_j$ ($\sigma_k$ and $\sigma_j$ in the case of type B) to put $N$ in the desired form with $D=(\begin{tikzpicture}
	\begin{pgfonlayer}{nodelayer}
		\node [style=none] (0) at (-0.3, 0) {};
		\node [style=none] (1) at (0.3, 0) {};
	\end{pgfonlayer}
	\begin{pgfonlayer}{edgelayer}
		\draw [style=noire] (0.center) to (1.center);
	\end{pgfonlayer}
\end{tikzpicture}^{\otimes i'})\otimes h\otimes (\begin{tikzpicture}
	\begin{pgfonlayer}{nodelayer}
		\node [style=none] (0) at (-0.3, 0) {};
		\node [style=none] (1) at (0.3, 0) {};
	\end{pgfonlayer}
	\begin{pgfonlayer}{edgelayer}
		\draw [style=noire] (0.center) to (1.center);
	\end{pgfonlayer}
\end{tikzpicture}^{\otimes n-i'-3})$.
\item If $N$ is of type A and $i_2=n-1$, then by manipulating the wires we can write $g\circ N$ in the form
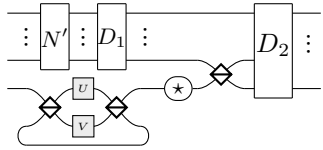
\[\input{typeAfiltreUVinductionD1bs.tikz}\]
where $D_1$ and $D_2$ are permutations of the wires. Then, if \begin{tikzpicture}[scale=0.6]
	\begin{pgfonlayer}{nodelayer}
		\node [style=none] (1) at (-1.2, 0) {};
		\node [style=none] (2) at (1.2, 0) {};
		\node [style=snegpotentiel] (4) at (0, 0) {};
	\end{pgfonlayer} 
	\begin{pgfonlayer}{edgelayer}
		\draw [style=noire] (4) to (2.center);
		\draw [style=noire] (1.center) to (4);
	\end{pgfonlayer}
\end{tikzpicture}
 is \begin{tikzpicture}[scale=0.8]
	\begin{pgfonlayer}{nodelayer}
		\node [style=cartouche] (0) at (0, 0) {\tiny$\neg$};
		\node [style=none] (1) at (-0.8, 0) {};
		\node [style=none] (2) at (0.8, 0) {};
	\end{pgfonlayer}
	\begin{pgfonlayer}{edgelayer}
		\draw [style=noire] (1.center) to (0);
		\draw [style=noire] (0) to (2.center);
	\end{pgfonlayer}
\end{tikzpicture}
, we apply the following equation:
\begin{equation}\labeletpreuve{nbs}{\tikzset{tikzfig/.style={baseline=-0.25em,scale=\echellefils,every node/.style={scale=0.8}}}\begin{array}{rcl}
\input{nbbeamsplitter.tikz}&=&\input{beamsplitternnnx.tikz}
\end{array}}\end{equation}
and the\begin{tikzpicture}[scale=0.8]
	\begin{pgfonlayer}{nodelayer}
		\node [style=cartouche] (0) at (0, 0) {\tiny$\neg$};
		\node [style=none] (1) at (-0.8, 0) {};
		\node [style=none] (2) at (0.8, 0) {};
	\end{pgfonlayer}
	\begin{pgfonlayer}{edgelayer}
		\draw [style=noire] (1.center) to (0);
		\draw [style=noire] (0) to (2.center);
	\end{pgfonlayer}
\end{tikzpicture}
 on the left is composed with $D_1$ to give us $D$. Finally, we get the desired form by manipulation of the wires.
\item If $N$ is of type A and $i_1=n-1$, then by manipulating the wires, and applying once the following equation :
\begin{equation}\labeletpreuve{swapbsx}{\tikzset{tikzfig/.style={baseline=-0.25em,scale=\echellefils,every node/.style={scale=0.8}}}\begin{array}{rcl}
\input{swapbeamsplitter.tikz}&=&\input{beamsplitterswap.tikz}
\end{array}}\end{equation}
we can write $g\circ N$ in the form
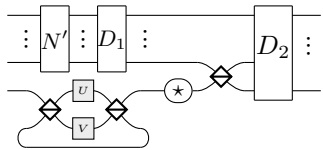
\[\input{typeAfiltreUVinductionD1bs.tikz}\]
where $D_1$ and $D_2$ are permutations of the wires. Then we proceed as in the previous case.
\item If $N$ is of type B, $i_1<n-2$ and $i_2=n-2$, then by manipulating the wires we can write $g\circ N$ in the form
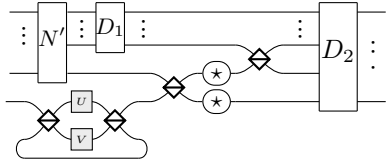
\[\input{typeCfiltreUVinductionD1bsnpbsesc.tikz}\]
where $D_1$ and $D_2$ are permutations of the wires. Then, according to the\!\begin{tikzpicture}[scale=0.6]
	\begin{pgfonlayer}{nodelayer}
		\node [style=none] (1) at (-1.2, 0) {};
		\node [style=none] (2) at (1.2, 0) {};
		\node [style=snegpotentiel] (4) at (0, 0) {};
	\end{pgfonlayer} 
	\begin{pgfonlayer}{edgelayer}
		\draw [style=noire] (4) to (2.center);
		\draw [style=noire] (1.center) to (4);
	\end{pgfonlayer}
\end{tikzpicture}
 between the two beam splitters, we use one of the following two equations:
\begin{equation}\labeletpreuve{swapbsbsvar}{\tikzset{tikzfig/.style={baseline=-0.25em,scale=\echellefils,every node/.style={scale=0.8}}}\begin{array}{rcl}\input{bsbsmonte.tikz}&=&\input{xbsbsxdescendhaut.tikz}\end{array}}\end{equation}
\begin{equation}\labeletpreuve{bsnbsesc}{\tikzset{tikzfig/.style={baseline=-0.25em,scale=\echellefils,every node/.style={scale=0.8}}}\begin{array}{rcl}\input{bsbnmbsh.tikz}&=&\input{nhbshnmxhbsbnm.tikz}\end{array}}\end{equation}
Immediately in the second case, or after a few manipulation of wires in the first case, we get the desired form.
\item If $N$ is of type B, $i_2<n-2$ and $i_1=n-2$, then by manipulating the wires and using once Equation \eqref{swapbsx}, we can write $g\circ N$ in the same form as in the previous case. Then we proceed in the same way.
\item If $N$ is of type B, $i_1<n-2$ and $i_2=n-1$, then by manipulating the wires we can write $g\circ N$ in the form
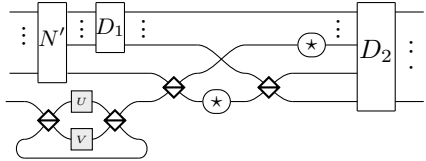
\[\input{typeCfiltreUVinductionD1bsxbs.tikz}\]
where $D_1$ and $D_2$ are permutations of the wires. Then if the \begin{tikzpicture}[scale=0.6]
	\begin{pgfonlayer}{nodelayer}
		\node [style=none] (1) at (-1.2, 0) {};
		\node [style=none] (2) at (1.2, 0) {};
		\node [style=snegpotentiel] (4) at (0, 0) {};
	\end{pgfonlayer} 
	\begin{pgfonlayer}{edgelayer}
		\draw [style=noire] (4) to (2.center);
		\draw [style=noire] (1.center) to (4);
	\end{pgfonlayer}
\end{tikzpicture}
 between the two beam splitters is\begin{tikzpicture}
	\begin{pgfonlayer}{nodelayer}
		\node [style=none] (0) at (-0.3, 0) {};
		\node [style=none] (1) at (0.3, 0) {};
	\end{pgfonlayer}
	\begin{pgfonlayer}{edgelayer}
		\draw [style=noire] (0.center) to (1.center);
	\end{pgfonlayer}
\end{tikzpicture}, then we apply the following equation:
\begin{equation}\labeletpreuve{bsxbs}{\tikzset{tikzfig/.style={baseline=-0.25em,scale=\echellefils,every node/.style={scale=0.8}}}\begin{array}{rcl}\input{bsxbspointeenhaut.tikz}&=&\input{bsbsxpointeenbas.tikz}\end{array}}\end{equation}
if the \begin{tikzpicture}[scale=0.6]
	\begin{pgfonlayer}{nodelayer}
		\node [style=none] (1) at (-1.2, 0) {};
		\node [style=none] (2) at (1.2, 0) {};
		\node [style=snegpotentiel] (4) at (0, 0) {};
	\end{pgfonlayer} 
	\begin{pgfonlayer}{edgelayer}
		\draw [style=noire] (4) to (2.center);
		\draw [style=noire] (1.center) to (4);
	\end{pgfonlayer}
\end{tikzpicture}
 between the two beam splitters is\begin{tikzpicture}[scale=0.8]
	\begin{pgfonlayer}{nodelayer}
		\node [style=cartouche] (0) at (0, 0) {\tiny$\neg$};
		\node [style=none] (1) at (-0.8, 0) {};
		\node [style=none] (2) at (0.8, 0) {};
	\end{pgfonlayer}
	\begin{pgfonlayer}{edgelayer}
		\draw [style=noire] (1.center) to (0);
		\draw [style=noire] (0) to (2.center);
	\end{pgfonlayer}
\end{tikzpicture}
, then we proceed as follows:
\begin{equation}\notag\tikzset{tikzfig/.style={baseline=-0.25em,scale=\echellefils,every node/.style={scale=0.8}}}\begin{array}{rcl}\input{bsnxbspointeenhaut.tikz}&\eqeqref{nbs}&\input{nhbsbxhbsbnmnbxb.tikz}\\\\
&\eqeqref{bsxbs}&\input{nhbshnhbsbxhnbxb.tikz}\end{array}\end{equation}
which gives us the desired form after some manipulation of wires.
\item If $N$ is of type B, $i_2<n-2$ and $i_1=n-1$, then by manipulating the wires and applying Equation \eqref{swapbsx} we write $g\circ N$ in the same form as in the previous case, and we proceed in the same way.
\item If $N$ is of type B, $i_1=n-2$ and $i_2=n-1$, then by manipulating the wires, we can write $g\circ N$ in the following form:
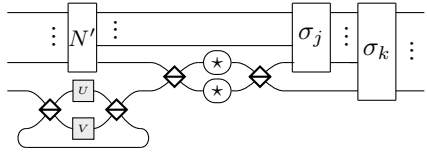
\[\input{typeCfiltreUVbsnpnpbs.tikz}\]
then we apply one of the following equations:
\begingroup
\tikzset{tikzfig/.style={baseline=-0.25em,scale=\echellefils,every node/.style={scale=0.8}}}
\begin{equation}\tag{\ref{bsbs}}\begin{array}{rcl}\hspace{0.35cm}\input{beamsplitterbeamsplitter.tikz}&=&\begin{tikzpicture}
	\begin{pgfonlayer}{nodelayer}
		\node [style=none] (0) at (-3, 1) {};
		\node [style=none] (1) at (3, 1) {};
		\node [style=none] (2) at (3, -1) {};
		\node [style=none] (3) at (-3, -1) {};
	\end{pgfonlayer}
	\begin{pgfonlayer}{edgelayer}
		\draw [style=new] (0.center) to (1.center);
		\draw [style=new] (3.center) to (2.center);
	\end{pgfonlayer}
\end{tikzpicture}\end{array}\end{equation}
\begin{equation}\tag{\ref{bsnbsh}}\begin{array}{rcl}\input{beamsplitternhautbeamsplitter.tikz}&=&\input{beamsplitternnhaut.tikz}\end{array}\end{equation}
\begin{equation}\labeletpreuve{bsnbsb}{\begin{array}{rcl}
\input{beamsplitternbasbeamsplitter.tikz}&=&\input{nhbsnbx.tikz}
\end{array}}\end{equation}
\begin{equation}\labeletpreuve{bsnnbs}{\begin{array}{rcl}
\input{bsnnbs.tikz}&=&\input{nnswap.tikz}
\end{array}}\end{equation}
\endgroup
In the four cases, this gives us the desired form, after a few manipulation of wires if necessary.
\item If $N$ is of type B, $i_1=n-1$ and $i_2=n-2$, then by manipulating the wires and applying Equation \eqref{swapbsx} once, we can write $g\circ N$ in the same form as in the previous case and proceed in the same way. This finishes the case distinction.
\end{itemize}
\end{itemize}

\begingroup
\tikzset{tikzfig/.style={baseline=-0.25em,scale=\echellefils,every node/.style={scale=0.8}}}
It remains to prove Equations \eqref{neguv} to \eqref{bsnnbs}.

\phantomsection\label{preuveswapbsx}To prove Equation \eqref{swapbsx}, we have, by Equation \eqref{bsnnnn}:
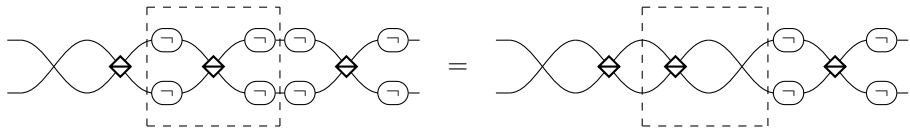
\begin{equation}\notag\label{preuveswapbsx}\begin{array}{rcl}\input{xbsnhnbbsnhnbnhnbbsnhnbboite.tikz}&=&\input{xbsbsxnhnbbsnhnbboite.tikz}\end{array}\end{equation}
by Equations \eqref{negneg} and \eqref{bsbs}, and inverse law, this implies that
\begin{equation}\notag\begin{array}{rcl}\input{swapbeamsplitter.tikz}&=&\input{beamsplitternnnn.tikz}\end{array}\end{equation}
which, together with Equation \eqref{bsnnnn}, implies Equation \eqref{swapbsx}.\\\\

\phantomsection\label{preuveneguv}To prove Equation \eqref{neguv}, we have:
\begin{longtable}{RCL}\!\!\!\!\!\!\!\!\!\!\!\input{negfiltreUV.tikz}&\eer{negneg}&\input{tracenbnhnbbsnhnbnhnbUhVbbsnhpuisnh.tikz}\\\\
&\eer{negu}&\input{tracenbnhnbbsnhnbUhVbnhnbbsnhpuisnh.tikz}\\\\
&\overset{\text{dinaturality}}{=}&\input{tracenhnbbsnhnbUhVbnhnbbsnhnbpuisnh.tikz}\\\\
&\eqeqref{bsnnnn}&\input{tracebsxUhVbbsxpuisnh.tikz}\\\\
&\eqeqref{swapbsx}&\input{tracebsxUhVbxbspuisnh.tikz}\\\\
&\overset{\text{naturality of}}{\overset{\text{the swap,}}{\eqexpl{\text{inverse law}}}}&\input{filtreVUneg.tikz}\end{longtable}

\phantomsection\label{preuvenbs}To prove Equation \eqref{nbs}, we have:
\begin{longtable}{RCL}\input{nbbeamsplitter.tikz}&\eqeqref{negneg}&\input{nhnhnbbsnhnbnhnb.tikz}\\\\
&\eqeqref{bsnnnn}&\input{nhbsnhnbx.tikz}\end{longtable}

\phantomsection\label{preuvebsnbsb}To prove \cref{bsnbsb}, we have:
\begin{longtable}{RCL}\input{beamsplitternbasbeamsplitter.tikz}&\eqeqref{nbs}&\input{bsnhbsnhnbx.tikz}\\\\
&\eqdeuxeqref{bsnbsh}{negneg}&\input{nhbsnbx.tikz}\end{longtable}

\phantomsection\label{preuvebsnnbs}To prove \cref{bsnnbs}, we have:
\begin{longtable}{RCL}
\input{bsnnbs.tikz}&\eqeqref{negneg}&\input{bsnhnbbsnhnbnhnb.tikz}\\\\
&\overset{\eqref{bsnnnn},}{\overset{\text{naturality of}}{\eqexpl{\text{the swap}}}}&\input{bsbsnhnbx.tikz}\\\\
&\eqeqref{bsbs}&\input{nnswap.tikz}\end{longtable}

\phantomsection\label{preuveswapbsbsvar}To prove \cref{swapbsbsvar}, we have:

\begin{longtable}{RCL}\input{bsbsmonte.tikz}&\overset{\text{inverse law},}{\eqeqref{swapbsx}}&\input{xhxhbsbxhbshxh.tikz}\\\\
&\eqeqref{bsbsbs}&\input{xhbshbsbbshbshxh.tikz}\\\\
&\eqeqref{bsbs}&\input{xbsbsxdescendhaut.tikz}\end{longtable}

\subparagraph{Ancillary equations}
To prove the remaining equations, we need some ancillary equations:

\begin{lemma}
The following equations are consequences of the axioms of the \textup{PBS}-calculus:

\begin{equation}\labeletpreuve{bsbsbsvar}{\begin{array}{rcl}
\input{bsbsbspointeenhaut.tikz}&=&\input{xbsxpointeenhaut.tikz}
\end{array}}\end{equation}

\begin{equation}\label{swapbsbsbvar}\begin{array}{rcl}
\input{bsbsmonte.tikz}&=&\input{xbsbsxdescendbas.tikz}
\end{array}\end{equation}
\end{lemma}

\begin{proof}\vspace{-0.16cm}
\phantomsection\label{preuvebsbsbsvar}To prove \cref{bsbsbsvar}, we have:
\begin{longtable}{RCL}\input{bsbsbspointeenhaut.tikz}&\overset{\text{inverse law},}{\overset{\text{naturality of}}{\overset{\text{the swap},\eqref{swapbsx}}{=}}}&\input{xbxhxbbshbsbbshxbxhxb.tikz}\\\\
&\eqeqref{bsbsbs}&\input{xbxhxbxhbsbxhxbxhxb.tikz}\\\\
&\overset{\text{naturality of}}{\overset{\text{the swap},\eqref{swapbsx},}{\eqexpl{\text{inverse law}}}}&\input{xbsxpointeenhaut.tikz}
\end{longtable}

The proof of \cref{swapbsbsbvar} is obtained by rotating the proof of \cref{swapbsbsvar} by 180° (it uses \cref{bsbsbsvar} instead of Equation \eqref{bsbsbs}).
\end{proof}

\phantomsection\label{preuvebsnbsesc}To prove \cref{bsnbsesc}, we have :
\begin{longtable}{RCL}\input{bsbnmbsh.tikz}&\eqeqref{nbs}&\input{nhbsbbshnhnmxh.tikz}\\\\
&\eqeqref{swapbsbsvar}&\input{nhxhbshbsbxhnhnmxh.tikz}\\\\
&\eqeqref{swapbsx}&\input{nhbshxhbsbxhnhnmxh.tikz}\\\\
&\overset{\text{naturality of}}{\overset{\text{the swap,}}{\eqexpl{\text{inverse law}}}}&\input{nhbshnmxhbsbnm.tikz}\end{longtable}

\phantomsection\label{preuvebsxbs}To prove \cref{bsxbs}, we have:
\begin{longtable}{RCL}\!\!\!\!\!\!\!\!\!\input{bsxbspointeenhaut.tikz}&\overset{\text{inverse law,}}{\eqeqref{swapbsx}}&\input{xbbsbxbxhbsb.tikz}\\\\
&\overset{\text{naturality of}}{\overset{\text{the swap}}{=}}&\input{xbbsbbshxbxh.tikz}\\\\
&\overset{\eqref{swapbsbsbvar},}{\eqexpl{\text{inverse law}}}&\input{bsbsxpointeenbas.tikz}\end{longtable}

\subparagraph{Ancillary equations}
To prove the remaining equations, we need additional ancillary equations:
\begin{lemma}
The following equations are consequences of the axioms of the $\textup{PBS}$-calculus:
\begin{equation}\labeletpreuve{bsxbsvar}{\begin{array}{rcl}
\input{bsxbspointeenbas.tikz}&=&\input{xbsbspointeenhaut.tikz}
\end{array}}\end{equation}
\begin{equation}\labeletpreuve{loopI}{\begin{array}{rcl}
\input{boucletraverseI.tikz}&=&\begin{tikzpicture}
	\begin{pgfonlayer}{nodelayer}
		\node [style=none] (0) at (-2, 0) {};
		\node [style=none] (1) at (2, 0) {};
	\end{pgfonlayer}
	\begin{pgfonlayer}{edgelayer}
		\draw [style=new] (0.center) to (1.center);
	\end{pgfonlayer}
\end{tikzpicture}
\end{array}}\end{equation}
\begin{equation}\labeletpreuve{splituv}{\begin{array}{rcl}
\input{bouclerebondisUboucletraverseV.tikz}&=&\input{filtreUV.tikz}
\end{array}}\end{equation}

\begin{equation}\labeletpreuve{rebondisbsb}{\begin{array}{rcl}
\input{bouclerebondisUbasgrandbs.tikz}&=&\input{grandbsbouclerebondisUbas.tikz}
\end{array}}\end{equation}

\begin{equation}\labeletpreuve{traversebshb}{\begin{array}{rcl}
\input{boucletraverseUhautgrandbs.tikz}&=&\input{grandbsboucletraverseUbas.tikz}
\end{array}}\end{equation}

\begin{equation}\labeletpreuve{rebondisbsh}{\begin{array}{rcl}
\input{bouclerebondisUhautgrandbs.tikz}&=&\input{grandbsbouclerebondisUhaut.tikz}
\end{array}}\end{equation}

\begin{equation}\labeletpreuve{traversebsbh}{\begin{array}{rcl}
\input{boucletraverseUbasgrandbs.tikz}&=&\input{grandbsboucletraverseUhaut.tikz}
\end{array}}\end{equation}

\begin{equation}\labeletpreuve{bouclesrebondistraversecommutent}{\begin{array}{rcl}\input{bouclerebondisUboucletraverseV.tikz}&=&\input{boucletraverseVbouclerebondisU.tikz}\end{array}}\end{equation}

\begin{equation}\labeletpreuve{fusionboucletraverse}{\begin{array}{rcl}\input{boucletraverseUboucletraverseV.tikz}&=&\input{boucletraverseVU.tikz}
\end{array}}\end{equation}

\begin{equation}\labeletpreuve{fusionbouclerebondis}{\begin{array}{rcl}\input{bouclerebondisUbouclerebondisV.tikz}&=&\input{bouclerebondisVU.tikz}
\end{array}}\end{equation}
\end{lemma}
\begin{proof}\vspace{-0.16cm}
\phantomsection\label{preuvebsxbsvar}The proof of Equation \eqref{bsxbsvar} is obtained by rotating the proof of Equation \eqref{bsxbs} by 180° (it uses Equation \eqref{swapbsbsvar} instead of Equation \eqref{swapbsbsbvar}).\\\\

\phantomsection\label{preuveloopI}To prove \cref{loopI}, we have:
\begin{longtable}{RCL}
\input{boucletraverseI.tikz}&\eqeqref{duplicateloop}&\input{filtrerienIetI.tikz}\\\\
&\eqeqref{idbox}&\input{filtrerienrienetI.tikz}\\\\\\
&\eqeqref{bsbs}&\input{filsurbouclevideI.tikz}\\\\
&\eqeqref{loopemptysimple}&\begin{tikzpicture}
	\begin{pgfonlayer}{nodelayer}
		\node [style=none] (0) at (-2, 0) {};
		\node [style=none] (1) at (2, 0) {};
	\end{pgfonlayer}
	\begin{pgfonlayer}{edgelayer}
		\draw [style=new] (0.center) to (1.center);
	\end{pgfonlayer}
\end{tikzpicture}
\end{longtable}

\phantomsection\label{preuvesplituv}To prove \cref{splituv}, we have:
\begin{longtable}{RCL}\!\!\!\!\!\!\!\!\!\input{bouclerebondisUboucletraverseV.tikz}&\overset{\text{naturality of}}{\overset{\text{the swap}}{=}}&\input{trace2bshbsbxhUmVb.tikz}\\\\
&\eqeqref{bsxbs}&\input{trace2bsbxhbsbUmVb.tikz}\\\\
&\overset{\text{dinaturality,}}{\overset{\text{naturality of}}{\overset{\text{the swap,}}{\overset{\text{yanking}}{=}}}}&\input{filtreUV.tikz}
\end{longtable}

\phantomsection\label{preuverebondisbsb}To prove \cref{rebondisbsb}, we have:
\hspace{-1em}\begin{longtable}{RCL}\input{bouclerebondisUbasgrandbs.tikz}&\eqeqref{swapbsx}&\input{tracexbbsbgrandbshUb.tikz}\\\\
&\eqeqref{swapbsbsbvar}&\input{tracexbxbgrandbshbsbxbUb.tikz}\\\\
&\eqexpl{\text{inverse law}}&\input{grandbsbouclerebondisUbas.tikz}\end{longtable}

\phantomsection\label{preuvetraversebshb}To prove \cref{traversebshb}, we have:
\begin{longtable}{RCL}\input{boucletraverseUhautgrandbs.tikz}&\overset{\text{axioms of}}{\overset{\text{the traced}}{\overset{\text{PROP}}{=}}}&\input{tracexbbshxbbshUb.tikz}\\\\
&\eqeqref{bsxbsvar}&\input{tracexbxbbshbsbUb.tikz}\\\\
&\eqexpl{\text{inverse law}}&\input{grandbsboucletraverseUbas.tikz}\end{longtable}

\phantomsection\label{preuverebondisbsh}To prove \cref{rebondisbsh}, we have:
\begin{longtable}{RCL}\input{bouclerebondisUhautgrandbs.tikz}&\overset{\text{inverse law},\eqref{swapbsx},}{\overset{\text{naturality of}}{\eqexpl{\text{the swap}}}}&\input{grandxbouclerebondisUbgrandbsgrandx.tikz}\\\\
&\eqeqref{rebondisbsb}&\input{grandxgrandbsbouclerebondisUbgrandx.tikz}\\\\
&\overset{\eqref{swapbsx},}{\overset{\text{naturality of}}{\overset{\text{the swap,}}{\eqexpl{\text{inverse law}}}}}&\input{grandbsbouclerebondisUhaut.tikz}\end{longtable}

\phantomsection\label{preuvetraversebsbh}To prove \cref{traversebsbh}, we have:
\begin{longtable}{RCL}\input{boucletraverseUbasgrandbs.tikz}&\overset{\text{inverse law},\eqref{swapbsx},}{\overset{\text{naturality of}}{\eqexpl{\text{the swap}}}}&\input{grandxboucletraverseUhgrandbsgrandx.tikz}\\\\
&\eqeqref{traversebshb}&\input{grandxgrandbsboucletraverseUbgrandx.tikz}\\\\
&\overset{\eqref{swapbsx},}{\overset{\text{naturality of}}{\overset{\text{the swap,}}{\eqexpl{\text{inverse law}}}}}&\input{grandbsboucletraverseUhaut.tikz}\end{longtable}

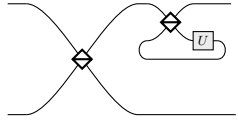
\phantomsection\label{preuvebouclesrebondistraversecommutent}Equation \eqref{bouclesrebondistraversecommutent} is a direct consequence of Equation \eqref{rebondisbsh}.\\\\

\phantomsection\label{preuvefusionboucletraverse}To prove \cref{fusionboucletraverse}, we have:
\begin{longtable}{RCL}\!\!\!\!\!\!\!\!\!\!\!\!\!\!\!\input{boucletraverseUboucletraverseV.tikz}&\eqeqref{traversebshb}&\input{trace2bshbsbVmUb.tikz}\\\\
&\eqeqref{duplicateloop}&\input{trace2bshbsbUbbsbVmVb.tikz}\\\\
&\eqdeuxeqref{bsuu}{fusionuvsimple}&\input{trace2bshbsbVmVUbbsb.tikz}\\\\
&\eqexpl{\text{dinaturality}}&\input{trace2bsbbshbsbVmVUb.tikz}\\\\
&\eqeqref{bsbsbsvar}&\input{trace2xbbshxbVmVUb.tikz}\\\\
&\overset{\text{dinaturality},}{\overset{\text{naturality of}}{\overset{\text{the swap,}}{\eqexpl{\text{inverse law}}}}}&\input{bouclevideUdansboucletraverseVU.tikz}\\\\
&\eqeqref{loopemptysimple}&\input{boucletraverseVU.tikz}
\end{longtable}

\phantomsection\label{preuvefusionbouclerebondis}To prove \cref{fusionbouclerebondis}, we have:
\begin{longtable}{RCL}\!\!\!\!\!\!\!\input{bouclerebondisUbouclerebondisV.tikz}&\eqeqref{bsnnnn}&\input{boucletraversennnnUboucletraversennnnV.tikz}\\\\
&\overset{\text{dinaturality},}{\eqexpl{\eqref{negu},\eqref{negneg}}}&\input{negboucletraverseUboucletraverseVneg.tikz}\\\\
&\eqeqref{fusionboucletraverse}&\input{negboucletraverseVUneg.tikz}\\\\
&\overset{\eqref{negneg},}{\overset{\text{dinaturality},}{\eqexpl{\eqref{negu},\eqref{bsnnnn}}}}&\input{bouclerebondisVU.tikz}
\end{longtable}

\end{proof}

Then we are ready to prove the last three equations:

\phantomsection\label{preuvebsuvh}To prove Equation \eqref{bsuvh}, we have:
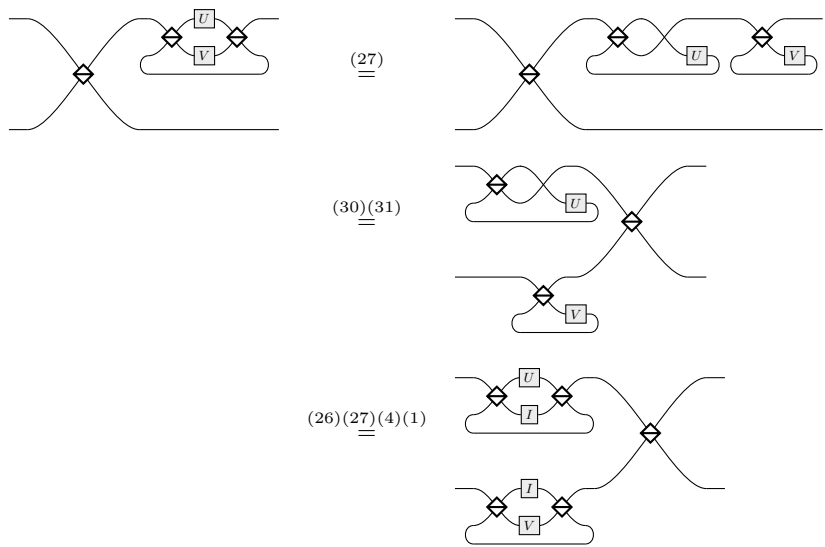
\begin{longtable}{RCL}\input{grandbsfiltreUVhaut.tikz}&\eqeqref{splituv}&\input{grandbsbouclerebondisUhboucletraverseVh.tikz}\\\\
&\eqdeuxeqref{rebondisbsh}{traversebsbh}&\input{bouclerebondisUhboucletraverseVbgrandbs.tikz}\\\\
&\eqquatreeqref{loopI}{splituv}{duplicateloop}{idbox}&\input{filtreUIhautfiltreIVbasgrandbs.tikz}\end{longtable}

\phantomsection\label{preuvebsuvb}\cref{bsuvb} is proved in the same way as \cref{bsuvh}, using Equations \eqref{rebondisbsb} and \eqref{traversebshb} instead of \eqref{rebondisbsh} and \eqref{traversebsbh}.\\\\

\phantomsection\label{preuvefusionuv}To prove Equation \eqref{fusionuv}, we have:
{\small{\tikzset{tikzfig/.style={baseline=-0.25em,scale=\echellefils,xscale=0.80455,every node/.style={scale=0.8}}}\begin{longtable}{RCL}\input{filtreUVfiltreUprimeVprime.tikz}&\eqeqref{splituv}&\input{brUbtVbrUprimebtVprime.tikz}\\\\
&\eqeqref{bouclesrebondistraversecommutent}&\input{brUbrUprimebtVbtVprime.tikz}\\\\
&\eqdeuxeqref{fusionbouclerebondis}{fusionboucletraverse}&\input{bouclerebondisUprimeUboucletraverseVprimeV.tikz}\\\\
&\eqeqref{splituv}&\tikzset{tikzfig/.style={baseline=-0.25em,scale=\echellefils,every node/.style={scale=0.8}}}\input{filtreUprimeUVprimeV.tikz}\end{longtable}}}
\endgroup

\paragraph{Proof of \cref{NFtrace}}\label{proofcomptrace}

We write $N$ in the form
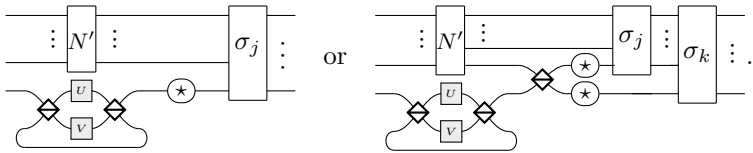
\[\input{typeAfiltreUV.tikz}\quad\text{or}\quad\input{typeCfiltreUV.tikz}.\]
As in Section \ref{proofcompleteness_comp}, we call these two forms type A and B respectively.

We proceed by case distinction:

\begin{itemize}
\item If $N$ is of type A and $j=n-1$, then we apply one of the following two equations, that we prove below to be consequences of the axioms of the PBS-calculus:
\begin{equation}\labeletpreuve{loopempty}{\tikzset{tikzfig/.style={baseline=-0.25em,scale=\echellefils,every node/.style={scale=0.8},borddiagrammevide/.style={-, dash pattern=on \traitsdiagrammevide off \traitsdiagrammevide on \traitsdiagrammevide off \traitsdiagrammevide on \traitsdiagrammevide off 0em}}}\begin{array}{rcl}
\input{bouclevidefiltreUV.tikz}&=&\begin{tikzpicture}
	\begin{pgfonlayer}{nodelayer}
		\node [style=diagrammevide] (0) at (0, 0) {};
	\end{pgfonlayer}
\end{tikzpicture}

\end{array}}\end{equation}
\begin{equation}\labeletpreuve{loopnegempty}{\tikzset{tikzfig/.style={baseline=-0.25em,scale=\echellefils,every node/.style={scale=0.8},borddiagrammevide/.style={-, dash pattern=on \traitsdiagrammevide off \traitsdiagrammevide on \traitsdiagrammevide off \traitsdiagrammevide on \traitsdiagrammevide off 0em}}}\begin{array}{rcl}
\input{bouclevidefiltreUVneg.tikz}&=&\begin{tikzpicture}
	\begin{pgfonlayer}{nodelayer}
		\node [style=diagrammevide] (0) at (0, 0) {};
	\end{pgfonlayer}
\end{tikzpicture}
\end{array}}\end{equation}
\item If $N$ is of type A and $j\neq n-1$, then we slide the $E(U,V)$ and the\begin{tikzpicture}[scale=0.6]
	\begin{pgfonlayer}{nodelayer}
		\node [style=none] (1) at (-1.2, 0) {};
		\node [style=none] (2) at (1.2, 0) {};
		\node [style=snegpotentiel] (4) at (0, 0) {};
	\end{pgfonlayer} 
	\begin{pgfonlayer}{edgelayer}
		\draw [style=noire] (4) to (2.center);
		\draw [style=noire] (1.center) to (4);
	\end{pgfonlayer}
\end{tikzpicture}
 through the trace and $\sigma_j$, then integrate them to $N'$ by \cref{NFajout}. Finally, we remove the trace by yanking, which gives us a normal form after a few additional manipulation of wires.
\item If $N$ is of type B and $k=n-1$, then we apply one of the following two equations:
\tikzset{tikzfig/.style={baseline=-0.25em,scale=\echellefils,every node/.style={scale=0.8}}}
\begin{equation}\labeletpreuve{truvbs}{\begin{array}{rcl}
\input{filtreUVdansbouclebs.tikz}&=&\input{filtreIV.tikz}
\end{array}}\end{equation}
\begin{equation}\labeletpreuve{truvbsnb}{\begin{array}{rcl}
\input{filtreUVdansbouclebsnb.tikz}&=&\input{filtreIVU.tikz}
\end{array}}\end{equation}
then we conclude by \cref{NFajout} and manipulation of wires.
\item If  $N$ is of type B, $k<n-1$ and $j=n-2$, then we apply one of the following two equations:
\begin{equation}\labeletpreuve{truvbsx}{\begin{array}{rcl}
\input{filtreUVdansbouclebsx.tikz}&=&\input{filtreUI.tikz}
\end{array}}\end{equation}
\begin{equation}\labeletpreuve{truvbsnhx}{\begin{array}{rcl}
\input{filtreUVdansbouclebsnhx.tikz}&=&\input{filtreUVI.tikz}
\end{array}}\end{equation}
then we conclude by \cref{NFajout} and manipulation of wires.
\item If N is of type B, $k<n-1$ and $j<n-2$, let\!\begin{tikzpicture}
	\begin{pgfonlayer}{nodelayer}
		\node [style=cartouche] (0) at (0, 0) {$D$};
		\node [style=none] (1) at (-1.5, 0) {};
		\node [style=none] (2) at (1.5, 0) {};
	\end{pgfonlayer}
	\begin{pgfonlayer}{edgelayer}
		\draw [style=noire] (1.center) to (0);
		\draw [style=noire] (0) to (2.center);
	\end{pgfonlayer}
\end{tikzpicture} represent $E(U,V)$. We proceed as follows:
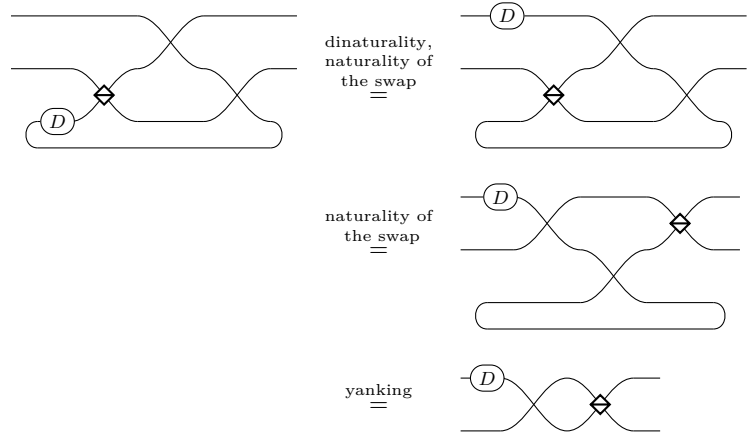
\begin{longtable}{RCL}\input{traceDgeneriquebbsbxhxb.tikz}&\overset{\text{dinaturality},}{\overset{\text{naturality of}}{\eqexpl{\text{the swap}}}}&\input{traceDgeneriquehbsbxhxb.tikz}\\\\
&\overset{\text{naturality of}}{\eqexpl{\text{ the swap}}}&\input{traceDgeneriquehxhxbbsh.tikz}\\\\
&\overset{\text{yanking}}{=}&\input{Dgeneriquehxbs.tikz}\end{longtable}
then we conclude by applying \cref{NFajout} three times and manipulating the wires.
\end{itemize}

\begingroup
\tikzset{tikzfig/.style={baseline=-0.25em,scale=\echellefils,every node/.style={scale=0.8},borddiagrammevide/.style={-, dash pattern=on 3.699pt off 3.699pt on 3.699pt off 3.699pt on 3.699pt off 0pt}}}
It remains to prove Equations \eqref{loopempty} to \eqref{truvbsnhx}.

\phantomsection\label{preuveloopempty}To prove \cref{loopempty}, we have:
\begin{longtable}{RCL}\!\!\!\!\!\!\!\!\!\input{bouclevidefiltreUV.tikz}&\eqexpl{\text{dinaturality}}&\input{bouclevideUhVbbsbs.tikz}\\\\
&\eqeqref{bsbs}&\input{bouclevideVdansbouclevideU.tikz}\\\\
&\eqeqref{loopemptysimple}&\begin{tikzpicture}
	\begin{pgfonlayer}{nodelayer}
		\node [style=diagrammevide] (0) at (0, 0) {};
	\end{pgfonlayer}
\end{tikzpicture}
\end{longtable}

\phantomsection\label{preuveloopnegempty}To prove \cref{loopnegempty}, we have:
\begin{longtable}{RCL}\!\!\!\!\!\!\!\!\!\input{bouclevidefiltreUVneg.tikz}&\eqexpl{\text{dinaturality}}&\input{bouclevideUhVbbsnhbs.tikz}\\\\
&\eqeqref{bsnbsh}&\input{bouclevideUhVbnhbsnh.tikz}\\\\
&\eqexpl{\text{dinaturality}}&\input{bouclevidenhUhVbnhbs.tikz}\\\\
&\eqdeuxeqref{negu}{negneg}&\input{bouclevideUhVbbs.tikz}\\\\
&\eqeqref{duplicateloop}&\input{bouclevideUhbsVbbsUb.tikz}\\\\
&\eqexpl{\text{dinaturality}}&\input{bouclevideUhUbbsVbbs.tikz}\\\\
&\eqdeuxeqref{bsuu}{fusionuvsimple}&\input{bouclevidebsUhVUbbs.tikz}\\\\
&\eqeqref{loopempty}&\begin{tikzpicture}
	\begin{pgfonlayer}{nodelayer}
		\node [style=diagrammevide] (0) at (0, 0) {};
	\end{pgfonlayer}
\end{tikzpicture}
\end{longtable}

\phantomsection\label{preuvetruvbs}To prove \cref{truvbs}, we have:
\begin{longtable}{RCL}\!\!\!\!\!\!\!\!\!\input{filtreUVdansbouclebs.tikz}&\eqeqref{swapbsbsbvar}&\input{trace2bsbUmVbxbbshbsbxb.tikz}\\\\
&\eqexpl{\text{dinaturality}}&\input{trace2bsbxbbsbUmVbxbbsh.tikz}\\\\
&\overset{\eqref{swapbsx},\eqref{bsbs},}{\overset{\text{inverse law}}{=}}&\input{trace2VmUbbsh.tikz}\\\\
&\eqtroiseqref{loopemptysimple}{duplicateloop}{idbox}&\input{filtreIV.tikz}\end{longtable}

\phantomsection\label{preuvetruvbsnb}To prove \cref{truvbsnb}, we have:
\begingroup
\tikzset{tikzfig/.style={baseline=-0.25em,scale=\echellefils,xscale=0.92222,every node/.style={scale=0.8}}}
\begin{longtable}{RCL}\input{filtreUVdansbouclebsnb.tikz}
&\eqeqref{swapbsbsbvar}&\input{trace2bsbUmVbxbbshbsbxbnm.tikz}\\\\
&\eqexpl{\text{dinaturality}}&\input{trace2bsbxbnmbsbUmVbxbbsh.tikz}\\\\
&\overset{\eqref{swapbsx},\eqref{bsnbsh},}{\overset{\text{naturality of}}{\overset{\text{the swap},\eqref{swapbsx},}{\eqexpl{\text{inverse law}}}}}&\input{trace2nbbsbnbVmUbbsh.tikz}\\\\
&\overset{\text{dinaturality},}{\overset{\eqref{negu},\eqref{negneg}}{=}}&\input{trace2bsbVmUbbsh.tikz}\\\\
&\overset{\eqref{duplicateloop}}{=}&\input{trace2bsbUbbsbVmVbbsh.tikz}\\\\
&\eqdeuxeqref{bsuu}{fusionuvsimple}&\input{filtreVVUdansbouclebs.tikz}\\\\
&\eqeqref{truvbs}&\input{filtreIVU.tikz}\end{longtable}
\endgroup

\phantomsection\label{preuvetruvbsx}To prove \cref{truvbsx}, we have:
\begin{longtable}{RCL}\input{filtreUVdansbouclebsx.tikz}
&\eqeqref{bsnnnn}&\input{trace2bsbUmVbbsbnhnmbshnhnm.tikz}\\\\
&\overset{\text{dinaturality,}}{\overset{\eqref{neguv},\eqref{negneg}}{=}}&\input{trace2bsbVmUbbsbnhbshnh.tikz}\\\\
&\eqeqref{truvbs}&\input{negfiltreIUneg.tikz}\\\\
&\eqdeuxeqref{neguv}{negneg}&\input{filtreUI.tikz}\end{longtable}

\phantomsection\label{preuvetruvbsnhx}To prove \cref{truvbsnhx}, we have:
\begingroup
\tikzset{tikzfig/.style={baseline=-0.25em,scale=\echellefils,xscale=0.9827,every node/.style={scale=0.8}}}
\begin{longtable}{RCL}\input{filtreUVdansbouclebsnhx.tikz}
&\overset{\text{naturality of}}{\overset{\text{the swap,}}{\eqexpl{\eqref{bsnnnn},\eqref{negneg}}}}&
\input{trace2bsbUmVbbsbnhnmbshnh.tikz}\\\\
&\overset{\eqref{neguv},}{\overset{\text{dinaturality}}{=}}&\input{trace2bsbUmVbbsbnhbshnhnm.tikz}\\\\
&\eqeqref{truvbsnb}&\input{negfiltreIUVneg.tikz}\\\\
&\eqdeuxeqref{neguv}{negneg}&\input{filtreUVI.tikz}\end{longtable}
\endgroup

\paragraph{Normal forms of the generators} 
\label{proofNFGenerators}
The following equations are consequences of the axioms of the PBS-calculus, and allow us to put the generators in normal form:
\begin{equation}\labeletpreuve{spliti}{\begin{array}{rcl}\begin{tikzpicture}
	\begin{pgfonlayer}{nodelayer}
		\node [style=none] (0) at (-1.5, 0) {};
		\node [style=none] (1) at (1.5, 0) {};
	\end{pgfonlayer}
	\begin{pgfonlayer}{edgelayer}
		\draw [style=noire] (0.center) to (1.center);
	\end{pgfonlayer}
\end{tikzpicture}&=&\input{filtreIIcentresurfil.tikz}
\end{array}}\end{equation}
\begin{equation}\label{NFneg}\begin{array}{rcl}\begin{tikzpicture}
	\begin{pgfonlayer}{nodelayer}
		\node [style=cartouche] (0) at (0, 0) {$\neg$};
		\node [style=none] (1) at (-1.5, 0) {};
		\node [style=none] (2) at (1.5, 0) {};
	\end{pgfonlayer}
	\begin{pgfonlayer}{edgelayer}
		\draw [style=noire] (1.center) to (0);
		\draw [style=noire] (0) to (2.center);
	\end{pgfonlayer}
\end{tikzpicture}
&=&\input{filtreIInegcentresurfil.tikz}
\end{array}\end{equation}
\begin{equation}\labeletpreuve{splitu}{\begin{array}{rcl}\begin{tikzpicture}
	\begin{pgfonlayer}{nodelayer}
		\node [style=newe] (0) at (0, 0) {$U$};
		\node [style=none] (1) at (-1.5, 0) {};
		\node [style=none] (2) at (1.5, 0) {};
	\end{pgfonlayer}
	\begin{pgfonlayer}{edgelayer}
		\draw [style=new] (1.center) to (2.center);
	\end{pgfonlayer}
\end{tikzpicture}&=&\input{filtreUUcentresurfil.tikz}
\end{array}}\end{equation}
\begin{equation}\label{NFswap}\begin{array}{rcl}\input{swap.tikz}&=&\input{filtreIIhautfiltreIIbasgrandx.tikz}
\end{array}\end{equation}
\begin{equation}\label{NFbs}\begin{array}{rcl}\input{beamsplitter.tikz}&=&\input{filtreIIhautfiltreIIbasgrandbs.tikz}
\end{array}\end{equation}

\phantomsection\label{preuvesplitu}To prove Equation \eqref{splitu}, we have:
\begin{longtable}{RCL}\begin{tikzpicture}
	\begin{pgfonlayer}{nodelayer}
		\node [style=newe] (0) at (0, 0) {$U$};
		\node [style=none] (1) at (-1.5, 0) {};
		\node [style=none] (2) at (1.5, 0) {};
	\end{pgfonlayer}
	\begin{pgfonlayer}{edgelayer}
		\draw [style=new] (1.center) to (2.center);
	\end{pgfonlayer}
\end{tikzpicture}&\eqeqref{loopemptysimple}&\input{gateUsurbouclevideU.tikz}\ \ \eer{bsbs}\ \ \input{traceUhUbbsbs.tikz}\\\\
&\eqeqref{bsuu}&\input{filtreUU.tikz}\end{longtable}

\phantomsection\label{preuvespliti}To prove Equation \eqref{spliti}, we have:
\begin{longtable}{RCL}\begin{tikzpicture}
	\begin{pgfonlayer}{nodelayer}
		\node [style=none] (0) at (-1.5, 0) {};
		\node [style=none] (1) at (1.5, 0) {};
	\end{pgfonlayer}
	\begin{pgfonlayer}{edgelayer}
		\draw [style=noire] (0.center) to (1.center);
	\end{pgfonlayer}
\end{tikzpicture}&\eqeqref{idbox}&\begin{tikzpicture}[scale=0.6]
	\begin{pgfonlayer}{nodelayer}
		\node [style=newe] (0) at (0, 0) {$I$};
		\node [style=none] (1) at (-1.5, 0) {};
		\node [style=none] (2) at (1.5, 0) {};
	\end{pgfonlayer}
	\begin{pgfonlayer}{edgelayer}
		\draw [style=new] (1.center) to (2.center);
	\end{pgfonlayer}
\end{tikzpicture}\ \ \eer{splitu}\ \ \input{filtreIIcentresurfil.tikz}\end{longtable}

Equations \eqref{NFneg}, \eqref{NFswap} and \eqref{NFbs} are direct consequences of Equation \eqref{spliti}.
\endgroup

\subsubsection{Proof of Theorem \ref{minimality}}\label{proofmin}

We prove for each equation that it is not a consequence of the others in a dedicated lemma. For Equations \eqref{idbox}, \eqref{negu}, \eqref{bsnnnn} and \eqref{fusionuvsimple}, the proof follows a common pattern: we introduce a new semantics $\varinterp{.}$ and check that it preserves every equation except the one that we want to prove to be independent from the others. In each case, \cref{tracedpropfunctorcongruence} gives us that the consequences of the preserved equations are preserved too, which proves that the unpreserved equation is not a consequence of the others.

\begin{lemma}\label{tracedpropfunctorcongruence}
Let $\varinterp .$ be a function mapping any diagram $D:n\to n$ to a linear map $\varinterp{D}\in\S_n$, defined inductively in the same way as $\interp.$ except maybe in the case of \input{beamsplitter-xs.tikz}, \begin{tikzpicture}[scale=0.8]
	\begin{pgfonlayer}{nodelayer}
		\node [style=cartouche] (0) at (0, 0) {\tiny$\neg$};
		\node [style=none] (1) at (-0.8, 0) {};
		\node [style=none] (2) at (0.8, 0) {};
	\end{pgfonlayer}
	\begin{pgfonlayer}{edgelayer}
		\draw [style=noire] (1.center) to (0);
		\draw [style=noire] (0) to (2.center);
	\end{pgfonlayer}
\end{tikzpicture}
 and \begin{tikzpicture}[scale=0.6]
	\begin{pgfonlayer}{nodelayer}
		\node [style=newe] (0) at (0, 0) {\tiny\!$U$\!};
		\node [style=none] (1) at (-1, 0) {};
		\node [style=none] (2) at (1, 0) {};
	\end{pgfonlayer}
	\begin{pgfonlayer}{edgelayer}
		\draw [style=new] (1.center) to (2.center);
	\end{pgfonlayer}
\end{tikzpicture}. Let $A$ be a set of equations of the form $D_1=D_2$ where $D_1,D_2$ are $\textup{PBS}$-diagrams, such that every equation of $A$ is preserved by $\varinterp.$ (that is, for every equation $D_1=D_2$ in $A$ we have $\varinterp{D_1}=\varinterp{D_2}$). Then $A$ is sound with respect to $\varinterp.$, that is, for any two diagrams $D_1,D_2:n\to n$, if $A\vdash D_1=D_2$ then $\varinterp{D_1}=\varinterp{D_2}$.
\end{lemma}
\begin{proof}\vspace{-0.16cm}
The same proof as for $\interp.$ shows that $\varinterp.$ is well-defined.

By definition, $A\vdash .=.$ is the smallest congruence satisfying the equations of $A$. Since $\varinterp{D_2\circ D_1}$ and $\varinterp{D_1\otimes D_2}$ only depend on $\varinterp{D_1}$ and $\varinterp{D_2}$, and $\varinterp{Tr(D)}$ only depends on $\varinterp{D}$, the relation $\sim$, defined as $D_1\sim D_2$ if and only if $\varinterp{D_1}=\varinterp{D_2}$, is a congruence. Therefore, it contains $A\vdash .=.$, which is what we wanted to prove.
\end{proof}

\begin{lemma}\label{independanceidbox}
Equation \eqref{idbox} is not a consequence of Equations \eqref{negu} to \eqref{bsnbsh}.
\end{lemma}
\begin{proof}\vspace{-0.16cm}
Let us define $\varinterp .$ inductively in the same way as $\interp .$, except in the case of\begin{tikzpicture}[scale=0.6]
	\begin{pgfonlayer}{nodelayer}
		\node [style=newe] (0) at (0, 0) {\tiny\!$U$\!};
		\node [style=none] (1) at (-1, 0) {};
		\node [style=none] (2) at (1, 0) {};
	\end{pgfonlayer}
	\begin{pgfonlayer}{edgelayer}
		\draw [style=new] (1.center) to (2.center);
	\end{pgfonlayer}
\end{tikzpicture}, for which we define $\varinterp{\begin{tikzpicture}[scale=0.6]
	\begin{pgfonlayer}{nodelayer}
		\node [style=newe] (0) at (0, 0) {\tiny\!$U$\!};
		\node [style=none] (1) at (-1, 0) {};
		\node [style=none] (2) at (1, 0) {};
	\end{pgfonlayer}
	\begin{pgfonlayer}{edgelayer}
		\draw [style=new] (1.center) to (2.center);
	\end{pgfonlayer}
\end{tikzpicture}}:\H_1\to\H_1::\ket{c,0,y}\mapsto0$.\\
Equations \eqref{negu}, \eqref{bsuu} and \eqref{fusionuvsimple} are preserved by $\varinterp.$ because both sides are interpreted by the zero map. Equation \eqref{loopemptysimple} is preserved because both side are interpreted by the unique map $\H_0\to\H_0$. Equation \eqref{duplicateloop} is preserved because both sides are interpreted as $\begin{cases}\ket{\rightarrow,0,y}&\mapsto\ket{\rightarrow,0,y}\\\ket{\uparrow,0,y}&\mapsto 0\end{cases}$. Finally, Equations \eqref{bsnnnn}  and \eqref{bsbs} to \eqref{bsnbsh} are preserved because both sides are interpreted in the same way as by $\interp .$. As a consequence, by \cref{tracedpropfunctorcongruence}, all consequences of equations \eqref{negu} to \eqref{bsnbsh} are preserved by $\varinterp.$. By contrast, Equation \eqref{idbox} is not preserved by $\varinterp.$ because one side is interpreted by the identity whereas the other side is interpreted by the zero map. Hence, Equation \eqref{idbox} is not a consequence of Equations \eqref{negu} to \eqref{bsnbsh}.
\end{proof}

\begin{lemma}
If $U\neq I$, then Equation \eqref{negu} is not a consequence of Equations \eqref{idbox} and \eqref{bsuu} to \eqref{bsnbsh}.
\end{lemma}
\begin{proof}\vspace{-0.16cm}
Let us define $\varinterp .$ inductively in the same way as $\interp .$, except in the case of\begin{tikzpicture}[scale=0.6]
	\begin{pgfonlayer}{nodelayer}
		\node [style=newe] (0) at (0, 0) {\tiny\!$U$\!};
		\node [style=none] (1) at (-1, 0) {};
		\node [style=none] (2) at (1, 0) {};
	\end{pgfonlayer}
	\begin{pgfonlayer}{edgelayer}
		\draw [style=new] (1.center) to (2.center);
	\end{pgfonlayer}
\end{tikzpicture}, for which we define $\varinterp{\begin{tikzpicture}[scale=0.6]
	\begin{pgfonlayer}{nodelayer}
		\node [style=newe] (0) at (0, 0) {\tiny\!$U$\!};
		\node [style=none] (1) at (-1, 0) {};
		\node [style=none] (2) at (1, 0) {};
	\end{pgfonlayer}
	\begin{pgfonlayer}{edgelayer}
		\draw [style=new] (1.center) to (2.center);
	\end{pgfonlayer}
\end{tikzpicture}}:=\interp{\input{boucletraverseU.tikz}}=\begin{cases}\ket{\rightarrow,0,y}&\mapsto\ket{\rightarrow,0,y}\\\ket{\uparrow,0,y}&\mapsto\ket{\uparrow,0}\otimes U\ket y\end{cases}$.

Equation \eqref{negu} is not satisfied unless $U=I$, because the left-hand side is interpreted as $\begin{cases}\ket{\rightarrow,0,y}&\mapsto\ket{\uparrow,0}\otimes U\ket{y}\\\ket{\uparrow,0,y}&\mapsto\ket{\rightarrow,0,y}\end{cases}$ whereas the right-hand side is interpreted as\smallskip\linebreak $\begin{cases}\ket{\rightarrow,0,y}&\mapsto\ket{\uparrow,0,y}\\\ket{\uparrow,0,y}&\mapsto\ket{\rightarrow,0}\otimes U\ket{y}\end{cases}$. By using the graphical characterisation of the denotational semantics adapted to $\varinterp.$, it is easy to check that Equations \eqref{idbox} and \eqref{bsuu} to \eqref{bsnbsh} are preserved by $\varinterp.$. By \cref{tracedpropfunctorcongruence}, this implies that all consequences of these equations are preserved by $\varinterp.$, so that Equation \eqref{negu} is not a consequence of them.
\end{proof}

\begin{lemma}
If $\det(U)\neq 1$, then Equation \eqref{bsuu} is not a consequence of Equations \eqref{idbox}, \eqref{negu} and \eqref{duplicateloop} to \eqref{bsnbsh}.
\end{lemma}
\begin{proof}\vspace{-0.16cm}
Since Equation \eqref{splitu} is a consequence of Equations \eqref{idbox} to \eqref{bsnbsh}, to prove that Equation \eqref{bsuu} is not a consequence of Equations \eqref{idbox}, \eqref{negu} and \eqref{duplicateloop} to \eqref{bsnbsh}, it suffices to prove that Equation \eqref{splitu} is not a consequence of these equations.

Given a diagram $D:n\to n$, let us say that a wire in $D$ is \emph{used} if there exists $c\in\{\rightarrow,\uparrow\}$ and $p\in[n]$ such that an input photon with classical polarisation $c$ and position $p$ passes through this wire. Let us define $d(D)$ as the product of all determinants of the matrices labelling the gates that are on used wires of $D$.

Let us fix a diagram $D$ and consider the effect of applying the axioms inside it. It is easy to check that all axioms of traced PROP, as well as Equations \eqref{idbox}, \eqref{negu}, \eqref{bsnnnn} and \eqref{bsbs} to \eqref{bsnbsh} preserve the gates of $D$ and the fact that their wire is used or not. Equation \eqref{loopemptysimple} can only add or remove gates on unused wires. Equation \eqref{duplicateloop} adds or removes\begin{tikzpicture}[scale=0.6]
	\begin{pgfonlayer}{nodelayer}
		\node [style=newe] (0) at (0, 0) {\tiny\!$V$\!};
		\node [style=none] (1) at (-1, 0) {};
		\node [style=none] (2) at (1, 0) {};
	\end{pgfonlayer}
	\begin{pgfonlayer}{edgelayer}
		\draw [style=new] (1.center) to (2.center);
	\end{pgfonlayer}
\end{tikzpicture}
 on an unused wire and does change the fact that the wire of\begin{tikzpicture}[scale=0.6]
	\begin{pgfonlayer}{nodelayer}
		\node [style=newe] (0) at (0, 0) {\tiny\!$U$\!};
		\node [style=none] (1) at (-1, 0) {};
		\node [style=none] (2) at (1, 0) {};
	\end{pgfonlayer}
	\begin{pgfonlayer}{edgelayer}
		\draw [style=new] (1.center) to (2.center);
	\end{pgfonlayer}
\end{tikzpicture} is used or not, indeed, in the patterns on both sides of the equation, the wire of\begin{tikzpicture}[scale=0.6]
	\begin{pgfonlayer}{nodelayer}
		\node [style=newe] (0) at (0, 0) {\tiny\!$U$\!};
		\node [style=none] (1) at (-1, 0) {};
		\node [style=none] (2) at (1, 0) {};
	\end{pgfonlayer}
	\begin{pgfonlayer}{edgelayer}
		\draw [style=new] (1.center) to (2.center);
	\end{pgfonlayer}
\end{tikzpicture} is used if and only if it is possible to have a photon with polarisation $\uparrow$ arrive at the input of the pattern. Applying Equation \eqref{fusionuvsimple} replaces\begin{tikzpicture}[scale=0.6]
	\begin{pgfonlayer}{nodelayer}
		\node [style=newe] (0) at (0, 0) {\tiny\!$U$\!};
		\node [style=none] (1) at (-1, 0) {};
		\node [style=none] (2) at (1, 0) {};
	\end{pgfonlayer}
	\begin{pgfonlayer}{edgelayer}
		\draw [style=new] (1.center) to (2.center);
	\end{pgfonlayer}
\end{tikzpicture} and\begin{tikzpicture}[scale=0.6]
	\begin{pgfonlayer}{nodelayer}
		\node [style=newe] (0) at (0, 0) {\tiny\!$V$\!};
		\node [style=none] (1) at (-1, 0) {};
		\node [style=none] (2) at (1, 0) {};
	\end{pgfonlayer}
	\begin{pgfonlayer}{edgelayer}
		\draw [style=new] (1.center) to (2.center);
	\end{pgfonlayer}
\end{tikzpicture}
 by\begin{tikzpicture}[scale=0.6]
	\begin{pgfonlayer}{nodelayer}
		\node [style=newe] (0) at (0, 0) {\tiny\!$VU$\!};
		\node [style=none] (1) at (-1, 0) {};
		\node [style=none] (2) at (1, 0) {};
	\end{pgfonlayer}
	\begin{pgfonlayer}{edgelayer}
		\draw [style=new] (1.center) to (2.center);
	\end{pgfonlayer}
\end{tikzpicture}
 (or\begin{tikzpicture}[scale=0.6]
	\begin{pgfonlayer}{nodelayer}
		\node [style=newe] (0) at (0, 0) {\tiny\!$VU$\!};
		\node [style=none] (1) at (-1, 0) {};
		\node [style=none] (2) at (1, 0) {};
	\end{pgfonlayer}
	\begin{pgfonlayer}{edgelayer}
		\draw [style=new] (1.center) to (2.center);
	\end{pgfonlayer}
\end{tikzpicture}
 by\begin{tikzpicture}[scale=0.6]
	\begin{pgfonlayer}{nodelayer}
		\node [style=newe] (0) at (0, 0) {\tiny\!$U$\!};
		\node [style=none] (1) at (-1, 0) {};
		\node [style=none] (2) at (1, 0) {};
	\end{pgfonlayer}
	\begin{pgfonlayer}{edgelayer}
		\draw [style=new] (1.center) to (2.center);
	\end{pgfonlayer}
\end{tikzpicture} and\begin{tikzpicture}[scale=0.6]
	\begin{pgfonlayer}{nodelayer}
		\node [style=newe] (0) at (0, 0) {\tiny\!$V$\!};
		\node [style=none] (1) at (-1, 0) {};
		\node [style=none] (2) at (1, 0) {};
	\end{pgfonlayer}
	\begin{pgfonlayer}{edgelayer}
		\draw [style=new] (1.center) to (2.center);
	\end{pgfonlayer}
\end{tikzpicture}
) on a given wire, which does not change $d(D)$. Thus, applying Equations \eqref{idbox}, \eqref{negu} and \eqref{duplicateloop} to \eqref{bsnbsh} does not change $d(D)$. On the other hand, we have $d\left(\begin{tikzpicture}[scale=0.6]
	\begin{pgfonlayer}{nodelayer}
		\node [style=newe] (0) at (0, 0) {\tiny\!$U$\!};
		\node [style=none] (1) at (-1, 0) {};
		\node [style=none] (2) at (1, 0) {};
	\end{pgfonlayer}
	\begin{pgfonlayer}{edgelayer}
		\draw [style=new] (1.center) to (2.center);
	\end{pgfonlayer}
\end{tikzpicture}\right)=\det(U)$ and $d\left(\input{filtreUU-s.tikz}\right)=\det(U)^2$. Hence, as soon as $\det(U)\neq 1$, Equation \eqref{splitu} changes $d(D)$, so that it is not a consequence of \eqref{idbox}, \eqref{negu} and \eqref{duplicateloop} to \eqref{bsnbsh}, which is what we wanted to prove.
\end{proof}

\begin{lemma}
For any $U$, Equation \eqref{duplicateloop} is not a consequence of Equations \eqref{idbox} to \eqref{bsuu} and \eqref{bsnnnn} to \eqref{bsnbsh}.
\end{lemma}
\begin{proof}\vspace{-0.16cm}
This is clear, because Equations \eqref{idbox} to \eqref{bsuu} and \eqref{bsnnnn} to \eqref{bsnbsh}, as well as the axioms of traced PROP, preserve the parity of the total number of\input{beamsplitter-xs.tikz} and\begin{tikzpicture}[scale=0.8]
	\begin{pgfonlayer}{nodelayer}
		\node [style=cartouche] (0) at (0, 0) {\tiny$\neg$};
		\node [style=none] (1) at (-0.8, 0) {};
		\node [style=none] (2) at (0.8, 0) {};
	\end{pgfonlayer}
	\begin{pgfonlayer}{edgelayer}
		\draw [style=noire] (1.center) to (0);
		\draw [style=noire] (0) to (2.center);
	\end{pgfonlayer}
\end{tikzpicture}
 in a given diagram, whereas Equation \eqref{duplicateloop} changes this parity.
\end{proof}

\begin{lemma}
Equation \eqref{bsnnnn} is not a consequence of Equations \eqref{idbox} to \eqref{duplicateloop} and \eqref{fusionuvsimple} to \eqref{bsnbsh}.
\end{lemma}
\begin{proof}\vspace{-0.16cm}
Let us define $\varinterp .$ inductively in the same way as $\interp .$, except in the cases of\input{beamsplitter-xs.tikz} and $\begin{tikzpicture}[scale=0.8]
	\begin{pgfonlayer}{nodelayer}
		\node [style=cartouche] (0) at (0, 0) {\tiny$\neg$};
		\node [style=none] (1) at (-0.8, 0) {};
		\node [style=none] (2) at (0.8, 0) {};
	\end{pgfonlayer}
	\begin{pgfonlayer}{edgelayer}
		\draw [style=noire] (1.center) to (0);
		\draw [style=noire] (0) to (2.center);
	\end{pgfonlayer}
\end{tikzpicture}
$, for which we define $\varinterp{\input{beamsplitter-xs.tikz}}$ and $\varinterp{\begin{tikzpicture}[scale=0.8]
	\begin{pgfonlayer}{nodelayer}
		\node [style=cartouche] (0) at (0, 0) {\tiny$\neg$};
		\node [style=none] (1) at (-0.8, 0) {};
		\node [style=none] (2) at (0.8, 0) {};
	\end{pgfonlayer}
	\begin{pgfonlayer}{edgelayer}
		\draw [style=noire] (1.center) to (0);
		\draw [style=noire] (0) to (2.center);
	\end{pgfonlayer}
\end{tikzpicture}
}$ as being the identity (the proof also works if we additionally define $\varinterp{\begin{tikzpicture}[scale=0.6]
	\begin{pgfonlayer}{nodelayer}
		\node [style=newe] (0) at (0, 0) {\tiny\!$U$\!};
		\node [style=none] (1) at (-1, 0) {};
		\node [style=none] (2) at (1, 0) {};
	\end{pgfonlayer}
	\begin{pgfonlayer}{edgelayer}
		\draw [style=new] (1.center) to (2.center);
	\end{pgfonlayer}
\end{tikzpicture}}$ as the identity). Then it is clear that Equations \eqref{idbox} to \eqref{duplicateloop} and \eqref{fusionuvsimple} to \eqref{bsnbsh} are preserved, and Equation \eqref{bsnnnn} is not preserved because its left-hand side is interpreted as the identity whereas its right-hand side is interpreted as $\interp{\input{swap-xs.tikz}}$. By \cref{tracedpropfunctorcongruence}, this implies that Equation \eqref{bsnnnn} is not a consequence of Equations \eqref{idbox} to \eqref{duplicateloop} and \eqref{fusionuvsimple} to \eqref{bsnbsh}.
\end{proof}

\begin{lemma}
If $U,V\neq I$, then Equation \eqref{fusionuvsimple} is not a consequence of Equations \eqref{idbox} to \eqref{bsnnnn} and \eqref{loopemptysimple} to \eqref{bsnbsh}.
\end{lemma}
\begin{proof}\vspace{-0.16cm}
Let us define $\varinterp .$ inductively in the same way as $\interp .$, except in the case of\begin{tikzpicture}[scale=0.6]
	\begin{pgfonlayer}{nodelayer}
		\node [style=newe] (0) at (0, 0) {\tiny\!$U$\!};
		\node [style=none] (1) at (-1, 0) {};
		\node [style=none] (2) at (1, 0) {};
	\end{pgfonlayer}
	\begin{pgfonlayer}{edgelayer}
		\draw [style=new] (1.center) to (2.center);
	\end{pgfonlayer}
\end{tikzpicture}, for which we define $\varinterp{\begin{tikzpicture}[scale=0.6]
	\begin{pgfonlayer}{nodelayer}
		\node [style=newe] (0) at (0, 0) {\tiny\!$U$\!};
		\node [style=none] (1) at (-1, 0) {};
		\node [style=none] (2) at (1, 0) {};
	\end{pgfonlayer}
	\begin{pgfonlayer}{edgelayer}
		\draw [style=new] (1.center) to (2.center);
	\end{pgfonlayer}
\end{tikzpicture}}:=\begin{cases}\ket{c,p,x}\mapsto\ket{c,p,x}&\text{if $U=I$}\\\ket{c,p,x}\mapsto\ket{c,p}\otimes M\ket{x}&\text{if $U\neq I$}\end{cases}$ where $M$ is a fixed arbitrary matrix such that $M^2\neq M$. Then it is easy to check that Equations \eqref{idbox} to \eqref{bsnnnn} and \eqref{loopemptysimple} to \eqref{bsnbsh} are preserved by $\varinterp.$. But Equation \eqref{fusionuvsimple} is not preserved if $U,V\neq I$, because then the left-hand side is interpreted as $\ket{c,p,x}\mapsto\ket{c,p}\otimes M^2\ket{x}$ whereas the left-hand side is interpreted as $\ket{c,p,x}\mapsto\ket{c,p}\otimes M\ket{x}$, and $M^2\neq M$.
By \cref{tracedpropfunctorcongruence}, this implies that Equation \eqref{fusionuvsimple} is not a consequence of Equations \eqref{idbox} to \eqref{bsnnnn} and \eqref{loopemptysimple} to \eqref{bsnbsh}.
\end{proof}

\begin{lemma}
For any $U$, Equation \eqref{loopemptysimple} is not a consequence of Equations \eqref{idbox} to \eqref{fusionuvsimple} and \eqref{bsbs} to \eqref{bsnbsh}.
\end{lemma}
\begin{proof}\vspace{-0.16cm}
This is clear, because Equation \eqref{loopemptysimple} is the only one that allows us to make a nonempty diagram equivalent to the empty diagram.
\end{proof}

\begin{lemma}
Equation \eqref{bsbs} is not a consequence of Equations \eqref{idbox} to \eqref{loopemptysimple}, \eqref{bsbsbs} and \eqref{bsnbsh}.
\end{lemma}
\begin{proof}\vspace{-0.16cm}
This is clear, because Equation \eqref{bsbs} is the only one that allows us to make a diagram without beam splitters equivalent to a diagram containing beam splitters.
\end{proof}

\begin{lemma}
Equation \eqref{bsnbsh} is not a consequence of Equations \eqref{idbox} to \eqref{bsbsbs}.
\end{lemma}
\begin{proof}\vspace{-0.16cm}
It suffices to remark that Equation \eqref{bsnbsh} is the only one that allows us to change the parity of the number of\begin{tikzpicture}[scale=0.8]
	\begin{pgfonlayer}{nodelayer}
		\node [style=cartouche] (0) at (0, 0) {\tiny$\neg$};
		\node [style=none] (1) at (-0.8, 0) {};
		\node [style=none] (2) at (0.8, 0) {};
	\end{pgfonlayer}
	\begin{pgfonlayer}{edgelayer}
		\draw [style=noire] (1.center) to (0);
		\draw [style=noire] (0) to (2.center);
	\end{pgfonlayer}
\end{tikzpicture}
 in a diagram.
\end{proof}

To prove that Equation \eqref{bsbsbs} is not a consequence of the others, we will need to talk about sub-diagrams in a context where not all axioms of traced PROP are assumed. Although the notion of sub-diagram is clear in a traced PROP, it becomes less obvious when some axioms are missing. This is why we give a formal inductive definition of it:

\begin{definition}\label{defsubdiagram}
We define the notion of sub-diagram inductively as follows. Given two diagrams $d$ and $D$, we say that $d$ is a sub-diagram of $D$ if we at least one of the following properties is satisfied (up to the currently assumed axioms, which are the axioms of a traced PROP in most of this paper but will be the axioms of a PROTWEB in the proof of \cref{independancebsbsbs}):
\begin{itemize}
\item $d=D$
\item there exists two nonempty diagrams $D_1$ and $D_2$ such that $D=D_2\circ D_1$ and $d$ is a sub-diagram of $D_1$ or a sub-diagram of $D_2$
\item there exists two nonempty diagrams $D_1$ and $D_2$ such that $D=D_1\otimes D_2$ and $d$ is a sub-diagram of $D_1$ or a sub-diagram of $D_2$
\item there exists a diagram $D'$ such that $D=Tr(D')$ and $d$ is a sub-diagram of $D'$.
\end{itemize}
\end{definition}

\begin{lemma}\label{independancebsbsbs}
Equation \eqref{bsbsbs} is not a consequence of equations \eqref{idbox} to \eqref{bsbs} and \eqref{bsnbsh}.
\end{lemma}
\begin{proof}\vspace{-0.16cm}
Let us first make two remarks.

First, since Equation \eqref{bsbsbs} does not contain gates, if it is a consequence of the other equations, then it is a consequence of these equations where all $U$ and $V$ are instantiated by $I$. Indeed, all of these equations that contain gates are still true when all $U$ and $V$ are instantiated by $I$. Hence, given a valid derivation of Equation \eqref{bsbsbs} from the others, by replacing every unitary matrix by $I$ in this derivation we get a valid derivation of Equation \eqref{bsbsbs}.

Second, by Equation \eqref{idbox}, being a consequence of Equations \eqref{idbox} to \eqref{bsbs} and \eqref{bsnbsh} where all $U$ and $V$ are instantiated by $I$ is equivalent to being a consequence of these equations where the gates have been removed (except in Equation \eqref{idbox}). That is, being a consequence of the following equations:

\begingroup
\tikzset{tikzfig/.style={baseline=-0.25em,scale=\echellefils,every node/.style={scale=0.8},borddiagrammevide/.style={-, dash pattern=on 0.35em off 0.35em on 0.35em off 0.35em on 0.35em off 0em}}}

\begin{multicols}{2}
\begin{equation}\tag{\ref*{idbox}}\begin{array}{rcl}\begin{tikzpicture}
	\begin{pgfonlayer}{nodelayer}
		\node [style=none] (0) at (-2, 0) {};
		\node [style=none] (1) at (2, 0) {};
	\end{pgfonlayer}
	\begin{pgfonlayer}{edgelayer}
		\draw [style=new] (0.center) to (1.center);
	\end{pgfonlayer}
\end{tikzpicture}&=&\begin{tikzpicture}
	\begin{pgfonlayer}{nodelayer}
		\node [style=newe] (0) at (0, 0) {$I$};
		\node [style=none] (1) at (-1.5, 0) {};
		\node [style=none] (2) at (1.5, 0) {};
	\end{pgfonlayer}
	\begin{pgfonlayer}{edgelayer}
		\draw [style=new] (1.center) to (2.center);
	\end{pgfonlayer}
\end{tikzpicture}\end{array}\end{equation}

\begin{equation}\notag\begin{array}{rcl}\begin{tikzpicture}
	\begin{pgfonlayer}{nodelayer}
		\node [style=cartouche] (0) at (0, 0) {$\neg$};
		\node [style=none] (1) at (-1.5, 0) {};
		\node [style=none] (2) at (1.5, 0) {};
	\end{pgfonlayer}
	\begin{pgfonlayer}{edgelayer}
		\draw [style=noire] (1.center) to (0);
		\draw [style=noire] (0) to (2.center);
	\end{pgfonlayer}
\end{tikzpicture}
&=&\begin{tikzpicture}
	\begin{pgfonlayer}{nodelayer}
		\node [style=cartouche] (0) at (0, 0) {$\neg$};
		\node [style=none] (1) at (-1.5, 0) {};
		\node [style=none] (2) at (1.5, 0) {};
	\end{pgfonlayer}
	\begin{pgfonlayer}{edgelayer}
		\draw [style=noire] (1.center) to (0);
		\draw [style=noire] (0) to (2.center);
	\end{pgfonlayer}
\end{tikzpicture}
\end{array}\end{equation}

\begin{equation}\notag\begin{array}{rcl}\input{beamsplitter.tikz}&=&\input{beamsplitter.tikz}\end{array}\end{equation}

\begin{equation}\label{boucletraversedupliquebs}\tag{\ref*{duplicateloop}'}\begin{array}{rcl}\input{boucletraverseidvariante.tikz}&=&\input{filtrerienrienetrien.tikz}
\end{array}\end{equation}

\begin{equation}\tag{\ref*{bsnnnn}}\begin{array}{rcl}\input{beamsplitternnnn.tikz}&=&\input{beamsplitterswap.tikz}
\end{array}\end{equation}

\begin{equation}\notag\begin{array}{rcl}\begin{tikzpicture}
	\begin{pgfonlayer}{nodelayer}
		\node [style=none] (0) at (-2, 0) {};
		\node [style=none] (1) at (2, 0) {};
	\end{pgfonlayer}
	\begin{pgfonlayer}{edgelayer}
		\draw [style=new] (0.center) to (1.center);
	\end{pgfonlayer}
\end{tikzpicture}&=&\begin{tikzpicture}
	\begin{pgfonlayer}{nodelayer}
		\node [style=none] (0) at (-2, 0) {};
		\node [style=none] (1) at (2, 0) {};
	\end{pgfonlayer}
	\begin{pgfonlayer}{edgelayer}
		\draw [style=new] (0.center) to (1.center);
	\end{pgfonlayer}
\end{tikzpicture}\end{array}\end{equation}

\begin{equation}\label{bouclevide1}\tag{\ref*{loopemptysimple}'}\begin{array}{rcl}\input{bouclevide.tikz}&=&\begin{tikzpicture}
	\begin{pgfonlayer}{nodelayer}
		\node [style=diagrammevide, scale=0.875] (0) at (0, 0) {};
	\end{pgfonlayer}
\end{tikzpicture}
\end{array}\end{equation}

\begin{equation}\tag{\ref*{bsbs}}\begin{array}{rcl}\input{beamsplitterbeamsplitter.tikz}&=&\begin{tikzpicture}
	\begin{pgfonlayer}{nodelayer}
		\node [style=none] (0) at (-3, 1) {};
		\node [style=none] (1) at (3, 1) {};
		\node [style=none] (2) at (3, -1) {};
		\node [style=none] (3) at (-3, -1) {};
	\end{pgfonlayer}
	\begin{pgfonlayer}{edgelayer}
		\draw [style=new] (0.center) to (1.center);
		\draw [style=new] (3.center) to (2.center);
	\end{pgfonlayer}
\end{tikzpicture}
\end{array}\end{equation}

\[\]

\begin{equation}\tag{\ref*{bsnbsh}}\begin{array}{rcl}\input{beamsplitternhautbeamsplitter.tikz}&=&\input{beamsplitternnhaut.tikz}
\end{array}\end{equation}
\end{multicols}

Equation \eqref{idbox} is now useless since it only allows us to create and remove $I$ gates without changing anything else, and neither the other equations nor Equation \eqref{bsbsbs} contain gates. Equations that have become an instance of reflexivity are now useless too. Finally, Equation \eqref{boucletraversedupliquebs} can be simplified through Equations \eqref{bsbs} and \eqref{bouclevide1} into Equation \eqref{loopid}. Thus, what we have to prove is that Equation \eqref{bsbsbs} is not a consequence of the following equations:

\begin{multicols}{2}

\[\]

\begin{equation}\label{loopid}\begin{array}{rcl}\input{boucletraverseidvariante.tikz}&=&\begin{tikzpicture}
	\begin{pgfonlayer}{nodelayer}
		\node [style=none] (0) at (-2, 0) {};
		\node [style=none] (1) at (2, 0) {};
	\end{pgfonlayer}
	\begin{pgfonlayer}{edgelayer}
		\draw [style=new] (0.center) to (1.center);
	\end{pgfonlayer}
\end{tikzpicture}
\end{array}\end{equation}

\begin{equation}\tag{\ref*{bsnnnn}}\begin{array}{rcl}\input{beamsplitternnnn.tikz}&=&\input{beamsplitterswap.tikz}
\end{array}\end{equation}

\begin{equation}\label{bouclevide}\tag{\ref*{bouclevide1}}\begin{array}{rcl}\input{bouclevide.tikz}&=&\begin{tikzpicture}
	\begin{pgfonlayer}{nodelayer}
		\node [style=diagrammevide, scale=0.875] (0) at (0, 0) {};
	\end{pgfonlayer}
\end{tikzpicture}
\end{array}\end{equation}

\begin{equation}\tag{\ref*{bsbs}}\begin{array}{rcl}\input{beamsplitterbeamsplitter.tikz}&=&\begin{tikzpicture}
	\begin{pgfonlayer}{nodelayer}
		\node [style=none] (0) at (-3, 1) {};
		\node [style=none] (1) at (3, 1) {};
		\node [style=none] (2) at (3, -1) {};
		\node [style=none] (3) at (-3, -1) {};
	\end{pgfonlayer}
	\begin{pgfonlayer}{edgelayer}
		\draw [style=new] (0.center) to (1.center);
		\draw [style=new] (3.center) to (2.center);
	\end{pgfonlayer}
\end{tikzpicture}
\end{array}\end{equation}

\begin{equation}\tag{\ref*{bsnbsh}}\begin{array}{rcl}\input{beamsplitternhautbeamsplitter.tikz}&=&\input{beamsplitternnhaut.tikz}
\end{array}\end{equation}
\end{multicols}

In the rest of the proof, we no longer assume the yanking and inverse law axioms, but we consider the corresponding equations instead:

\begin{equation}\tag{$y$}\label{yankingeq}\begin{array}{rcl}\input{yanking.tikz}&=&\begin{tikzpicture}
	\begin{pgfonlayer}{nodelayer}
		\node [style=none] (0) at (-2, 0) {};
		\node [style=none] (1) at (2, 0) {};
	\end{pgfonlayer}
	\begin{pgfonlayer}{edgelayer}
		\draw [style=new] (0.center) to (1.center);
	\end{pgfonlayer}
\end{tikzpicture}\end{array}\end{equation}

\begin{equation}\tag{$\sigma\sigma$}\label{xxeq}\begin{array}{rcl}\input{swapswap.tikz}&=&\input{filsparalleleslongbs.tikz}\end{array}\end{equation}

We have to prove that Equation \eqref{bsbsbs} is not a consequence of Equations \eqref{loopid}, \eqref{bsnnnn}, \eqref{bouclevide}, \eqref{bsbs}, \eqref{bsnbsh}, \eqref{yankingeq} and \eqref{xxeq}, still assuming the other axioms of the traced PROP, which by \cref{addswapswapyanking} are the axioms of a PROTWEB.

We also consider the notion of sub-diagram with respect to the axioms of a PROTWEB, that is, in \cref{defsubdiagram}, the conditions are considered up to these axioms. Intuitively, a sub-diagram in this sense is a part of a diagram that can be separated from the rest of the diagram by drawing a box around it.

Let us say that a diagram is circle-free if it does not have nonempty $0\to0$ sub-diagrams. Intuitively, a $0\to 0$ sub-diagram in the context of the PROTWEB is represented graphically as a union of connected components, which cannot be reached by a photon and do not affect the semantics of a diagram.

We consider the following set of rewriting rules on the set of gate-free diagrams: 

\setcounter{equation}{0}

\begin{equation}\label{rule00}\begin{array}{rcll}D&\rightarrow&\begin{tikzpicture}
	\begin{pgfonlayer}{nodelayer}
		\node [style=diagrammevide, scale=0.875] (0) at (0, 0) {};
	\end{pgfonlayer}
\end{tikzpicture}
&\text{ for every nonempty diagram $D:0\to0$}
\end{array}\end{equation}

\begin{equation}\label{rule11id}\begin{array}{rcll}\begin{tikzpicture}
	\begin{pgfonlayer}{nodelayer}
		\node [style=cartouche] (0) at (0, 0) {$D$};
		\node [style=none] (1) at (-1.5, 0) {};
		\node [style=none] (2) at (1.5, 0) {};
	\end{pgfonlayer}
	\begin{pgfonlayer}{edgelayer}
		\draw [style=noire] (1.center) to (0);
		\draw [style=noire] (0) to (2.center);
	\end{pgfonlayer}
\end{tikzpicture}&\rightarrow&\begin{tikzpicture}
	\begin{pgfonlayer}{nodelayer}
		\node [style=none] (0) at (-1.5, 0) {};
		\node [style=none] (1) at (1.5, 0) {};
	\end{pgfonlayer}
	\begin{pgfonlayer}{edgelayer}
		\draw [style=noire] (0.center) to (1.center);
	\end{pgfonlayer}
\end{tikzpicture}&\text{ for every circle-free $D:1\to1$ such that $D\neq\begin{tikzpicture}
	\begin{pgfonlayer}{nodelayer}
		\node [style=none] (0) at (-0.3, 0) {};
		\node [style=none] (1) at (0.3, 0) {};
	\end{pgfonlayer}
	\begin{pgfonlayer}{edgelayer}
		\draw [style=noire] (0.center) to (1.center);
	\end{pgfonlayer}
\end{tikzpicture}$ and $\interp{D}=\Id$}
\end{array}\end{equation}

\begin{equation}\label{rule11neg}\begin{array}{rcll}\begin{tikzpicture}
	\begin{pgfonlayer}{nodelayer}
		\node [style=cartouche] (0) at (0, 0) {$D$};
		\node [style=none] (1) at (-1.5, 0) {};
		\node [style=none] (2) at (1.5, 0) {};
	\end{pgfonlayer}
	\begin{pgfonlayer}{edgelayer}
		\draw [style=noire] (1.center) to (0);
		\draw [style=noire] (0) to (2.center);
	\end{pgfonlayer}
\end{tikzpicture}&\rightarrow&\begin{tikzpicture}
	\begin{pgfonlayer}{nodelayer}
		\node [style=cartouche] (0) at (0, 0) {$\neg$};
		\node [style=none] (1) at (-1.5, 0) {};
		\node [style=none] (2) at (1.5, 0) {};
	\end{pgfonlayer}
	\begin{pgfonlayer}{edgelayer}
		\draw [style=noire] (1.center) to (0);
		\draw [style=noire] (0) to (2.center);
	\end{pgfonlayer}
\end{tikzpicture}
&\text{ for every circle-free $D:1\to1$ such that $D\neq\begin{tikzpicture}[scale=0.8]
	\begin{pgfonlayer}{nodelayer}
		\node [style=cartouche] (0) at (0, 0) {\tiny$\neg$};
		\node [style=none] (1) at (-0.8, 0) {};
		\node [style=none] (2) at (0.8, 0) {};
	\end{pgfonlayer}
	\begin{pgfonlayer}{edgelayer}
		\draw [style=noire] (1.center) to (0);
		\draw [style=noire] (0) to (2.center);
	\end{pgfonlayer}
\end{tikzpicture}
$ and $\interp{D}=\interp{\begin{tikzpicture}[scale=0.8]
	\begin{pgfonlayer}{nodelayer}
		\node [style=cartouche] (0) at (0, 0) {\tiny$\neg$};
		\node [style=none] (1) at (-0.8, 0) {};
		\node [style=none] (2) at (0.8, 0) {};
	\end{pgfonlayer}
	\begin{pgfonlayer}{edgelayer}
		\draw [style=noire] (1.center) to (0);
		\draw [style=noire] (0) to (2.center);
	\end{pgfonlayer}
\end{tikzpicture}
}$}
\end{array}\end{equation}

\begin{equation}\label{ruleroundtrip}\begin{array}{rcl}\input{granddiagrammetroud22.tikz}&\rightarrow&\input{granddiagrammetrouid22.tikz}
\end{array}\end{equation}
for every diagram $D$ with a circle-free, nonidentity sub-diagram $d:2\to 2$ that we can slide along its two wires inside $D$, by using the axioms of the PROTWEB, in a constant direction and make it come back to the initial point, without having to use naturality of the swap or dinaturality to move anything else than $d$ while doing so

\begin{equation}\label{rulebsbs}\begin{array}{rcl}\input{beamsplitterbeamsplitter.tikz}&\rightarrow&\begin{tikzpicture}
	\begin{pgfonlayer}{nodelayer}
		\node [style=none] (0) at (-3, 1) {};
		\node [style=none] (1) at (3, 1) {};
		\node [style=none] (2) at (3, -1) {};
		\node [style=none] (3) at (-3, -1) {};
	\end{pgfonlayer}
	\begin{pgfonlayer}{edgelayer}
		\draw [style=new] (0.center) to (1.center);
		\draw [style=new] (3.center) to (2.center);
	\end{pgfonlayer}
\end{tikzpicture}
\end{array}\end{equation}

\begin{equation}\label{rulebsnbsh}\begin{array}{rcl}\input{beamsplitternhautbeamsplitter.tikz}&\rightarrow&\input{beamsplitternnhaut.tikz}
\end{array}\end{equation}

\begin{equation}\begin{array}{rcl}\input{beamsplitternbasbeamsplitter.tikz}&\rightarrow&\input{beamsplitternnbas.tikz}
\end{array}\end{equation}

\begin{equation}\label{rulebsnnbs}\begin{array}{rcl}\input{bsnnbs.tikz}&\rightarrow&\input{nnswap.tikz}
\end{array}\end{equation}

\begin{equation}\label{rulebsx}\begin{array}{rcl}\input{beamsplitterswap.tikz}&\rightarrow&\input{beamsplitternnnn.tikz}\end{array}\end{equation}

\begin{equation}\begin{array}{rcl}\input{swapbeamsplitter.tikz}&\rightarrow&\input{beamsplitternnnn.tikz}\end{array}\end{equation}

\begin{equation}\label{rulexx}\begin{array}{rcl}\input{swapswap.tikz}&\rightarrow&\begin{tikzpicture}
	\begin{pgfonlayer}{nodelayer}
		\node [style=none] (0) at (-3, 1) {};
		\node [style=none] (1) at (3, 1) {};
		\node [style=none] (2) at (3, -1) {};
		\node [style=none] (3) at (-3, -1) {};
	\end{pgfonlayer}
	\begin{pgfonlayer}{edgelayer}
		\draw [style=new] (0.center) to (1.center);
		\draw [style=new] (3.center) to (2.center);
	\end{pgfonlayer}
\end{tikzpicture}\end{array}\end{equation}

It is easy to see that these rules preserve the semantics.\\
\begin{remark}\label{reduction11toujourspossible}
Any gate-free $1\to 1$ diagram is interpreted as $\Id$ or $\interp{\begin{tikzpicture}[scale=0.8]
	\begin{pgfonlayer}{nodelayer}
		\node [style=cartouche] (0) at (0, 0) {\tiny$\neg$};
		\node [style=none] (1) at (-0.8, 0) {};
		\node [style=none] (2) at (0.8, 0) {};
	\end{pgfonlayer}
	\begin{pgfonlayer}{edgelayer}
		\draw [style=noire] (1.center) to (0);
		\draw [style=noire] (0) to (2.center);
	\end{pgfonlayer}
\end{tikzpicture}
}$, so it can be reduced to $\begin{tikzpicture}
	\begin{pgfonlayer}{nodelayer}
		\node [style=none] (0) at (-0.3, 0) {};
		\node [style=none] (1) at (0.3, 0) {};
	\end{pgfonlayer}
	\begin{pgfonlayer}{edgelayer}
		\draw [style=noire] (0.center) to (1.center);
	\end{pgfonlayer}
\end{tikzpicture}$ or $\begin{tikzpicture}[scale=0.8]
	\begin{pgfonlayer}{nodelayer}
		\node [style=cartouche] (0) at (0, 0) {\tiny$\neg$};
		\node [style=none] (1) at (-0.8, 0) {};
		\node [style=none] (2) at (0.8, 0) {};
	\end{pgfonlayer}
	\begin{pgfonlayer}{edgelayer}
		\draw [style=noire] (1.center) to (0);
		\draw [style=noire] (0) to (2.center);
	\end{pgfonlayer}
\end{tikzpicture}
$ by first applying Rule \eqref{rule00} repeatedly to remove all its $0\to0$ sub-diagrams, then applying Rule \eqref{rule11id} or \eqref{rule11neg}.
\end{remark}

Since all diagrams have their number of input wires equal to their number of output wires, the axioms of the PROTWEB do not change the number of \begin{tikzpicture}
	\begin{pgfonlayer}{nodelayer}
		\node [style=bs] (0) at (0, 0) {};
	\end{pgfonlayer}
\end{tikzpicture}, \input{swap-s.tikz}, \begin{tikzpicture}
	\begin{pgfonlayer}{nodelayer}
		\node [style=cartouche] (0) at (0, 0) {$\neg$};
	\end{pgfonlayer}
\end{tikzpicture} or of trace wires in a diagram, so these numbers are well-defined for a given diagram. This allows us to define the \emph{level} of a diagram as a tuple $(b,x,n,t)$, where:
\begin{itemize}
\item $b$ is the number of \begin{tikzpicture}
	\begin{pgfonlayer}{nodelayer}
		\node [style=bs] (0) at (0, 0) {};
	\end{pgfonlayer}
\end{tikzpicture}\\
\item $x$ is the number of \input{swap-s.tikz}\\
\item $n$ is the number of \begin{tikzpicture}
	\begin{pgfonlayer}{nodelayer}
		\node [style=cartouche] (0) at (0, 0) {$\neg$};
	\end{pgfonlayer}
\end{tikzpicture}\\
\item $t$ is the number of trace wires.
\end{itemize}

It is easy to check that each of the rewriting rules strictly decreases the level, according to the lexicographic order. Since the lexicographic order on $\N^4$ is well-founded, this implies that the rewriting system is strongly normalising.\\

Let us prove that the rewriting system is confluent. Because of strong normalisation, it suffices to prove that it is locally confluent. Let $\rightarrow^*$ be the reflexive transitive closure of $\rightarrow$. Let $D$ be a diagram and let $D\overset{(a)}{\rightarrow}D_1$ and $D\overset{(b)}{\rightarrow}D_2$ be two reduction steps, where $(a)$ and $(b)$ are the respective rules applied. We have to prove that there exists a diagram $D'$ such that $D_1\rightarrow^* D'$ and $D_2\rightarrow^* D'$.\\

If the two patterns in $D$ that are transformed by $(a)$ and $(b)$ do not overlap, then after applying $(a)$ to the first pattern or $(b)$ to the second one, we can still apply the other rule to the other pattern and the final result does not depend on the order in which $(a)$ and $(b)$ are applied. That is, there exists $D'$ such that $D_1\overset{(b)}{\rightarrow}D'$ and $D_2\overset{(a)}{\rightarrow}D'$.\\

In the following case distinction, we assume that the patterns concerned by $(a)$ and $(b)$ overlap.\\

It is easy to see that if $(a)$ is \eqref{rule00}, \eqref{rule11id} or \eqref{rule11neg} and $(b)$ is among Rules \eqref{rulebsbs} to \eqref{rulexx}, then the only way the concerned patterns in $D$ can overlap is that the pattern concerned by $(b)$ is included in this of $(a)$. In this case, on the one hand, $(a)$ transforms its pattern into $\begin{tikzpicture}[scale=0.3]
	\begin{pgfonlayer}{nodelayer}
		\node [style=xsdiagrammevide] (0) at (0, 0) {};
	\end{pgfonlayer}
\end{tikzpicture}
$, $\begin{tikzpicture}
	\begin{pgfonlayer}{nodelayer}
		\node [style=none] (0) at (-0.3, 0) {};
		\node [style=none] (1) at (0.3, 0) {};
	\end{pgfonlayer}
	\begin{pgfonlayer}{edgelayer}
		\draw [style=noire] (0.center) to (1.center);
	\end{pgfonlayer}
\end{tikzpicture}$ or $\begin{tikzpicture}[scale=0.8]
	\begin{pgfonlayer}{nodelayer}
		\node [style=cartouche] (0) at (0, 0) {\tiny$\neg$};
		\node [style=none] (1) at (-0.8, 0) {};
		\node [style=none] (2) at (0.8, 0) {};
	\end{pgfonlayer}
	\begin{pgfonlayer}{edgelayer}
		\draw [style=noire] (1.center) to (0);
		\draw [style=noire] (0) to (2.center);
	\end{pgfonlayer}
\end{tikzpicture}
$, and on the other hand, the effect of applying $(b)$ is to transform the pattern of $(a)$ into a semantically equivalent diagram (which is not $\begin{tikzpicture}[scale=0.3]
	\begin{pgfonlayer}{nodelayer}
		\node [style=xsdiagrammevide] (0) at (0, 0) {};
	\end{pgfonlayer}
\end{tikzpicture}
$, $\begin{tikzpicture}
	\begin{pgfonlayer}{nodelayer}
		\node [style=none] (0) at (-0.3, 0) {};
		\node [style=none] (1) at (0.3, 0) {};
	\end{pgfonlayer}
	\begin{pgfonlayer}{edgelayer}
		\draw [style=noire] (0.center) to (1.center);
	\end{pgfonlayer}
\end{tikzpicture}$ or $\begin{tikzpicture}[scale=0.8]
	\begin{pgfonlayer}{nodelayer}
		\node [style=cartouche] (0) at (0, 0) {\tiny$\neg$};
		\node [style=none] (1) at (-0.8, 0) {};
		\node [style=none] (2) at (0.8, 0) {};
	\end{pgfonlayer}
	\begin{pgfonlayer}{edgelayer}
		\draw [style=noire] (1.center) to (0);
		\draw [style=noire] (0) to (2.center);
	\end{pgfonlayer}
\end{tikzpicture}
$ because it contains at least a trace), which can then be transformed into $\begin{tikzpicture}[scale=0.3]
	\begin{pgfonlayer}{nodelayer}
		\node [style=xsdiagrammevide] (0) at (0, 0) {};
	\end{pgfonlayer}
\end{tikzpicture}
$, $\begin{tikzpicture}
	\begin{pgfonlayer}{nodelayer}
		\node [style=none] (0) at (-0.3, 0) {};
		\node [style=none] (1) at (0.3, 0) {};
	\end{pgfonlayer}
	\begin{pgfonlayer}{edgelayer}
		\draw [style=noire] (0.center) to (1.center);
	\end{pgfonlayer}
\end{tikzpicture}$ or $\begin{tikzpicture}[scale=0.8]
	\begin{pgfonlayer}{nodelayer}
		\node [style=cartouche] (0) at (0, 0) {\tiny$\neg$};
		\node [style=none] (1) at (-0.8, 0) {};
		\node [style=none] (2) at (0.8, 0) {};
	\end{pgfonlayer}
	\begin{pgfonlayer}{edgelayer}
		\draw [style=noire] (1.center) to (0);
		\draw [style=noire] (0) to (2.center);
	\end{pgfonlayer}
\end{tikzpicture}
$ by applying $(a)$. Since the rules preserve the semantics, the final sub-diagrams obtained in each case are the same. Therefore, $D_2\overset{(a)}{\rightarrow}D_1$. Of course, the same argument applies with $(a)$ and $(b)$ exchanged.\\

If $(a)$ is \eqref{rule00} and $(b)$ is \eqref{rule11id} or \eqref{rule11neg}, then since the pattern concerned by $(b)$ does not contain any $0\to0$ sub-diagram, it is necessarily included in the pattern concerned by $(a)$, which, after applying $(b)$, can still be transformed into the empty diagram by applying \eqref{rule00}. Therefore, $D_2\overset{\eqref{rule00}}{\rightarrow}D_1$. The same argument applies with $(a)$ and $(b)$ exchanged.\\

If both $(a)$ and $(b)$ are Rule \eqref{rule00}, then the union of the two patterns concerned by $(a)$ and $(b)$ is a $0\to0$ sub-diagram of $D$. Applying $(a)$ or $(b)$ does not change the fact that it is of type $0\to0$, so that right after that we can transform it into the empty diagram by applying Rule \eqref{rule00} (unless it has already become empty in which case there is nothing more to do). This gives us the desired $D'$\\

If both $(a)$ and $(b)$ each are Rule \eqref{rule11id} or \eqref{rule11neg}, then the union of the two concerned patterns can be written in the form $d_2\circ d\circ d_1$ in such a way that, up to exchanging the roles of $(a)$ and $(b)$, the pattern concerned by $(a)$ is $d\circ d_1$ and the pattern concerned by $(b)$ is $d_2\circ d$. Therefore, after applying $(a)$ or $(b)$, we can apply Rule \eqref{rule11id} or \eqref{rule11neg} to transform the resulting whole sub-diagram into $\begin{tikzpicture}
	\begin{pgfonlayer}{nodelayer}
		\node [style=none] (0) at (-0.3, 0) {};
		\node [style=none] (1) at (0.3, 0) {};
	\end{pgfonlayer}
	\begin{pgfonlayer}{edgelayer}
		\draw [style=noire] (0.center) to (1.center);
	\end{pgfonlayer}
\end{tikzpicture}$ or $\begin{tikzpicture}[scale=0.8]
	\begin{pgfonlayer}{nodelayer}
		\node [style=cartouche] (0) at (0, 0) {\tiny$\neg$};
		\node [style=none] (1) at (-0.8, 0) {};
		\node [style=none] (2) at (0.8, 0) {};
	\end{pgfonlayer}
	\begin{pgfonlayer}{edgelayer}
		\draw [style=noire] (1.center) to (0);
		\draw [style=noire] (0) to (2.center);
	\end{pgfonlayer}
\end{tikzpicture}
$, and since the rules preserve the semantics, the result is the same regardless of whether $(a)$ or $(b)$ was applied. This gives us the desired $D'$.\\

If $(a)$ is Rule \eqref{ruleroundtrip}, then:
\begin{itemize}
\item if $(b)$ is Rule \eqref{rule00}, then since $d$ is circle-free, it does not intersect the pattern concerned by $(b)$. Therefore, the situation is the same as when the two patterns do not overlap and there exists $D'$ such that $D_1\overset{(b)}{\rightarrow}D'$ and $D_2\overset{(a)}{\rightarrow}D'$.
\item if $(b)$ is Rule \eqref{rule11id} or \eqref{rule11neg}, then the condition of Rule \eqref{ruleroundtrip} implies that the pattern concerned by $(b)$ either is included in $d$, in which case we have $D_2\overset{\eqref{ruleroundtrip}}{\rightarrow}D_1$, or contains $d$ as a sub-diagram, in which case we have $D_1\overset{(b)}{\rightarrow}D_2$, or is disjoint from it, in which case we are in the same situation as when the two patterns do not overlap and there exists $D'$ such that $D_1\overset{(b)}{\rightarrow}D'$ and $D_2\overset{(a)}{\rightarrow}D'$.
\item if $(b)$ is Rule \eqref{ruleroundtrip} too, then $(a)$ and $(b)$ each transform an instance of $d$ into the identity. After this, the other instance of $d$ can be transformed into the identity by applying Rule \eqref{ruleroundtrip} again (unless it has already become equal to the identity), and the result is the same regardless of whether $(a)$ or $(b)$ was applied. This gives us the desired $D'$.
\item if $(b)$ is among Rules \eqref{rulebsbs} to \eqref{rulexx}, then the condition of Rule \eqref{ruleroundtrip} implies that the pattern concerned by $(b)$ is either included in $d$, in which case we have $D_2\overset{\eqref{ruleroundtrip}}{\rightarrow}D_1$, or disjoint from it, in which case we are in the same situation as when the two patterns do not overlap and there exists $D'$ such that $D_1\overset{(b)}{\rightarrow}D'$ and $D_2\overset{(a)}{\rightarrow}D'$.
\end{itemize}

If both $(a)$ and $(b)$ are among Rules \eqref{rulebsbs} to \eqref{rulexx}, then by looking at the possible left-hand sides of these rules, we can see that unless they are the same and $D_1=D_2$, the two patterns cannot have a \begin{tikzpicture}
	\begin{pgfonlayer}{nodelayer}
		\node [style=cartouche] (0) at (0, 0) {$\neg$};
	\end{pgfonlayer}
\end{tikzpicture} in common, and any generator in common cannot be the leftmost one in both patterns, nor the rightmost one. So the cases to consider are:
\begin{itemize}
\item those in which the two patterns have one generator in common, which is on the right of one pattern and on the left of the other
\item those in which the two patterns have two generators in common, the leftmost generator of each pattern being the rightmost one of the other pattern.
\end{itemize}

The first possibility means that the two patterns in $D$ are in a sub-diagram of one of the following forms:
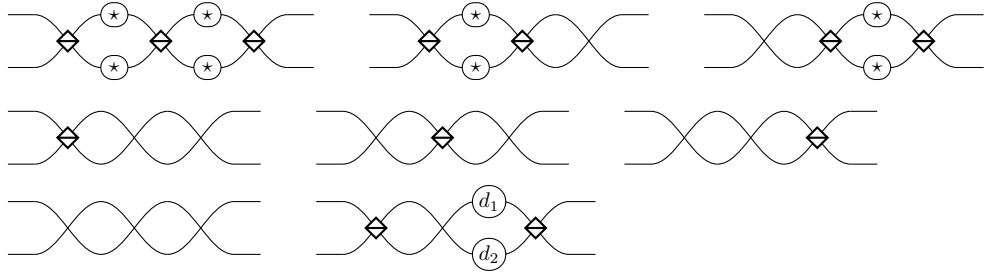
\[\input{bsphpbbsphpbbs.tikz}\qquad\input{bsphpbbsx.tikz}\qquad\input{xbsphpbbs.tikz}\]
\[\input{bsxx.tikz}\qquad\input{xbsx.tikz}\qquad\input{xxbs.tikz}\]
\[\input{swapswapswap.tikz}\qquad\input{bsxbsavecdgeneriquesaumilieu.tikz}\]

where \begin{tikzpicture}
	\begin{pgfonlayer}{nodelayer}
		\node [style=none] (0) at (-1.5, 0) {};
		\node [style=none] (1) at (1.5, 0) {};
		\node [style=negpotentiel] (2) at (0, 0) {};
	\end{pgfonlayer}
	\begin{pgfonlayer}{edgelayer}
		\draw [style=noire] (2) to (1.center);
		\draw [style=noire] (0.center) to (2);
	\end{pgfonlayer}
\end{tikzpicture} denotes either \begin{tikzpicture}
	\begin{pgfonlayer}{nodelayer}
		\node [style=cartouche] (0) at (0, 0) {$\neg$};
		\node [style=none] (1) at (-1.5, 0) {};
		\node [style=none] (2) at (1.5, 0) {};
	\end{pgfonlayer}
	\begin{pgfonlayer}{edgelayer}
		\draw [style=noire] (1.center) to (0);
		\draw [style=noire] (0) to (2.center);
	\end{pgfonlayer}
\end{tikzpicture}
 or \begin{tikzpicture}
	\begin{pgfonlayer}{nodelayer}
		\node [style=none] (0) at (-1.5, 0) {};
		\node [style=none] (1) at (1.5, 0) {};
	\end{pgfonlayer}
	\begin{pgfonlayer}{edgelayer}
		\draw [style=noire] (0.center) to (1.center);
	\end{pgfonlayer}
\end{tikzpicture}, and $d_1,d_2:1\to1$ are arbitrary diagrams.

$D_1$ and $D_2$ are obtained from $D$ by applying one of the Rules \eqref{rulebsbs} to \eqref{rulexx}, to the left part of the sub-diagram for one of the two, and to the right part of the sub-diagram for the other (possibly after sliding $d_1$ and $d_2$ through the swap by naturality of it). To reduce them to a common diagram, we still focus on the same sub-diagram. If relevant, we reduce $d_1$ and $d_2$ to $\begin{tikzpicture}
	\begin{pgfonlayer}{nodelayer}
		\node [style=none] (0) at (-0.3, 0) {};
		\node [style=none] (1) at (0.3, 0) {};
	\end{pgfonlayer}
	\begin{pgfonlayer}{edgelayer}
		\draw [style=noire] (0.center) to (1.center);
	\end{pgfonlayer}
\end{tikzpicture}$ or $\begin{tikzpicture}[scale=0.8]
	\begin{pgfonlayer}{nodelayer}
		\node [style=cartouche] (0) at (0, 0) {\tiny$\neg$};
		\node [style=none] (1) at (-0.8, 0) {};
		\node [style=none] (2) at (0.8, 0) {};
	\end{pgfonlayer}
	\begin{pgfonlayer}{edgelayer}
		\draw [style=noire] (1.center) to (0);
		\draw [style=noire] (0) to (2.center);
	\end{pgfonlayer}
\end{tikzpicture}
$ as described in \cref{reduction11toujourspossible}. Otherwise we apply Rule \eqref{rule11id} to all double negations to remove them. Then, if there are still two generators of type \input{beamsplitter-s.tikz} or \input{swap-s.tikz}, we apply the appropriate rule among \eqref{rulebsbs} to \eqref{rulexx}, and finally we apply Rule \eqref{rule11id} repeatedly to all resulting double negations to remove them. After that, the sub-diagram is of the form
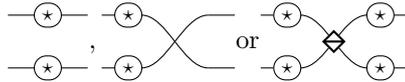
\[\input{negpotentielsurnegpotentiel.tikz}, \input{negnegpotentielsswap.tikz}\text{or}\input{bsnnnnpotentiels.tikz}\]
where \begin{tikzpicture}
	\begin{pgfonlayer}{nodelayer}
		\node [style=none] (0) at (-1.5, 0) {};
		\node [style=none] (1) at (1.5, 0) {};
		\node [style=negpotentiel] (2) at (0, 0) {};
	\end{pgfonlayer}
	\begin{pgfonlayer}{edgelayer}
		\draw [style=noire] (2) to (1.center);
		\draw [style=noire] (0.center) to (2);
	\end{pgfonlayer}
\end{tikzpicture} denotes either \begin{tikzpicture}
	\begin{pgfonlayer}{nodelayer}
		\node [style=cartouche] (0) at (0, 0) {$\neg$};
		\node [style=none] (1) at (-1.5, 0) {};
		\node [style=none] (2) at (1.5, 0) {};
	\end{pgfonlayer}
	\begin{pgfonlayer}{edgelayer}
		\draw [style=noire] (1.center) to (0);
		\draw [style=noire] (0) to (2.center);
	\end{pgfonlayer}
\end{tikzpicture}
 or \begin{tikzpicture}
	\begin{pgfonlayer}{nodelayer}
		\node [style=none] (0) at (-1.5, 0) {};
		\node [style=none] (1) at (1.5, 0) {};
	\end{pgfonlayer}
	\begin{pgfonlayer}{edgelayer}
		\draw [style=noire] (0.center) to (1.center);
	\end{pgfonlayer}
\end{tikzpicture}. It is easy to see that two diagrams of these forms have the same semantics only if they are equal. And since the reduction rules preserve the semantics, the two final sub-diagrams must have the same semantics, so they are equal.\\

The second possibility means that the union of the two patterns is of the form
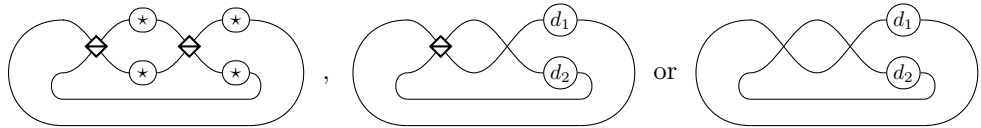
\[\input{motifchevauchementcirculairebsbs.tikz}, \input{motifchevauchementcirculairebsx.tikz}\text{or}\input{motifchevauchementcirculairexx.tikz}\]
where \begin{tikzpicture}
	\begin{pgfonlayer}{nodelayer}
		\node [style=none] (0) at (-1.5, 0) {};
		\node [style=none] (1) at (1.5, 0) {};
		\node [style=negpotentiel] (2) at (0, 0) {};
	\end{pgfonlayer}
	\begin{pgfonlayer}{edgelayer}
		\draw [style=noire] (2) to (1.center);
		\draw [style=noire] (0.center) to (2);
	\end{pgfonlayer}
\end{tikzpicture} still denotes either \begin{tikzpicture}
	\begin{pgfonlayer}{nodelayer}
		\node [style=cartouche] (0) at (0, 0) {$\neg$};
		\node [style=none] (1) at (-1.5, 0) {};
		\node [style=none] (2) at (1.5, 0) {};
	\end{pgfonlayer}
	\begin{pgfonlayer}{edgelayer}
		\draw [style=noire] (1.center) to (0);
		\draw [style=noire] (0) to (2.center);
	\end{pgfonlayer}
\end{tikzpicture}
 or \begin{tikzpicture}
	\begin{pgfonlayer}{nodelayer}
		\node [style=none] (0) at (-1.5, 0) {};
		\node [style=none] (1) at (1.5, 0) {};
	\end{pgfonlayer}
	\begin{pgfonlayer}{edgelayer}
		\draw [style=noire] (0.center) to (1.center);
	\end{pgfonlayer}
\end{tikzpicture}, and $d_1,d_2:1\to1$ are arbitrary diagrams. This union is not necessarily a sub-diagram of $D$. Indeed, on the one hand, there can be some $0\to 0$ diagrams inside the loop, and on the other hand we may have to use the naturality of the swap to transform each of the two patterns into the other, which means that there are external wires that intersect the union. However, in any case, after applying $(a)$ or $(b)$, we can apply Rule \eqref{ruleroundtrip} to transform it into \input{doublebouclevide.tikz}. This reduces $D_1$ and $D_2$ to a common diagram, and finishes proving that the rewriting system is confluent.\\

Transforming a diagram by applying Equation  \eqref{loopid}, \eqref{bsnnnn}, \eqref{bouclevide}, \eqref{bsbs}, \eqref{bsnbsh}, \eqref{yankingeq} or \eqref{xxeq} amounts to applying, or to applying the opposite of, Rule \eqref{rule11id}, \eqref{rulebsx}, \eqref{rule00}, \eqref{rulebsbs}, \eqref{rulebsnbsh}, \eqref{rule11id} or \eqref{rulexx} respectively. Therefore, if two diagrams $D_1$ and $D_2$ are equal according to these equations, they are equivalent according to the equivalence relation generated by the reduction relation $\rightarrow$. By confluence, this implies that there exists a diagram $D'$ such that $D_1\rightarrow^* D'$ and $D_2\rightarrow^* D'$. Since \input{bsbsbspointeenbas.tikz} and \input{xbsxpointeenbas.tikz} are normal forms for the rewriting system, this proves that they are not equal according to Equations \eqref{loopid}, \eqref{bsnnnn}, \eqref{bouclevide}, \eqref{bsbs}, \eqref{bsnbsh}, \eqref{yankingeq} and \eqref{xxeq}, and therefore that Equation \eqref{bsbsbs} is not a consequence of these equations, which is what we wanted to prove.

\endgroup

\end{proof}

\subsection{Removing the trace -- Loop unrolling}\label{proofremovingtrace}

\subsubsection{Proof of Proposition \ref{traceinutilesiunitaire}}\label{preuvetraceinutilesiunitaire}
By \cref{existenceofthenormalform}, there exists a diagram $N$ in normal form such that $\textup{PBS}\vdash D=N$. What we have to prove is that $N$ is equivalent through the axioms of the \textup{PBS}-calculus to a trace-free diagram. By remark \ref{decompNFsupperm}, let us decompose $N$ into $P\circ E$, where $E$ is of the form $E(U_0,V_0)\otimes\cdots\otimes E(U_{n-1},V_{n-1})$, and $P$ is trace-free and gate-free. We just have to prove that $E$ is equivalent to a trace-free diagram. By the axioms of the PROP, we can write $E$ in the form $E=\displaystyle\prod_{p=0}^{n-1}(\begin{tikzpicture}
	\begin{pgfonlayer}{nodelayer}
		\node [style=none] (0) at (-0.3, 0) {};
		\node [style=none] (1) at (0.3, 0) {};
	\end{pgfonlayer}
	\begin{pgfonlayer}{edgelayer}
		\draw [style=noire] (0.center) to (1.center);
	\end{pgfonlayer}
\end{tikzpicture}^{\otimes p}\otimes E(U_{\rightarrow,p},U_{\uparrow,p})\otimes\begin{tikzpicture}
	\begin{pgfonlayer}{nodelayer}
		\node [style=none] (0) at (-0.3, 0) {};
		\node [style=none] (1) at (0.3, 0) {};
	\end{pgfonlayer}
	\begin{pgfonlayer}{edgelayer}
		\draw [style=noire] (0.center) to (1.center);
	\end{pgfonlayer}
\end{tikzpicture}^{\otimes n-1-p})$, so it is sufficient to prove that every diagram $\begin{tikzpicture}
	\begin{pgfonlayer}{nodelayer}
		\node [style=none] (0) at (-0.3, 0) {};
		\node [style=none] (1) at (0.3, 0) {};
	\end{pgfonlayer}
	\begin{pgfonlayer}{edgelayer}
		\draw [style=noire] (0.center) to (1.center);
	\end{pgfonlayer}
\end{tikzpicture}^{\otimes p}\otimes E(U_{\rightarrow,p},U_{\uparrow,p})\otimes\begin{tikzpicture}
	\begin{pgfonlayer}{nodelayer}
		\node [style=none] (0) at (-0.3, 0) {};
		\node [style=none] (1) at (0.3, 0) {};
	\end{pgfonlayer}
	\begin{pgfonlayer}{edgelayer}
		\draw [style=noire] (0.center) to (1.center);
	\end{pgfonlayer}
\end{tikzpicture}^{\otimes n-1-p}$ is equivalent to a trace-free diagram. For this it is enough to prove that any diagram of the form $E(U,V)\otimes\begin{tikzpicture}
	\begin{pgfonlayer}{nodelayer}
		\node [style=none] (0) at (-0.3, 0) {};
		\node [style=none] (1) at (0.3, 0) {};
	\end{pgfonlayer}
	\begin{pgfonlayer}{edgelayer}
		\draw [style=noire] (0.center) to (1.center);
	\end{pgfonlayer}
\end{tikzpicture}$ or $\begin{tikzpicture}
	\begin{pgfonlayer}{nodelayer}
		\node [style=none] (0) at (-0.3, 0) {};
		\node [style=none] (1) at (0.3, 0) {};
	\end{pgfonlayer}
	\begin{pgfonlayer}{edgelayer}
		\draw [style=noire] (0.center) to (1.center);
	\end{pgfonlayer}
\end{tikzpicture}\otimes E(U,V)$ is equivalent to a trace-free diagram. And since $\begin{tikzpicture}
	\begin{pgfonlayer}{nodelayer}
		\node [style=none] (0) at (-0.3, 0) {};
		\node [style=none] (1) at (0.3, 0) {};
	\end{pgfonlayer}
	\begin{pgfonlayer}{edgelayer}
		\draw [style=noire] (0.center) to (1.center);
	\end{pgfonlayer}
\end{tikzpicture}\otimes E(U,V)=\input{swap-xs.tikz}\circ (E(U,V)\otimes\begin{tikzpicture}
	\begin{pgfonlayer}{nodelayer}
		\node [style=none] (0) at (-0.3, 0) {};
		\node [style=none] (1) at (0.3, 0) {};
	\end{pgfonlayer}
	\begin{pgfonlayer}{edgelayer}
		\draw [style=noire] (0.center) to (1.center);
	\end{pgfonlayer}
\end{tikzpicture})\circ\input{swap-xs.tikz}$, it suffices to prove that $E(U,V)\otimes\begin{tikzpicture}
	\begin{pgfonlayer}{nodelayer}
		\node [style=none] (0) at (-0.3, 0) {};
		\node [style=none] (1) at (0.3, 0) {};
	\end{pgfonlayer}
	\begin{pgfonlayer}{edgelayer}
		\draw [style=noire] (0.center) to (1.center);
	\end{pgfonlayer}
\end{tikzpicture}$ is equivalent to a trace-free diagram.

First, assume that $U$ and $V$ have a square root. Then $E(U,V)\otimes\begin{tikzpicture}
	\begin{pgfonlayer}{nodelayer}
		\node [style=none] (0) at (-0.3, 0) {};
		\node [style=none] (1) at (0.3, 0) {};
	\end{pgfonlayer}
	\begin{pgfonlayer}{edgelayer}
		\draw [style=noire] (0.center) to (1.center);
	\end{pgfonlayer}
\end{tikzpicture}$ is equivalent to\[\tikzset{tikzfig/.style={baseline=-0.25em,scale=\echellefils,every node/.style={scale=0.8},borddiagrammevide/.style={-, dash pattern=on \traitsdiagrammevide off \traitsdiagrammevide on \traitsdiagrammevide off \traitsdiagrammevide on \traitsdiagrammevide off 0em}}}\input{filtreUVtracefree.tikz}.\]
If $U$ or $V$ does not have a square root, let us consider their polar decompositions $U=QS$ and $V=Q'S'$ with $Q,Q'$ unitary and $S,S'$ positive-definite Hermitian. Then by Equation \eqref{fusionuv}, $PBS\vdash E(U,V)\otimes\begin{tikzpicture}
	\begin{pgfonlayer}{nodelayer}
		\node [style=none] (0) at (-0.3, 0) {};
		\node [style=none] (1) at (0.3, 0) {};
	\end{pgfonlayer}
	\begin{pgfonlayer}{edgelayer}
		\draw [style=noire] (0.center) to (1.center);
	\end{pgfonlayer}
\end{tikzpicture}=(E(Q,Q')\otimes\begin{tikzpicture}
	\begin{pgfonlayer}{nodelayer}
		\node [style=none] (0) at (-0.3, 0) {};
		\node [style=none] (1) at (0.3, 0) {};
	\end{pgfonlayer}
	\begin{pgfonlayer}{edgelayer}
		\draw [style=noire] (0.center) to (1.center);
	\end{pgfonlayer}
\end{tikzpicture})\circ (E(S,S')\otimes\begin{tikzpicture}
	\begin{pgfonlayer}{nodelayer}
		\node [style=none] (0) at (-0.3, 0) {};
		\node [style=none] (1) at (0.3, 0) {};
	\end{pgfonlayer}
	\begin{pgfonlayer}{edgelayer}
		\draw [style=noire] (0.center) to (1.center);
	\end{pgfonlayer}
\end{tikzpicture})$, and since each of $Q$, $S$, $Q'$ and $S'$ have a square root, $E(Q,Q')\otimes\begin{tikzpicture}
	\begin{pgfonlayer}{nodelayer}
		\node [style=none] (0) at (-0.3, 0) {};
		\node [style=none] (1) at (0.3, 0) {};
	\end{pgfonlayer}
	\begin{pgfonlayer}{edgelayer}
		\draw [style=noire] (0.center) to (1.center);
	\end{pgfonlayer}
\end{tikzpicture}$ and $E(S,S')\otimes\begin{tikzpicture}
	\begin{pgfonlayer}{nodelayer}
		\node [style=none] (0) at (-0.3, 0) {};
		\node [style=none] (1) at (0.3, 0) {};
	\end{pgfonlayer}
	\begin{pgfonlayer}{edgelayer}
		\draw [style=noire] (0.center) to (1.center);
	\end{pgfonlayer}
\end{tikzpicture}$ are equivalent to trace-free diagrams of the form above, so that by composition $E(U,V)\otimes\begin{tikzpicture}
	\begin{pgfonlayer}{nodelayer}
		\node [style=none] (0) at (-0.3, 0) {};
		\node [style=none] (1) at (0.3, 0) {};
	\end{pgfonlayer}
	\begin{pgfonlayer}{edgelayer}
		\draw [style=noire] (0.center) to (1.center);
	\end{pgfonlayer}
\end{tikzpicture}$ is equivalent to a trace-free diagram too.

\subsubsection{Proof of Lemma \ref{zerooudeuxnoninversibles}}\label{preuvezerooudeuxnoninversibles}
We prove the result by structural induction on $D$.

If $D=\begin{tikzpicture}[scale=0.3]
	\begin{pgfonlayer}{nodelayer}
		\node [style=xsdiagrammevide] (0) at (0, 0) {};
	\end{pgfonlayer}
\end{tikzpicture}
,\begin{tikzpicture}
	\begin{pgfonlayer}{nodelayer}
		\node [style=none] (0) at (-0.3, 0) {};
		\node [style=none] (1) at (0.3, 0) {};
	\end{pgfonlayer}
	\begin{pgfonlayer}{edgelayer}
		\draw [style=noire] (0.center) to (1.center);
	\end{pgfonlayer}
\end{tikzpicture},\begin{tikzpicture}[scale=0.8]
	\begin{pgfonlayer}{nodelayer}
		\node [style=cartouche] (0) at (0, 0) {\tiny$\neg$};
		\node [style=none] (1) at (-0.8, 0) {};
		\node [style=none] (2) at (0.8, 0) {};
	\end{pgfonlayer}
	\begin{pgfonlayer}{edgelayer}
		\draw [style=noire] (1.center) to (0);
		\draw [style=noire] (0) to (2.center);
	\end{pgfonlayer}
\end{tikzpicture}
,\input{swap-xs.tikz}\text{ or }\input{beamsplitter-xs.tikz}$ then for every $(c,p)$ we have $[D]_{c,p}=I_q$ which is invertible, so the result holds.

If $D=\begin{tikzpicture}[scale=0.6]
	\begin{pgfonlayer}{nodelayer}
		\node [style=newe] (0) at (0, 0) {\tiny\!$U$\!};
		\node [style=none] (1) at (-1, 0) {};
		\node [style=none] (2) at (1, 0) {};
	\end{pgfonlayer}
	\begin{pgfonlayer}{edgelayer}
		\draw [style=new] (1.center) to (2.center);
	\end{pgfonlayer}
\end{tikzpicture}$ then for every $c\in\hv$ we have $[D]_{c,0}=U$. If $U$ is invertible, then the result holds, and if $U$ is not invertible, then the result holds too.

If $D=D_2\circ D_1$, then for any $(c,p)$ we have $[D]_{c,p}=[D_2]_{\tau_{D_1}(c,p)}[D_1]_{c,p}$. The product $[D_2]_{\tau_{D_1}(c,p)}[D_1]_{c,p}$ is invertible if and only if both $[D_2]_{\tau_{D_1}(c,p)}$ and $[D_1]_{c,p}$ are. Therefore, if all $[D_1]_{c,p}$ and all $[D_2]_{c,p}$ are invertible then all $[D]_{c,p}$ are invertible. If not all $[D_1]_{c,p}$ are invertible, then by induction hypothesis at least two of them are not, and consequently at least two $[D]_{c,p}$ are not invertible. If not all $[D_2]_{c,p}$ are invertible, then by induction hypothesis at least two of them are not; since $\tau_{D_1}$ is surjective this implies that at least two $[D_2]_{\tau_{D_1}(c,p)}$ are not invertible and consequently that at least two $[D]_{c,p}$ are not invertible. In all three cases, the result holds.

If $D=D_1\otimes D_2$, then the set of all $[D]_{c,p}$ is the union of the set of all $[D_1]_{c,p}$ and the set of all $[D_2]_{c,p}$. Therefore, if all $[D_1]_{c,p}$ and $[D_2]_{c,p}$ are invertible then all $[D]_{c,p}$ are, and if not all are invertible, then by induction hypothesis at least two $[D_1]_{c,p}$ or two $[D_2]_{c,p}$ are not invertible, so that at least two $[D]_{c,p}$ are not invertible. In both cases the result holds.

\subsubsection{Proof of Lemma \ref{proddetscarresU}}\label{preuveproddetscarresU}
We proceed by structural induction on $D$.

If $D=\begin{tikzpicture}[scale=0.3]
	\begin{pgfonlayer}{nodelayer}
		\node [style=xsdiagrammevide] (0) at (0, 0) {};
	\end{pgfonlayer}
\end{tikzpicture}
,\begin{tikzpicture}
	\begin{pgfonlayer}{nodelayer}
		\node [style=none] (0) at (-0.3, 0) {};
		\node [style=none] (1) at (0.3, 0) {};
	\end{pgfonlayer}
	\begin{pgfonlayer}{edgelayer}
		\draw [style=noire] (0.center) to (1.center);
	\end{pgfonlayer}
\end{tikzpicture},\begin{tikzpicture}[scale=0.8]
	\begin{pgfonlayer}{nodelayer}
		\node [style=cartouche] (0) at (0, 0) {\tiny$\neg$};
		\node [style=none] (1) at (-0.8, 0) {};
		\node [style=none] (2) at (0.8, 0) {};
	\end{pgfonlayer}
	\begin{pgfonlayer}{edgelayer}
		\draw [style=noire] (1.center) to (0);
		\draw [style=noire] (0) to (2.center);
	\end{pgfonlayer}
\end{tikzpicture}
,\input{swap-xs.tikz}\text{ or }\input{beamsplitter-xs.tikz}$, then $D$ does not contain any gate and for any $(c,p)$ we have $\det([D]_{c,p})=1$, so with the usual convention that the empty product is equal to $1$, the result holds.

If $D=\begin{tikzpicture}[scale=0.6]
	\begin{pgfonlayer}{nodelayer}
		\node [style=newe] (0) at (0, 0) {\tiny\!$U$\!};
		\node [style=none] (1) at (-1, 0) {};
		\node [style=none] (2) at (1, 0) {};
	\end{pgfonlayer}
	\begin{pgfonlayer}{edgelayer}
		\draw [style=new] (1.center) to (2.center);
	\end{pgfonlayer}
\end{tikzpicture}$, then we have $|D|=\displaystyle\prod_{c\in\hv}\det(U)=\det(U)^2$, and $\begin{tikzpicture}[scale=0.6]
	\begin{pgfonlayer}{nodelayer}
		\node [style=newe] (0) at (0, 0) {\tiny\!$U$\!};
		\node [style=none] (1) at (-1, 0) {};
		\node [style=none] (2) at (1, 0) {};
	\end{pgfonlayer}
	\begin{pgfonlayer}{edgelayer}
		\draw [style=new] (1.center) to (2.center);
	\end{pgfonlayer}
\end{tikzpicture}$ is the only gate in $D$, so the result holds.

If $D=D_2\circ D_1$, then on the one hand the set of gates of $D$ is the disjoint union of the set of gates of $D_1$ and this of $D_2$, so \[\displaystyle\prod_{G\text{ gate in }D}\det\left(U(G)\right)^2=\left(\displaystyle\prod_{G\text{ gate in }D_1}\det\left(U(G)\right)^2\right)\left(\displaystyle\prod_{G\text{ gate in }D_2}\det\left(U(G)\right)^2\right),\] which by induction hypothesis is equal to $|D_1||D_2|$; on the other hand we have
\[\begin{array}{rcl}
|D|&=&\displaystyle\prod_{c\in\hv,p\in[n]}\det\left([D]_{c,p}\right)\\
&=&\displaystyle\prod_{c\in\hv,p\in[n]}\det\left([D_2]_{\tau_{D_1}(c,p)}[D_1]_{c,p}\right)\\
&=&\displaystyle\prod_{c\in\hv,p\in[n]}\det\left([D_2]_{\tau_{D_1}(c,p)}\right)\det\left([D_1]_{c,p}\right)\\
&=&\left(\displaystyle\prod_{c\in\hv,p\in[n]}\det\left([D_2]_{\tau_{D_1}(c,p)}\right)\right)\left(\displaystyle\prod_{c\in\hv,p\in[n]}\det\left([D_1]_{c,p}\right)\right)\\
&=&\left(\displaystyle\prod_{c\in\hv,p\in[n]}\det\left([D_2]_{c,p}\right)\right)\left(\displaystyle\prod_{c\in\hv,p\in[n]}\det\left([D_1]_{c,p}\right)\right)\\
&=&|D_1||D_2|
\end{array}\]
which proves the result for $D$.

If $D=D_1\otimes D_2$, then on the one hand the set of gates of $D$ is the disjoint union of the set of gates of $D_1$ and this of $D_2$, so \[\displaystyle\prod_{G\text{ gate in }D}\det\left(U(G)\right)^2=\left(\displaystyle\prod_{G\text{ gate in }D_1}\det\left(U(G)\right)^2\right)\left(\displaystyle\prod_{G\text{ gate in }D_2}\det\left(U(G)\right)^2\right),\] which by induction hypothesis is equal to $|D_1||D_2|$; on the other hand the set of the $[D]_{c,p}$ is the disjoint union of the set of the $[D_1]_{c,p}$ and  the set of the $[D_2]_{c,p}$, so $|D|=|D_1||D_2|$. This proves the result for $D$.

\section{Example: controlled permutation}\label{diagrammecontrolepermutations}
Let $n\geq 4$ and let $U_1,...,U_n$ be transformations. The parallel composition of all diagrams of the following form, with $\sigma$ a permutation such that $\sigma(n-3)<\sigma(n-2)$ and $\sigma(n-1)<\sigma(n)$, is a $\frac{n!}2\to\frac{n!}2$ diagram, with $\frac{n!}2$ occurrences of each gate, that implements a controlled permutation of the $U_i$:
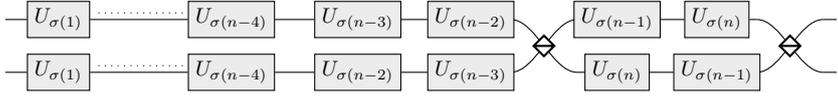
\[\input{motifpermutation.tikz}.\]

\section{Derivation of double-negation elimination from the axioms of the \textup{PBS}-calculus}\label{derivationnegneg}
\begin{lemma}
The following equation is a consequence of the axioms of the \textup{PBS}-calculus:
\setcounter{equation}{10}
\begin{equation}\labeletpreuve{negneg}{\tikzset{tikzfig/.style={baseline=-0.25em,scale=\echellefils,every node/.style={scale=0.8}}}\begin{array}{rcl}\input{negneg.tikz}&=&\begin{tikzpicture}
	\begin{pgfonlayer}{nodelayer}
		\node [style=none] (0) at (-2, 0) {};
		\node [style=none] (1) at (2, 0) {};
	\end{pgfonlayer}
	\begin{pgfonlayer}{edgelayer}
		\draw [style=new] (0.center) to (1.center);
	\end{pgfonlayer}
\end{tikzpicture}\end{array}}\end{equation}
\end{lemma}
\begin{proof}\vspace{-0.16cm}
To prove this equation, we have:

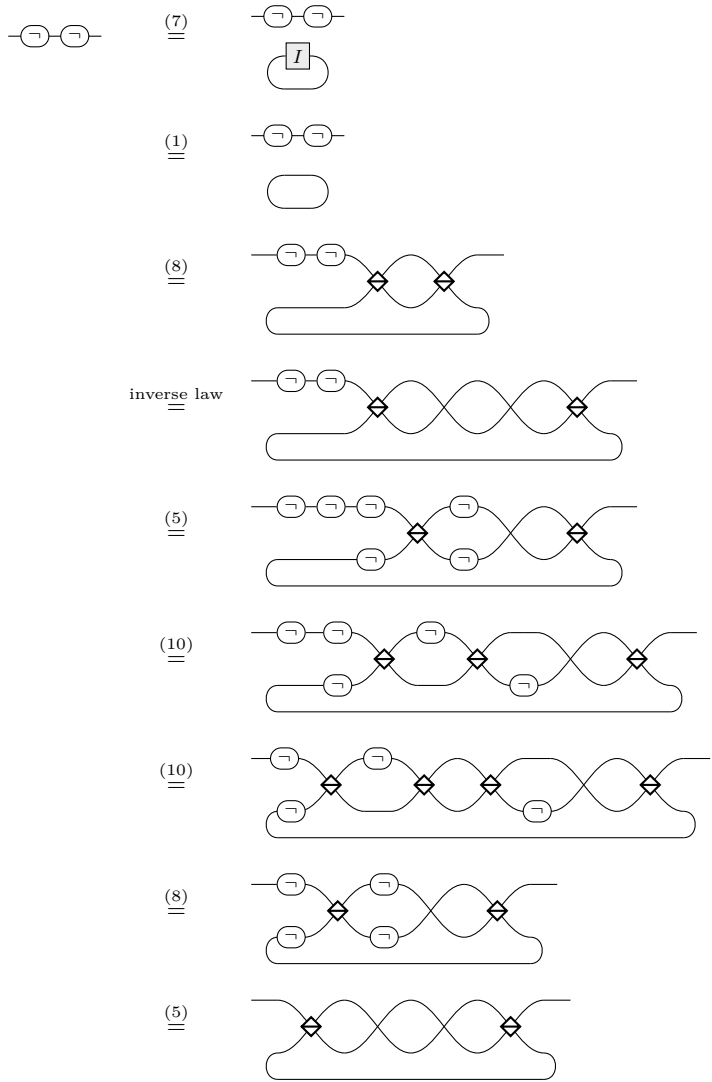
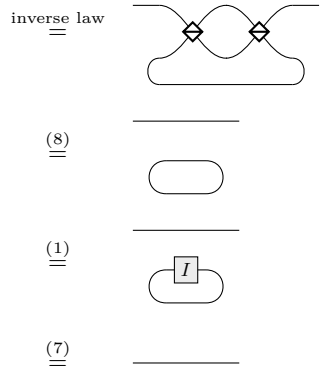
\tikzset{tikzfig/.style={baseline=-0.25em,scale=\echellefils,every node/.style={scale=0.8}}}
\begin{longtable}{RCL}\label{preuvenegneg}\input{negneg.tikz}&\eqeqref{loopemptysimple}&\input{negnegsurbouclevideI.tikz}\\\\
&\eqeqref{idbox}&\input{negnegsurbouclevide.tikz}\\\\
&\eqeqref{bsbs}&\input{tracenhnhbsbs.tikz}\\\\
&\eqexpl{\text{inverse law}}&\input{tracenhnhbsxxbs.tikz}\\\\
&\eqeqref{bsnnnn}&\input{tracenhnhnhnbbsnhnbxbs.tikz}\\\\
&\eqeqref{bsnbsh}&\input{tracenhnhnbbsnhbsnbxbs.tikz}\\\\
&\eqeqref{bsnbsh}&\input{tracenhnbbsnhbsbsnbxbs.tikz}\\\\
&\eqeqref{bsbs}&\input{tracenhnbbsnhnbxbs.tikz}\\\\
&\eqeqref{bsnnnn}&\input{tracebsxxbs.tikz}\\\\
&\eqexpl{\text{inverse law}}&\input{filtrerienrienetrien.tikz}\\\\
&\eqeqref{bsbs}&\input{filsurbouclevide.tikz}\\\\
&\eqeqref{idbox}&\input{filsurbouclevideI.tikz}\\\\
&\eqeqref{loopemptysimple}&\begin{tikzpicture}
	\begin{pgfonlayer}{nodelayer}
		\node [style=none] (0) at (-2, 0) {};
		\node [style=none] (1) at (2, 0) {};
	\end{pgfonlayer}
	\begin{pgfonlayer}{edgelayer}
		\draw [style=new] (0.center) to (1.center);
	\end{pgfonlayer}
\end{tikzpicture}
\end{longtable}

\end{proof}

\section{Proof of equivalence between the two diagrams of Figure \ref{fig:perm3} using the PBS-calculus}\label{equivperm3}
We need the following two ancillary equations:
\setcounter{equation}{46}
\begin{lemma}
The following equations are consequences of the axioms of the \textup{PBS}-calculus:
\begin{equation}\label{rebondisnegtraverse}
\begin{array}{rcl}
\input{bouclerebondisUneg.tikz}&=&\input{negboucletraverseU.tikz}
\end{array}
\end{equation}
\begin{equation}\label{traversenegrebondis}
\begin{array}{rcl}
\hspace{0.875cm}\input{boucletraverseUneg.tikz}&=&\input{negbouclerebondisU.tikz}
\end{array}
\end{equation}
\end{lemma}
\begin{proof}\vspace{-0.16cm}
To prove Equation \eqref{rebondisnegtraverse}, we have:
\begin{longtable}{RCL}
\input{bouclerebondisUneg.tikz}&\eqdeuxeqref{bsnnnn}{negneg}&\input{tracenhnbbsnbUb.tikz}\\\\
&\overset{\text{dinaturality,}}{\eqexpl{\eqref{negu},\eqref{negneg}}}&\input{negboucletraverseU.tikz}
\end{longtable}
To prove Equation \eqref{traversenegrebondis}, we have:
\begin{longtable}{RCL}
\input{boucletraverseUneg.tikz}&\overset{\eqref{negneg},\eqref{negu}}{\eqexpl{\text{dinaturality}}}&\input{negtracenhnbbsnhnbUb.tikz}\\\\
&\eqeqref{bsnnnn}&\input{negbouclerebondisU.tikz}
\end{longtable}
\end{proof}
We have to transform the following diagram into the other one of Figure \ref{fig:perm3}:
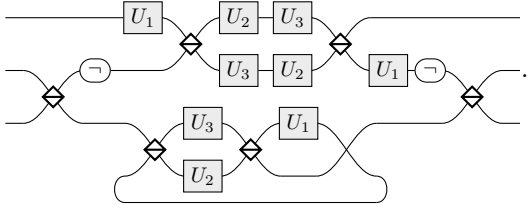
\[\input{perm3var.tikz}.\]
First, we transform each gate into two loops using Equations \eqref{splitu} and \eqref{splituv}:
\begingroup
\tikzset{tikzfig/.style={baseline=-0.25em,scale=\echellefils,every node/.style={scale=0.8}}}
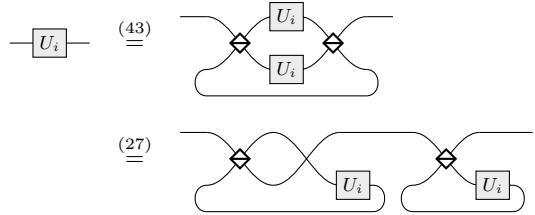
\begin{longtable}{RCL}
\begin{tikzpicture}
	\begin{pgfonlayer}{nodelayer}
		\node [style=newe] (0) at (0, 0) {$U_i$};
		\node [style=none] (1) at (-1.5, 0) {};
		\node [style=none] (2) at (1.5, 0) {};
	\end{pgfonlayer}
	\begin{pgfonlayer}{edgelayer}
		\draw [style=new] (1.center) to (2.center);
	\end{pgfonlayer}
\end{tikzpicture}
&\eqeqref{splitu}&\input{filtreUiUi.tikz}\\\\
&\eqeqref{splituv}&\input{bouclerebondisUiboucletraverseUi.tikz}
\end{longtable}
\endgroup
then we slide all loops to the right using \cref{traversebsbh,,traversebshb,,rebondisbsb,,rebondisbsh,,traversenegrebondis,,rebondisnegtraverse}. We get:
\[\input{perm3varboucles.tikz}.\]
Next, we transform the left part:
\begin{longtable}{RCL}
\input{perm3varsansgates.tikz}&\overset{\eqref{bsbs},\eqref{negneg},}{\eqexpl{\text{yanking}}}&\input{troisfilsparallelestreslongs.tikz}\\
&\overset{\text{inverse law,}}{\overset{\eqref{bsbs},}{\overset{\text{naturality of}}{\eqexpl{\text{the swap}}}}}&\input{xbbshxbxhbsbxh.tikz}\\\\
&\eqdeuxeqref{bsbsbs}{bsbsbsvar}&\input{bsbbshbsbbshbsbbsh.tikz}.
\end{longtable}
Finally, again using again \cref{traversebsbh,traversebshb,rebondisbsb,rebondisbsh,traversenegrebondis,rebondisnegtraverse}, then \eqref{splituv} and \eqref{splitu}, we slide the loops into the diagram and merge them two by two to get the desired diagram:
\[\input{perm3.tikz}.\]

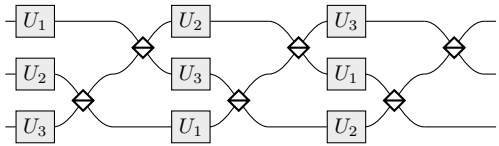
\end{document}